\documentclass[11pt,a4paper]{article}
\usepackage{amsfonts,amssymb,amsmath,amsthm,cite}
\usepackage{graphicx}
\evensidemargin=\oddsidemargin
\DeclareMathOperator{\Dom}{Dom}          
\DeclareMathOperator{\Ima}{Im}           
\DeclareMathOperator{\Ker}{Ker}          
\DeclareMathOperator{\Res}{Res}          
\DeclareMathOperator{\Tr}{Tr}            

\newtheorem{assumption}{Assumption}[section]
\newtheorem{theorem}[assumption]{Theorem}
\newtheorem{corollary}[assumption]{Corollary}
\newtheorem{conjecture}[assumption]{Conjecture}
\newtheorem{lemma}[assumption]{Lemma}
\newtheorem{definition}[assumption]{Definition}
\newtheorem{prop}[assumption]{Proposition}
\newtheorem{remark}[assumption]{Remark}

\newcommand{\Th}{\Theta}
\renewcommand{\th}{\theta}
\newcommand{\A}{\mathcal{A}}             
\renewcommand{\a}{\alpha}                
\newcommand{\B}{\mathcal{B}}             
\newcommand{\C}{\mathbb{C}}              
\newcommand{\Coo}{C^\infty}              
\newcommand{\DD}{\mathcal{D}}            
\newcommand{\eps}{\varepsilon}           
\newcommand{\Ga}{\Gamma}                 
\newcommand{\ga}{\gamma}                 
\renewcommand{\H}{\mathcal{H}}           
\newcommand{\half}{{\mathchoice{\thalf}{\thalf}{\shalf}{\shalf}}}
\newcommand{\hideqed}{\renewcommand{\qed}{}} 
\renewcommand{\L}{\mathcal{L}}           
\newcommand{\la}{\lambda}                
\newcommand{\N}{\mathbb{N}}              

\newcommand{\Q}{\mathbb{Q}}              
\newcommand{\R}{\mathbb{R}}              

\newcommand{\set}[1]{\{\,#1\,\}}         
\newcommand{\shalf}{{\scriptstyle\frac{1}{2}}}  
\renewcommand{\SS}{\mathcal{S}}          
\newcommand{\T}{\mathbb{T}}              
\newcommand{\thalf}{\tfrac{1}{2}}        
\newcommand{\wh}{\widehat}               
\newcommand{\wt}{\widetilde}             
\newcommand{\Z}{\mathbb{Z}}              

\def\<#1,#2>{\langle#1\,,\,#2\rangle}        
\newcommand{\norm}[1]{\left\lVert#1\right\rVert}   

\newbox\ncintdbox \newbox\ncinttbox
\setbox0=\hbox{$-$} \setbox2=\hbox{$\displaystyle\int$}
\setbox\ncintdbox=\hbox{\rlap{\hbox to \wd2{\kern-.1em\box2\relax\hfil}}\box0\kern.1em}
\setbox0=\hbox{$\vcenter{\hrule width 4pt}$}
\setbox2=\hbox{$\textstyle\int$}
\setbox\ncinttbox=\hbox{\rlap{\hbox
    to \wd2{\kern-.14em\box2\relax\hfil}}\box0\kern.1em}
\newcommand{\ncint}{\mathop{\mathchoice{\copy\ncintdbox}
    {\copy\ncinttbox}{\copy\ncinttbox}
    {\copy\ncinttbox}}\nolimits}

\begin{document}

\thispagestyle{empty}

\begin{center}

CENTRE DE PHYSIQUE TH\'EORIQUE$\,^1$\\
CNRS--Luminy, Case 907\\
13288 Marseille Cedex 9\\
FRANCE\\

\vspace{3cm}

{\Large\textbf{Spectral action on noncommutative torus}} \\
\vspace{0.5cm}

{\large D. Essouabri$^{2}$, B. Iochum$^{1, 3}$, C. Levy$^{1,3}$
and A. Sitarz$^{4,5}$}

\vspace{1cm}
{\it Dedicated to Alain Connes on the occasion of his 60th birthday}

\vspace{1.5cm}

{\large\textbf{Abstract}}
\end{center}

\begin{quote}
The spectral action on noncommutative torus is obtained, using a
Chamseddine--Connes formula via computations of zeta functions.
The importance of a Diophantine condition is outlined. Several
results on holomorphic continuation of series of holomorphic
functions are obtained in this context.
\end{quote}

\vspace{1cm}

February 2007

\vspace{1cm}



CPT-P06-2007

\vspace{1cm}

\noindent $^1$ UMR 6207

-- Unit\'e Mixte de Recherche du CNRS et des
Universit\'es Aix-Marseille I, Aix-Marseille II et de l'Universit\'e
du Sud Toulon-Var

-- Laboratoire affili\'e \`a la FRUMAM -- FR 2291\\
$^2$ Universit\'e de Caen (Campus II),
 Laboratoire de Math. Nicolas Oresme (CNRS UMR 6139),
B.P. 5186, 14032 Caen, France, essoua@math.unicaen.fr \\
$^3$ Also at Universit\'e de Provence,
iochum@cpt.univ-mrs.fr, levy@cpt.univ-mrs.fr\\
$^4$ Institute of Physics, Jagiellonian University,
Reymonta 4, 30-059 Krak\'ow, Poland\\
sitarz@if.uj.edu.pl \\
$^5$ Partially supported by MNII Grant 115/E-343/SPB/6.PR UE/DIE
50/2005--2008
\newpage

\section{Introduction}

The spectral action introduced by
Chamseddine--Connes plays an important role \cite{CC} in
noncommutative geometry. More
precisely, given a spectral triple $(\A,\H,\DD)$ where $\A$ is an
algebra acting on the Hilbert space $\H$ and $\DD$ is a Dirac-like
operator (see \cite{Book,Polaris}), they proposed a physical action
depending only on the spectrum of the covariant Dirac operator
\begin{equation}
\label{covDirac} \DD_{A}:=\DD + A + \epsilon \,JAJ^{-1}
\end{equation}
where $A$ is a one-form represented on $\H$,
so has the decomposition
\begin{equation}
\label{oneform}
A=\sum_{i}a_{i}[\DD,\,b_{i}],
\end{equation}
with $a_{i}$, $b_{i}\in \A$, $J$ is a real
structure on the triple corresponding to charge conjugation and
$\epsilon \in \set{1,-1}$ depending on the dimension of this triple
and comes from the commutation relation
\begin{equation}
    \label{Jcom}
J\DD=\epsilon \, \DD J.
\end{equation}

This action is defined by
\begin{equation}
\label{action}
\SS(\DD_{A},\Phi,\Lambda):=\Tr \big( \Phi( \DD_{A} /\Lambda) \big)
\end{equation}
where $\Phi$ is any even positive cut-off function which could
be replaced by a step function up to some mathematical difficulties
investigated in \cite{Odysseus}. This means that $\SS$ counts the
spectral values of $\vert \DD_{A} \vert$ less than the mass scale
$\Lambda$ (note that the resolvent of $\DD_{A}$ is compact since, by
assumption, the same is true for $\DD$, see Lemma \ref{compres}
below).

In \cite{GI2002}, the spectral action on NC-tori has been computed
only for operators of the form $\DD + A$ and computed for $\DD_{A}$
in \cite{GIVas}.  It appears that the implementation of
the real structure via $J$, does change the spectral action, up to a
coefficient when the torus has dimension 4.
Here we prove that this can be also directly obtained from
the Chamseddine--Connes analysis of \cite{CC1} that we follow quite
closely. Actually,
\begin{align}
\label{formuleaction}
    \SS(\DD_{A},\Phi,\Lambda) \, = \,\sum_{0<k\in Sd^+} \Phi_{k}\,
    \Lambda^{k} \ncint \vert D_{A}\vert^{-k} + \Phi(0) \,
    \zeta_{D_{A}}(0) +\mathcal{O}(\Lambda^{-1})
\end{align}
where $D_A = \DD_A + P_A$, $P_A$ the projection on $\Ker \DD_A$,
$\Phi_{k}= \half\int_{0}^{\infty} \Phi(t) \, t^{k/2-1} \, dt$ and
$Sd^+$ is the strictly positive part of the dimension
spectrum of $(\A,\H,\DD)$. As we will see,
$Sd^+=\set{1,2,\cdots,n}$ and $\ncint |D_{A}|^{-n}=\ncint |D|^{-n}$.
Moreover, the coefficient
$\zeta_{D_A}(0)$ related to the constant term in
(\ref{formuleaction}) can
be computed from the unperturbed spectral action since it has been
proved in \cite{CC1} (with
an invertible Dirac operator and a 1-form $A$ such that $\DD+A$ is
also invertible) that
\begin{align} \label{constant}
\zeta_{\DD+A}(0)-\zeta_{\DD}(0)= \sum_{q=1}^{n}\tfrac{(-1)^{q}}{q}
\ncint (A\DD^{-1})^{q},
\end{align}
using $\zeta_X(s)=\Tr(|X|^{-s})$. We will see how this formula can be
extended to
the case a noninvertible Dirac operator and noninvertible
perturbation of the form $\DD+\wt A$ where 
$\wt A:=A+\eps JAJ^{-1}$.

All this results on spectral action are quite important in physics,
especially in quantum field theory and particle physics, where one
adds to the effective action some counterterms explicitly given by
(\ref{constant}), see for instance
\cite{Carminati,CC,CC1,CCM,Gayral,Goursac,GI2002,GIVas,
Knecht,Strelchenko,Vassilevich2002,Vassilevich2005,Vassilevich2007}.

Since the computation of zeta functions is crucial here, we
investigate in section 2 residues of series and integrals.
This section contains independent interesting results on the
holomorphy of series of holomorphic functions. In particular,
the necessity of a Diophantine constraint is naturally emphasized.

In section 3, we revisit the notions of pseudodifferential
operators and their associated zeta functions and of dimension
spectrum. The reality operator $J$ is
incorporated and we pay a particular attention to kernels of
operators which can play a role in the constant term of
(\ref{formuleaction}). This section concerns general spectral triple
with simple dimension spectrum.

Section 4 is devoted to the example of the noncommutative torus.
It is shown that it has a vanishing tadpole.

In section 5, all previous technical points are then widely used for
the computation of terms in (\ref{formuleaction}) or
(\ref{constant}).

Finally, the spectral action (\ref{constant}) is obtained in section 6
and we conjecture that the noncommutative spectral action
of $\DD_{A}$ has terms proportional to the spectral action
of $\DD+A$ on the commutative torus.

\section{Residues of series and integral, holomorphic continuation,
etc}
Notations:

In the following, the prime in ${\sum}'$ means that we omit terms
with division by zero in the summand.
$B^{n}$ (resp. $S^{n-1}$) is the closed ball (resp.
the sphere) of $\R^n$ with center $0$ and radius 1 and the
Lebesgue measure on $S^{n-1}$ will be noted $dS$.

For any $x=(x_1,\dots,x_n) \in \R^n$ we denote by
$|x|=\sqrt{x_1^2+\dots+x_n^2}$ the
euclidean norm and $|x|_1 :=|x_1|+\dots +|x_n|$.

$\N =\{1, 2,\dots \}$ is the set of positive integers and $\N_0 =\N
\cup \{0\}$ the set of non negative integers.

By  $f(x,y) \ll_{y} g(x)$
uniformly in $x$, we mean that $\vert f(x,y)\vert \leq a(y)
\, \vert g(x) \vert$ for all $x$ and $y$ for some $a(y)>0$.

\subsection{Residues of series and integral}
In order to be able to compute later the residues of certain series,
we prove here the following

\begin{theorem}\label{res-int} Let $P(X)=\sum_{j=0}^{d} P_j(X)
 \in \C[X_1,\cdots,X_n]$ be
a polynomial function where $P_j$
is the homogeneous part of $P$ of degree $j$. The function
$$ \zeta^P(s):={\sum}'_{k\in\Z^n} \tfrac{P(k)}{|k|^s}, \,\,\, s \in \C
$$ has a meromorphic continuation to the whole complex plane
$\C$. \par Moreover $\zeta^P(s)$ is not entire if and only if
$\mathcal{P}_P:= \{j \  : \
\int_{u\in S^{n-1}} P_j(u)\, dS(u)\neq 0 \}\neq \varnothing$.
In that case, $\zeta^P$ has only simple poles at the points $j+n$,
$j\in \mathcal{P}_{P}$, with
$$
\underset{s=j+n}{\Res} \, \zeta^P(s) = \int_{u\in S^{n-1}} P_j(u)\,
dS(u).
$$
\end{theorem}
\medskip

The proof of this theorem is based on the following lemmas.

\begin{lemma}\label{majPQs}
For any polynomial $P\in \C[X_1,\dots,X_n]$ of total
degree $\delta (P):=\sum_{i=1}^n deg_{X_i}P$ and any $\alpha \in
\N_0^n$, we have
$$\partial^{\alpha} \left(P(x) |x|^{-s}\right)\ll_{P, \alpha, n}
  (1+|s|)^{|\alpha|_1} |x|^{-\sigma -|\alpha|_1 +\delta (P)}$$
uniformly  in  $x \in \R^n$ verifying $|x|\geq 1$, where
$\sigma=\Re(s)$.
\end{lemma}
\begin{proof}
By linearity, we may assume without loss of generality
that $P(X)=X^{\gamma}$ is a monomial. It is easy to
prove (for example by induction on $|\alpha|_1$) that for all $\alpha
\in \N_0^n$ and
$x \in \R^n \setminus \{0\}$:
$$
\partial^{\alpha} \left(|x|^{-s}\right)
=\alpha! \sum_{\genfrac{}{}{0pt}{2}{\beta, \mu \in \N_0^n}{\beta + 2
\mu =
\alpha }}
\genfrac(){0pt}{1}{-s/2}{|\beta|_1 +|\mu|_1}
\tfrac{(|\beta|_1+|\mu|_1)!}{\beta! ~ \mu!}
\tfrac{x^{\beta}}{|x|^{\sigma+2(|\beta|_1+|\mu|_1)}}.
$$
It follows that for all $\alpha \in \N_0^n$, we have uniformly in  $x
\in \R^n$
verifying $|x|\geq 1$:
\begin{equation}\label{majQs}
\partial^{\alpha} \left(|x|^{-s}\right) \ll_{\alpha,  n}
  (1+|s|)^{|\alpha|_1} |x|^{-\sigma -|\alpha|_1}\,.
\end{equation}

By Leibniz formula and (\ref{majQs}), we have uniformly in $x\in \R^n$
verifying $|x|\geq 1$:
\begin{align*}
\partial^{\alpha} \left(x^\gamma |x|^{-s}\right) & = \, \sum_{\beta
\leq \alpha}
\genfrac(){0pt}{1}{\alpha}{\beta} \,  \partial^{\beta} (x^\gamma)~
\partial^{\alpha
  -\beta} \left(|x|^{-s}\right) \\
& \ll_{\gamma, \alpha, n}  \sum_{\beta \leq \alpha; \beta \leq
\gamma} x^{\gamma
-\beta}~ (1+|s|)^{|\alpha|_1-|\beta|_1}~
|x|^{-\sigma-|\alpha|_1+|\beta|_1} \\
& \ll_{\gamma, \alpha, n}  (1+|s|)^{|\alpha|_1}~
|x|^{-\sigma-|\alpha|_1+|\gamma|_1}.
\tag*{\qed}
\end{align*}
\hideqed
\end{proof}

\begin{lemma}\label{deltaf} Let $P\in \C[X_1,\dots,X_n]$ be a
  polynomial of degree $d$.
Then, the difference
$$
\Delta_P(s):={\sum}'_{k\in\Z^n}
\tfrac{P(k)}{|k|^s}-\int_{\R^n\setminus B^{n}}
\tfrac{P(x)}{|x|^s} \, dx
$$
which is defined for $\Re(s)>d+n$, extends holomorphically on the
whole complex plane
$\C$.
\end{lemma}
\begin{proof}
We fix in the sequel a function $\psi\in C^\infty(\R^n ,\R)$
verifying for all $x\in
\R^n$
$$
0\leq \psi(x) \leq 1, \quad \psi(x)=1 \text{ if }|x|\geq 1
\quad \text{and} \quad \psi(x)=0 \text{ if } |x|\leq 1/2.
$$
The function $f(x,s) :=  \psi (x)~P(x)~ |x|^{-s}$, $x\in \R^n$ and
$s\in \C$,
is in ${\cal
  C}^\infty (\R^n \times \C)$  and depends holomorphically on $s$.

Lemma \ref{majPQs} above shows that $f$ is
a ``gauged symbol'' in the terminology of \cite[p. 4]{GSW}.
Thus \cite[Theorem 2.1]{GSW} implies that $\Delta_P(s)$ extends
holomorphically on the whole complex plane $\C$. However, to be
complete, we will give here a short proof of Lemma \ref{deltaf}:
\par It follows from the classical Euler--Maclaurin formula that for
any
function $h: \R \rightarrow \C$ of class ${\cal C}^{N+1}$
verifying $ \lim_{|t|\rightarrow +\infty} h^{(k)}(t)=0$ and $\int_{\R}
|h^{(k)} (t)| ~dt <+\infty$ for any
$k=0 \dots,N+1$, that we have
$$
\sum_{k\in \Z} h(k) = \int_{\R} h(t) + \tfrac{(-1)^N}{(N+1)!}
\int_{\R} B_{N+1}(t)~h^{(N+1)}(t) ~dt
$$
where $B_{N+1}$ is the Bernoulli function of
order $N+1$ (it is a bounded periodic function.)
\par Fix $m' \in \Z^{n-1}$ and $s\in
\C$. Applying this to the function $h(t):= \psi (m',t)~P(m',t)
~|(m',t)|^{-s}$ (we use Lemma \ref{majPQs} to verify hypothesis),
 we obtain that for any $N\in \N_0$:
\begin{equation}\label{*1}
\sum_{m_n \in \Z} \psi (m',m_n) ~P(m',m_n) ~|(m',m_n)|^{-s}
= \int_{\R} \psi(m',t)
~P(m',t) ~|(m',t)|^{-s} ~dt +{\cal R}_N(m';s)
\end{equation}
where $ {\cal R}_N(m';s):=\tfrac{(-1)^N}{(N+1)!} \int_{\R} B_{N+1}(t)~
\tfrac{\partial^{N+1}}{{\partial x_n}^{N+1}} \left(\psi(m',t)~P(m',t)
~|(m',t)|^{-s}\right)~dt$.\\
By Lemma \ref{majPQs},
$$
\int_{\R} {\Big |} B_{N+1}(t)~
\tfrac{\partial^{N+1}}{{\partial x_n}^{N+1}} \left(\psi (m',t)
~P(m',t)~|(m',t)|^{-s}\right){\Big |}~dt \ll_{P,n, N}
(1+|s|)^{N+1}~ (|m'|+1)^{-\sigma
-N+ \delta(P)}.
$$
Thus $ \sum_{m' \in \Z^{n-1}} {\cal R}_N(m';s)$
converges absolutely and define a holomorphic function in
the half plane $\{\sigma
=\Re (s) > \delta(P)+n-N\}$. \par Since $N$ is an arbitrary integer,
by letting
$N\rightarrow \infty$ and using $(\ref{*1})$ above, we conclude that:
$$s\mapsto \sum_{(m',m_n) \in \Z^{n-1}\times \Z} \psi (m',m_n)
~P(m',m_n)
~|(m',m_n)|^{-s}-\sum_{m'
  \in \Z^{n-1}} \int_{\R} \psi (m',t) ~P(m',t)~|(m',t)|^{-s}~dt$$ has
a holomorphic continuation to the whole complex plane $\C$.\par
After $n$ iterations, we obtain that
$$s\mapsto {\sum}_{m\in \Z^{n}} \psi(m)~P(m)
~|m|^{-s}-\int_{\R^n} \psi(x)~P(x) ~|x|^{-s}~dx$$ has a
holomorphic continuation to the whole $\C$.\\
To finish the proof of Lemma \ref{deltaf}, it is enough to notice
that:

\hspace{1 cm} $\bullet$ $\psi(0)=0$ and  $\psi (m)=1$, $\forall m\in
\Z^n\setminus \{0\}$;

\hspace{1 cm} $\bullet$ $s\mapsto \int_{B^n} \psi(x)~P(x)~|x|^{-s}~dx
= \int_{\{x\in \R^n : 1/2\leq |x|\leq 1\}} \psi(x)~P(x)~|x|^{-s}~dx$
is a holomorphic
function on $\C$.
\end{proof}

\begin{proof}[Proof of Theorem \ref{res-int}]

Using the polar decomposition of the volume form
$dx=\rho^{n-1}\,d\rho\, dS$ in $\R^n$, we get for $\Re(s)>d+n$,
$$
\int_{\R^n \setminus B^{n}} \tfrac{P_j(x)}{|x|^s}dx =
\int_{1}^{\infty}
\tfrac{\rho^{j+n-1}}{\rho^s}\int_{S^{n-1}} P_j(u)\, dS(u) =
\tfrac{1}{j+n-s}
\int_{S^{n-1}} P_j(u)\, dS(u).
$$
Lemma \ref{deltaf} now gives the result.
\end{proof}

\subsection{Holomorphy of certain series}
Before stating the main result of this section, we give first in the
following some
preliminaries from Diophantine approximation theory:

\begin{definition}\label{ba}
(i) Let $\delta >0$. A vector $a \in \R^n$ is said to be
$\delta-$diophantine
if there exists $c >0$ such that $|q . a -m| \geq c \,|q|^{-\delta}$,
$\forall q \in \Z^n \setminus \set{0}$ and $\forall m \in \Z$. \\
We note ${\cal BV}(\delta )$ the set of $\delta-$diophantine
vectors and ${\cal
BV} :=\cup_{\delta >0} {\cal BV}(\delta)$ the set of diophantine vectors.\par
(ii) A matrix $\Th \in {\cal M}_{n}(\R)$ (real $n \times n$ matrices)
will be
said to be diophantine if there
exists $u \in \Z^n$ such that  ${}^t\Th (u)$ is a diophantine
vector of $\R^n$.
\end{definition}
{\bf Remark.} A classical result from Diophantine approximation
asserts
that for all $\delta >n$, the Lebesgue measure of
$\R^n \setminus {\cal BV}(\delta)$ is zero (i.e almost any element of
$\R^n$ is $\delta-$diophantine.)
\par Let $\Th \in {\cal M}_n(\R)$. If its row
of index $i$ is a diophantine vector of $\R^n$ (i.e. if $L_i
\in {\cal BV}$)
then ${}^t \Th (e_i) \in {\cal BV}$ and thus $\Th$ is a diophantine matrix. It
follows that almost any matrix of ${\cal M}_n(\R)\approx \R^{n^2}$ is
diophantine.\par

\medskip

The goal of this section is to show the following
\begin{theorem}\label{analytic}
Let $P\in \C[X_1,\cdots,X_n]$ be a homogeneous polynomial of degree
$d$ and let $b$ be in
$\mathcal{S}(\Z^{n} \times \dots \times \Z^{n})$ ($q$ times,
$q\in\N$). Then,

(i)
Let $a \in \R^n$. We define $f_a(s):={\sum}'_{k\in \Z^n}
\frac{P(k)}{|k|^s}\,
e^{2\pi i k.a}$.

\quad 1.
If $a\in \Z^n$, then $f_a$ has a meromorphic
continuation to the whole complex plane
$\C$.\\ Moreover if  $S$ is the unit sphere and $dS$
its Lebesgue measure, then
$f_a$ is not entire if and only if
$\int_{u\in S^{n-1}} P(u)\, dS(u)\neq 0$. In that case, $f_a$
has only a simple pole at the point $d+n$, with
$\underset{s=d+n}{\Res} \, f_a(s) = \int_{u\in
  S^{n-1}} P(u)\, dS(u)$.

\quad 2.
If $a\in \R^n\setminus \Z^n$, then $f_a(s)$ extends holomorphically
to the whole
complex plane $\C$.

(ii)
 Suppose that  $\Th \in {\cal M}_{n}(\R)$ is diophantine.
 For any $(\eps_i)_i\in \{-1,0,1\}^{q}$, the function
$$
g(s):={\sum}_{l\in (\Z^n)^{q}} \, b(l) \,f_{\Th\,
\sum_i \eps_i l_i}(s)
$$
extends meromorphically to the whole complex plane  $\C$
with only one possible pole
on $s= d+n$.\par Moreover, if we set
${\cal Z}:=\{l\in(\Z^n)^{q} \ : \ \sum_{i=1}^q \eps_i
l_i= 0\}$ and $V:=\sum_{l\in {\cal Z}} \, b(l)$, then

1. If $V\int_{S^{n-1}} P(u)\, dS(u)\neq 0$, then $s=d+n$ is a
simple pole of $g(s)$ and
$$
\underset{s=d+n}{\Res} \, g(s) = V\,
\int_{u\in
  S^{n-1}} P(u)\, dS(u).
$$

2. If $V\int_{S^{n-1}} P(u)\, dS(u)=0$, then $g(s)$ extends
holomorphically to the
whole complex plane $\C$.

(iii)
Suppose that $\Th \in \mathcal{M}_n(\R)$ is diophantine. For any $(\eps_i)_i\in
\{-1,0,1\}^{q}$, the function
$$
g_0(s):={\sum}_{l\in (\Z^n)^{q}\setminus
  {\cal Z}} \,
b(l)\,f_{\Th\, \sum_{i=1}^q \eps_i l_i}(s)
$$
where ${\cal Z}:=\{l\in(\Z^n)^{q} \ : \
\sum_{i=1}^q \eps_i l_i= 0\}$ extends holomorphically to the whole
complex plane $\C$.
\end{theorem}

{\it Proof of Theorem \ref{analytic}}:
First we remark that

$\hspace{1cm}$ If $a \in \Z^n$ then
$f_a(s)={\sum}'_{k\in \Z^n} \frac{P(k)}{|k|^s}$. So,
the point $(i.1)$ follows from Theorem \ref{res-int};

$\hspace{1cm}$
$ g(s):=\sum_{l\in (\Z^n)^{q}\setminus {\cal Z}} \, b(l) \,f_{\Th\,
  \sum_i \eps_i l_i}(s)
+\left(\sum_{l\in {\cal Z}} \, b(l)\right) {\sum}'_{k\in \Z^n}
\frac{P(k)}{|k|^s}$. Thus, the point $(ii)$ rises

$\hspace{1.1cm}$easily from $(iii)$ and Theorem \ref{res-int}.

So, to complete the proof, it remains to prove the items $(i.2)$ and
$(iii)$.\par
The direct
proof of $(i.2)$ is easy but is not sufficient to deduce $(iii)$ of
which the proof is
more delicate and requires a more precise (i.e. more effective)
version of
$(i.2)$. The next lemma gives such crucial version, but before,
let us give some notations:
$$
{\cal F}:=\{\tfrac{P(X)}{(X_1^2+\dots +X_n^2+1)^{r/2}}
 \, :\,
P(X) \in \C[X_1,\dots, X_n] {\mbox { and }} r \in \N_0\}.
$$
\hspace{1cm}
  We set $g=$deg$(G) =$deg$(P) -r \in \Z$, the degree of
$G=\frac{P(X)}{(X_1^2+\dots +X_n^2+1)^{r/2}}\in
{\cal F}$.

\hspace{1cm} By convention we set deg$(0)=-\infty$.

\begin{lemma}\label{ieffective}
Let $a \in \R^n$. We assume that $d\left(a . u, \Z\right):=
\inf_{m\in \Z} |a . u  -m| >0$ for some $u \in \Z^n$.
For all $G\in {\cal F}$, we define formally,
\begin{align*}
    F_0(G;a;s):={\sum}'_{k\in\Z^n}
\tfrac{G(k)}{|k|^{s}}\, e^{2\pi i \,k . a}  \quad \text{and} \quad
F_1(G;a;s):={\sum}_{k \in \Z^n}
\tfrac{G(k)}{(|k|^2+1)^{s/2}} \,e^{2\pi i \,k . a} .
\end{align*}

Then for all $N\in \N$, all $G\in {\cal F}$ and all $i\in \{0,1\}$,
there exist
positive constants $C_i:=C_i(G,N,u)$, $B_i:=B_i(G,N,u)$ and
$A_i:=A_i(G,N,u)$ such
that $s\mapsto F_i(G;\a;s)$ extends  holomorphically to the half-plane
$\{\Re(s)>-N\}$ and verifies in it:
$$F_i(G;a;s)\leq C_i (1+|s|)^{B_i} \,
\big(d\left(a . u, \Z\right)\big)^{-A_i}.$$
\end{lemma}
\begin{remark} The important point here is that we obtain an explicit
  bound of $F_i(G;\a;s)$ in $\{\Re(s)>-N\}$ which depends on the
  vector $a$ only through $d(a.u,\Z)$, so depends on $u$ and
indirectly on $a$ (in the sequel, $a$ will vary.) In particular the
constants $C_i:=C_i(G,N,u)$,
$B_i=B_i(G,N)$ and $A_i:=A_i(G,N)$ do not depend on the vector $a$
but only on
$u$. {\it This is crucial for the proof of items $(ii)$ and
$(iii)$ of Theorem \ref{analytic}!}
\end{remark}

\subsubsection{Proof of Lemma \ref{ieffective} for $i=1$:}
Let $N\in \N_0$ be a fixed integer, and set $g_0:= n+N+1$.\\
We will prove Lemma \ref{ieffective} by induction on $g=$deg$(G)\in
\Z$. More
precisely, in order to prove case $i=1$, it suffices to
prove that:

\hspace{1cm} Lemma \ref{ieffective} is true for all $G\in {\cal F}$
  verifying deg$(G)\leq -g_0$.

\hspace{1cm} Let $g\in \Z$ with $g\geq -g_0+1$.
If Lemma \ref{ieffective} is true for all $G\in {\cal F}$
  such that deg$(G)\leq g -1$,

\hspace{1cm} then it is also true for all
$G\in {\cal F}$ satisfying deg$(G)= g$.

$\bullet$ {Step 1: Checking Lemma
\ref{ieffective} for
deg$(G) \leq -g_0:= -(n+N+1)$.}\\
Let $G(X)=\frac{P(X)}{(X_1^2+\dots +X_n^2+1)^{r/2}} \in {\cal F}$
verifying deg$(G)\leq -g_0$.
It is easy to see that we have uniformly in  $s=\sigma +i\tau  \in \C$
and in $k \in
\Z^n$:
\begin{align*}
\tfrac{|G(k) \,e^{2\pi i \,k .
a}|}{(|k|^2+1)^{\sigma/2}}=&\tfrac{|P(k)|}{(|k|^2+1)^{(r+\sigma)/2}}\ll_G
\tfrac{1}{(|k|^2+1)^{(r+\sigma-deg(P))/2}}
\ll_G  \tfrac{1}{(|k|^2+1)^{(\sigma-deg(G))/2}}\ll_G
\tfrac{1}{(|k|^2+1)^{(\sigma+g_0)/2}}.
\end{align*}
It follows that $F_1(G;a;s)=\sum_{k \in \Z^n }
\frac{G(k)}{(|k|^2+1)^{s/2}} \,e^{2\pi i
\, k . a}$ converges absolutely and defines a holomorphic function in
the half plane
$\{\sigma >-N\}$. Therefore, we have for any $s\in \{\Re(s) >-N\}$:
$$|F_1(G;a;s)|\ll_G \sum_{k \in \Z^n }
\tfrac{1}{(|k|^2+1)^{(-N+g_0)/2}}\ll_G \sum_{k \in \Z^n }
\tfrac{1}{(|k|^2+1)^{(n+1)/2}}\ll_G 1.$$
Thus, Lemma \ref{ieffective} is true when deg$(G)\leq -g_0$.

$\bullet$ { Step 2: Induction. }\\
Now let $g\in \Z$ satisfying $g\geq -g_0+1$ and suppose that Lemma
\ref{ieffective} is
valid for all $G\in {\cal F}$ verifying deg$(G) \leq g-1$. Let $G\in
{\cal F}$ with deg$(G)=g$. We will prove that  $G$ also
verifies conclusions of Lemma \ref{ieffective}:\\
There exist $P\in \C[X_1,\dots,X_n]$ of degree $d\geq 0$ and
$r\in \N_0$ such that $G(X)=\frac{P(X)}{(X_1^2+\dots
+X_n^2+1)^{r/2}}$ and
$g=$deg$(G)=d-r$.\\
Since $G(k)\ll (|k|^2+1)^{g/2}$ uniformly in $k \in \Z^n$, we deduce
that
$F_1(G;a;s)$ converges absolutely in $\{\sigma=\Re(s)>n+g\}$.\\
Since $k\mapsto k+u$ is a bijection from  $\Z^n $ into $\Z^n$, it
follows
that we also have for $\Re(s)>n+g$
\begin{align*}
F_1(G;a;s)&=\sum_{k \in \Z^n }
\tfrac{P(k)}{(|k|^2+1)^{(s+r)/2}} \,e^{2\pi i \,k . a}
= \sum_{k \in \Z^n }
\tfrac{P(k+u)}{(|k+u|^2+1)^{(s+r)/2}} \,  e^{2\pi i \, (k+u) .  a}\\
&= e^{2\pi i \,u . a} \sum_{k \in \Z^n } \tfrac{P(k+u)}{(|k|^2+2 k . u
+|u|^2+1)^{(s+r)/2}} \,e^{2\pi i \,k . a}\\
&= e^{2\pi i \,u . a} \sum_{\alpha \in \N_0^n;
|\alpha|_1=\alpha_1+\dots+\alpha_n \leq
d} \tfrac{u^{\alpha}}{\alpha !} \sum_{k \in \Z^n }
\tfrac{\partial^{\alpha}
P(k)}{(|k|^2+2 k . u +|u|^2+1)^{(s+r)/2}} \, e^{2\pi i \,k . a}\\
&= e^{2\pi i \,u . a} \sum_{|\alpha|_1\leq d}
\tfrac{u^{\alpha}}{\alpha !} \sum_{k \in
\Z^n } \tfrac{\partial^{\alpha} P(k)}{(|k|^2+1)^{(s+r)/2}}
\big(1+\tfrac{2 k . u
+|u|^2}{(|k|^2+1)}\big)^{-(s+r)/2} \, e^{2\pi i \,k . a}.
\end{align*}
Let $M:= \sup(N+n+g, 0)\in \N_0$. We have uniformly in $k \in \Z^n$
$$
\big(1+\tfrac{2 k . u + |u|^2}{(|k|^2+1)}\big)^{-(s+r)/2}=
\sum_{j=0}^M \genfrac(){0pt}{1}{-(s+r)/2}{j} \tfrac{\left(2 k . u +
|u|^2\right)^j}{(|k|^2+1)^j}+
O_{M, u}\big(  \tfrac{(1+|s|)^{M+1}}{(|k|^2+1)^{(M+1)/2}}\big).
$$
Thus, for
$\sigma =\Re(s)>n+d$,
\begin{eqnarray}\label{f0dev}
F_1(G;a;s)&=& e^{2\pi i \,u . a} \sum_{|\alpha|_1 \leq d}
\tfrac{u^{\alpha}}{\alpha !}
\sum_{k \in \Z^n } \tfrac{\partial^{\alpha}
P(k)}{(|k|^2+1)^{(s+r)/2}} \big(1+\tfrac{2
k . u +|u|^2}{(|k|^2+1)}\big)^{-(s+r)/2}
e^{2\pi i \,k . a}\nonumber \\
&=& e^{2\pi i \,u . a} \sum_{|\alpha|_1 \leq d} \sum_{j=0}^M
\tfrac{u^{\alpha}}{\alpha
!}  \genfrac(){0pt}{1}{-(s+r)/2}{j} \sum_{k \in \Z^n }
\tfrac{\partial^{\alpha} P(k) \left(2 k . u
+|u|^2\right)^j } {(|k|^2+1)^{(s+r+2j)/2}} \, e^{2\pi i
  \,k . a}\nonumber \\
& & \hspace{1cm}+O_{G, M,u} \big((1+|s|)^{M+1}\sum_{k \in \Z^n }
\tfrac{1}{(|k|^2+1)^{(\sigma
+M+1-g)/2}}\big).
\end{eqnarray}
Set $I:=\left\{(\alpha ,j) \in \N_0^n \times \{0,\dots,M\} \mid
|\alpha|_1 \leq d\right\}$ and $I^*:=I\setminus \set{(0,0)}$.\\
Set also $ G_{(\alpha ,j);u}(X):= \tfrac{\partial^{\alpha} P(X)
\left(2 X . u
+|u|^2\right)^j }
{(|X|^2+1)^{(r+2j)/2}}\in {\cal F}$ for all $(\alpha ,j) \in I^*$.\\
Since $M\geq N+n+g$, it follows from (\ref{f0dev}) that
\begin{eqnarray}\label{crucial1}
(1-e^{2\pi i \,u . a})~F_1(G;a;s)= e^{2\pi i \,u . a} \sum_{(\alpha,
j)\in I^*}
\tfrac{u^{\alpha}}{\alpha !}  \genfrac(){0pt}{1}{-(s+r)/2}{j}
F_1\left(G_{(\alpha ,j);u};\a;s\right)  +R_N(G; a; u; s)
\end{eqnarray}
where $s\mapsto R_N(G; a; u; s)$ is a holomorphic function in the
half plane
$\{\sigma =\Re(s) >-N\}$, in which it satisfies the bound
$R_N(G; a; u; s)\ll_{G,N,u} 1 $.\\
Moreover it is easy to see that, for any $(\alpha, j)\in I^*$,
$$
\text{deg}\hspace{-.06cm}\left(G_{(\alpha
,j);u}\right)=\text{deg}\hspace{-.05cm}\left(\partial^{\alpha}
P \right)+j -(r+2j)\leq d-|\alpha|_1 +j -(r+2j)=g-|\alpha|_1 -j\leq
g-1.
$$
Relation (\ref{crucial1}) and the induction hypothesis imply then that
\begin{equation}\label{crucial2}
(1-e^{2\pi i \,u . a })~F_1(G;a;s) {\mbox { verifies the conclusions
of Lemma \ref{ieffective}}}.
\end{equation}
Since $ |1-e^{2\pi i \,u . a}|=2|\sin(\pi u . a)|\geq d\left(u  . a,
\Z\right)$,
then (\ref{crucial2}) implies that $F_1(G;a;s)$ satisfies
conclusions of Lemma \ref{ieffective}. This completes the induction
and the proof for $i=1$.

\subsubsection{Proof of Lemma \ref{ieffective} for $i=0$:}
Let  $N\in \N$ be a fixed integer.
Let $G(X)=\frac{P(X)}{(X_1^2+\dots +X_n^2+1)^{r/2}}
\in {\cal F}$ and $g=$ deg$(G)=d-r$ where $d\geq 0$ is the degree of
the polynomial $P$.
Set also $M:=\sup (N+g+n, 0)\in \N_0$.\par
Since $P(k)\ll |k|^d$ for $k \in \Z^n\setminus
\set{0}$, it follows that $F_0(G;a;s)$ and $F_1(G;a;s)$ converge
absolutely in the
half plane
$\{\sigma =\Re(s)>n+g\}$.\\
Moreover, we have for $s=\sigma +i\tau  \in \C$ verifying $\sigma
>n+g$:
\begin{align}\label{crucial3}
F_0(G;a;s)&= \sum_{k \in \Z^n \setminus \set{0}}
\tfrac{G(k)}{(|k|^2+1-1)^{s/2}} \, e^{2\pi i \,k . a}
={\sum_{k \in \Z^n }}' \tfrac{G(k)}{(|k|^2+1)^{s/2}}
\left(1-\tfrac{1}{|k|^2+1}\right)^{-s/2} \,
e^{2\pi i \,k . a} \nonumber \\
&= {\sum_{k \in \Z^n }}' \,  \sum_{j=0}^M
\genfrac(){0pt}{1}{-s/2}{j} (-1)^j
\tfrac{G(k)}{(|k|^2+1)^{(s+2j)/2}} \,
e^{2\pi i \,k . a} \nonumber\\
&\hspace{2cm} +O_M\big((1+|s|)^{M+1} {\sum_{k \in \Z^n }}'
\tfrac{|G(k)|}{(|k|^2+1)^{(\sigma +2M+2)/2}}\big)\nonumber \\
&=\sum_{j=0}^M  \genfrac(){0pt}{1}{-s/2}{j}
(-1)^j F_1(G;a; s+2j) \nonumber\\
&\hspace{2cm}+ O_M\big[(1+|s|)^{M+1} \big(1+{\sum_{k \in \Z^n }}'
\tfrac{|G(k)|}{(|k|^2+1)^{(\sigma +2M+2)/2}}\big)\big].
\end{align}
In addition we have  uniformly in $s=\sigma +i\tau \in \C$ verifying
$\sigma >-N$,
$$
{\sum_{k \in \Z^n }}'
\tfrac{|G(k)|}{(|k|^2+1)^{(\sigma +2M+2)/2}}\ll {\sum_{k \in \Z^n }}'
\tfrac{|k|^g}{(|k|^2+1)^{(-N +2M+2)/2}}\ll {\sum_{k \in \Z^n }}'
\tfrac{1}{|k|^{n+1}}<+\infty.
$$
So (\ref{crucial3}) and Lemma \ref{ieffective} for
$i=1$ imply that Lemma \ref{ieffective} is also true for $i=0$. This
completes the
proof of Lemma \ref{ieffective}. \qed

\subsubsection{Proof of item $(i.2)$ of Theorem \ref{analytic}:}
Since  $a\in \R^n\setminus \Z^n$, there exists $i_0\in
\{1,\dots,n\}$ such that $a_{i_0}\not \in \Z$.
In particular $d(a . e_{i_0}, \Z)=d(a_{i_0},\Z)>0$. Therefore,  $a$
satisfies the assumption of Lemma \ref{ieffective} with  $u=e_{i_0}$.
Thus, for all $N\in \N$, $s\mapsto f_{a}(s)=F_0(P;a;s)$ has a
holomorphic
continuation to the half-plane $\{\Re (s)>-N\}$. It follows, by
letting
$N\rightarrow \infty$, that $s\mapsto f_{a}(s)$ has a holomorphic
continuation
to the whole complex plane $\C$.

\subsubsection{Proof of item $(iii)$ of Theorem \ref{analytic}:}
Let $\Th\in {\cal M}_n(\R)$, $(\eps_i)_i\in \{-1,0,1\}^{q}$ and $b
\in {\cal
S}(\Z^n\times \Z^n)$. We assume that
$\Th$ is a diophantine matrix. Set ${\cal Z}:= \set{
l=(l_1,\dots,l_q)\in (\Z^n)^q \,:  \, \sum_i \eps_i l_i =0}$ and
$P\in \C[X_1,\dots,X_n]$ of degree  $d\geq 0$.\\
It is easy to see that for $\sigma >n+d$:
\begin{align*}
\sum_{l\in (\Z^n)^{q}\setminus {\cal Z}} \, |b(l)| \,  {\sum_{k \in
\Z^n }}'
\tfrac{|P(k)|}{|k|^{\sigma}} \,|e^{2\pi i \,k . \Th\, \sum_i \eps_i
l_i }| & \ll_P  \sum_{l\in (\Z^n)^q\setminus {\cal Z}} |b (l)| \,
{\sum_{k \in \Z^n }}'
\tfrac{1}{|k|^{\sigma -d}}
\ll_{P,\sigma} \sum_{l \in (\Z^n)^q \setminus {\cal Z}} |b (l)|\\
&<+\infty.
\end{align*}
So
$$
g_0(s):=\sum_{l \in (\Z^n)^q \setminus {\cal Z}}
b (l) \, f_{\Th\, \sum_i \eps_i l_i}(s)= \sum_{l \in (\Z^n)^q
\setminus
  {\cal Z}}  b(l) \, {\sum_{k \in \Z^n }}'
\tfrac{P(k)}{|k|^{s}} e^{2\pi i \,k . \Th \, \sum_i \eps_i l_i}
$$
converges absolutely in the half plane $\{\Re(s) >n+d\}$.\\
Moreover with the notations of Lemma \ref{ieffective}, we have for
all $s=\sigma
+i\tau \in \C$ verifying $\sigma >n+d$:
\begin{equation}\label{crucial4}
g_0(s)=\sum_{l \in (\Z^n)^q\setminus {\cal Z}}
b(l) f_{\Th\, \sum_i \eps_i l_i}(s)=\sum_{l
\in (\Z^n)^q  \setminus {\cal Z}} b(l) F_0(P; \Th\, {\sum}_i \eps_i
l_i; s)
\end{equation}
But $\Th$ is diophantine, so there exists $u \in \Z^n$ and
$\delta ,c>0$ such
$$
|q . \,{}^t\Th u -m| \geq c\, (1+|q|)^{-\delta},\,  \forall q \in
\Z^n \setminus \set{0},
\, \forall m\in \Z.
$$
We deduce that $\forall l \in (\Z^n)^q \setminus {\cal Z},$
$$
\quad |\big(\Th\, {\sum}_i \eps_i l_i \big) .  u -m|= |\big({\sum}_i
\eps_i l_i\big) . {}^t\Th u -m| \geq c \,
\big(1+|{\sum}_i \eps_i l_i| \big)^{-\delta} \geq c \,
(1+|l|)^{-\delta}.
$$
It follows that there exists $u \in \Z^n$, $\delta >0$ and $c>0$ such
that
\begin{equation}\label{hypothesisOK!}
\forall l \in (\Z^n)^q \setminus {\cal Z}, \quad d\big((\Th\,
{\sum}_i \eps_i l_i) .  u;
\Z\big) \geq c \, (1+|l|)^{-\delta}.
\end{equation}
{\it Therefore, for any $l \in (\Z^n)^q \setminus {\cal Z}$,
the vector $a=\Th\, \sum_i \eps_i l_i$
verifies the assumption of Lemma \ref{ieffective} with the same $u$.
Moreover
$\delta$ and $c$ in
(\ref{hypothesisOK!}) are also independent on $l$}.\\
We fix now $N\in \N$. Lemma \ref{ieffective} implies that there exist
positive
constants $C_0:=C_0(P,N,u)$, $B_0:=B_i(P,N,u)$ and $A_0:=A_0(P,N,u)$
such that for all
$l \in (\Z^n)^q \setminus {\cal Z}$, $s\mapsto F_0(P; \Th\, \sum_i
\eps_i l_i; s)$ extends
{\it holomorphically} to the half plane $\{\Re(s)>-N\}$ and verifies
in it the bound
$$
F_0(P;\Th\, \sum_i \eps_i l_i; s)\leq C_0 \left(1+|s|\right)^{B_0} \,
d\big((\Th\, {\sum}_i \eps_i l_i) .  u; \Z\big) ^{-A_0}.
$$
This and (\ref{hypothesisOK!}) imply that for any compact set $K$
included
in the half plane $\{\Re(s)>-N\}$, there exist two constants
$C:=C(P,N,c, \delta,u, K)$ and  $D
:=D(P,N,c,\delta, u)$ (independent on $l \in (\Z^n)^q \setminus {\cal
Z}$) such that
\begin{equation}\label{crucialbig}
\forall s\in K {\mbox { and }} \forall l \in (\Z^n)^q \setminus {\cal
Z}, \quad F_0(P; \Th\,
\sum_i \eps_i l_i; s)\leq C \left(1+|l|\right)^{D}.
\end{equation}
It follows that $s\mapsto \sum_{l \in (\Z^n)^q \setminus {\cal Z}}
b(l) F_0(P;
\Th\, {\sum}_i \eps_i l_i;s)$
has a holomorphic continuation to the half plane
$\{\Re(s)>-N\}$.\\
This and ( \ref{crucial4}) imply that $
s\mapsto g_0(s)=\sum_{l \in (\Z^n)^q \setminus {\cal Z}} b(l)
f_{\Th\, \sum_i \eps_i l_i}(s)$
has a holomorphic continuation to
$\{\Re(s)>-N\}$. Since $N$ is an arbitrary integer, by letting
$N\rightarrow \infty$,
it follows that $s\mapsto g_0(s)$ has a holomorphic continuation to
the whole
complex plane $\C$ which completes the proof of the theorem. \qed

\begin{remark}
By equation (\ref{crucial2}), we see that a Diophantine
condition is sufficient to get Lemma \ref{ieffective}. Our Diophantine
condition appears also (in equivalent form) in Connes
\cite[Prop. 49]{NCDG} (see Remark 4.2 below). The following heuristic
argument shows that our condition seems to be necessary in order
to get the result of Theorem \ref{analytic}:\\
For simplicity we assume $n=1$ (but the argument extends easily to
any $n$).\\
Let $\theta \in \R \setminus \Q$. We know (see this reflection
formula in \cite[p. 6]{Elizaldebook})
that for any $l\in \Z\setminus\{0\}$,
$$
g_{\theta l}(s):={\sum_{k\in \Z}}' \, \tfrac{e^{2\pi i \theta l
k}}{|k|^s}=
\tfrac{\pi^{s-1/2}}{\Gamma (\tfrac{1-s}{2})}{\Gamma (\tfrac{s}{2})}
~h_{\theta l}(1-s) {\mbox { where }}
h_{\theta l}(s):={\sum_{k\in \Z}}'\,\tfrac{1}{|\theta l+ k|^s}.
$$
So, for any $(a_l) \in {\cal S}(\Z)$, the existence of meromorphic
continuation of
$g_0(s):=\sum_{l\in \Z} ' a_l \, g_{\theta l}(s)$ is equivalent to
the
existence of meromorphic continuation of
$$
h_0(s):={\sum_{l\in \Z}}' a_l \, h_{\theta l}(s)=
{\sum_{l\in \Z}} ' a_l \, {\sum_{k\in \Z}} ' \tfrac{1}{|\theta l+
k|^s}.
$$
So, for at least one $\sigma_0 \in \R$, we must have
$\tfrac{|a_l|}{|\theta l+ k|^{\sigma_0}} = O(1)
{\mbox { uniformly in }} k, l \in \Z^*.$

It follows that for any $(a_l) \in {\cal S}(\Z)$,
$|\theta l+ k| \gg |a_l|^{1/\sigma_0}$
uniformly in $k, l \in \Z^*$. Therefore, our Diophantine condition
seems
to be necessary.
\end{remark}

\subsubsection{Commutation between sum and residue}

Let $p\in \N$. Recall that $\mathcal{S}((\Z^n)^p)$ is the set of the
Schwartz
sequences on $(\Z^n)^p$. In other words, $b\in \mathcal{S}((\Z^n)^p)$
if and only if
for all $r\in \N_{0}$, $(1+|l_1|^2+\cdots
|l_p|^2)^r\,|b(l_1,\cdots,l_p)|^2$ is bounded on
$(\Z^n)^p$. We note that if $Q\in\R[X_1,\cdots,X_{np}]$ is a
polynomial, $(a_j)\in
\mathcal{S}(\Z^n)^{p}$, $b\in \mathcal{S}(\Z^n)$ and $\phi$ a
real-valued function, then
 $l:=(l_1,\cdots,l_p)\mapsto \wt a(l)\, b(-\wh l_p)\, Q(l)\,
 e^{i\phi(l)}$ is a Schwartz
sequence on $(\Z^n)^p$, where
\begin{align*}
\wt a(l) &:= a_1(l_1)\cdots a_p(l_p),\\
\wh l_i &:= l_1+\ldots +l_i.
\end{align*}

In the following, we will use several times the fact that for any
$(k,l)\in (\Z^n)^2$
such that $k\neq0$ and $k\neq -l$, we have
\begin{equation}\label{trick-0}
\frac{1}{|k+l|^2} = \frac{1}{|k|^2} -
\frac{2k.l+|l|^2}{|k|^2|k+l|^2}\,.
\end{equation}

\begin{lemma}\label{R-poly} There exists a polynomial $P\in
\R[X_1,\cdots,X_p]$ of degree $4p$ and with positive coefficients
such that for any
$k\in \Z^n$, and $l:=(l_1,\cdots,l_p)\in (\Z^n)^p$ such that $k\neq
0$ and $k\neq -\wh
l_i$ for all $1\leq i \leq p$, the following holds:
$$
\frac{1}{|k+\wh l_1|^2\ldots |k+\wh l_p|^2} \leq \frac{1}{|k|^{2p}}\
P(|l_1|,\cdots,|l_p|).
$$
\end{lemma}
\begin{proof} Let's fix $i$ such that $1\leq i\leq p$.
Using two times (\ref{trick-0}), Cauchy--Schwarz inequality and the
fact that $|k+\wh
l_i|^2\geq 1$, we get
\begin{align*}
\tfrac{1}{|k+\wh l_i|^2} &\leq \tfrac{1}{|k|^2} + \tfrac{2|k||\wh
l_i|+|\wh
l_i|^2}{|k|^4} + \tfrac{(2|k||\wh l_i|+|\wh l_i|^2)^2}{|k|^4|k+\wh
l_i|^2}\\
&\leq \tfrac{1}{|k|^2}+ \tfrac{2}{|k|^3}|\wh l_i| +
\big(\tfrac{1}{|k|^4}+\tfrac{4}{|k|^2}\big)|\wh l_i|^2 +
\tfrac{4}{|k|^3}|\wh l_i|^3 +
\tfrac{1}{|k|^4}|\wh l_i|^4.
\end{align*}
Since $|k|\geq 1$, and $|\wh l_i|^j \leq |\wh l_i|^4$ if $1\leq j\leq
4$, we find
\begin{align*}
&\tfrac{1}{|k+\wh l_i|^2} \leq \tfrac{5}{|k|^2} {\sum}_{j=0}^4 \,|\wh
l_i|^j \leq
\tfrac{5}{|k|^2} \big(1 + 4 |\wh l_i|^4\big)
\leq \tfrac{5}{|k|^2} \big(1 + 4 ({\sum}_{j=1}^p \, |l_j|)^4\big), \\
&\tfrac{1}{|k+\wh l_1|^2\ldots |k+\wh l_p|^2}  \leq
\tfrac{5^p}{|k|^{2p}} \big(1 + 4
({\sum}_{j=1}^p \,|l_j|)^{4}\big)^p.
\end{align*}
Taking $P(X_1,\cdots,X_p):= 5^p \big(1 + 4 ({\sum}_{j=1}^p
X_j)^{4}\big)^p$ now gives
the result.
\end{proof}

\begin{lemma}\label{abs-som}
Let $b\in \mathcal{S}((\Z^n)^p)$, $p\in\N$, $P_j\in
\R[X_1,\cdots,X_n]$ be a
homogeneous polynomial function of degree $j$, $k\in \Z^n$,
$l:=(l_1,\cdots,l_p)\in
(\Z^n)^p$, $r\in \N_{0}$, $\phi$ be a real-valued function on
$\Z^n\times (\Z^n)^{p}$ and
$$
h(s,k,l):=\frac{b(l)\, P_j(k)\, e^{i\phi(k,l)}}{|k|^{s+r}|k+\wh
l_1|^2\cdots |k+\wh
l_p|^2} \, ,
$$
with $h(s,k,l):=0$ if, for $k\neq 0$, one of the
denominators is zero.

For all $s\in \C$ such that $\Re(s)>n+j-r-2p$, the series
$$
H(s):={{\sum}'}_{(k,l)\in (\Z^n)^{p+1}} h(s,k,l)
$$
is absolutely summable. In particular,
$$
{\sum_{k\in\Z^n}}' \sum_{l\in (\Z^n)^p} h(s,k,l) = \sum_{l\in
  (\Z^n)^p}
{\sum_{k\in\Z^n}} ' h(s,k,l)\,.
$$
\end{lemma}

\begin{proof}
Let $s=\sigma+i\tau \in \C$ such that
$\sigma>n+j-r-2p$. By Lemma \ref{R-poly} we get, for $k\neq 0$,
$$
|h(s,k,l)|\leq |b(l)\, P_j(k)|\, |k|^{-r-\sigma-2p}\, P(l),
$$
where $P(l):=P(|l_1|,\cdots,|l_p|)$ and $P$ is a polynomial of degree
$4p$ with
positive coefficients. Thus,
$|h(s,k,l)|\leq  F(l)\, G(k)$
where $F(l):=|b(l)|\, P(l)$ and $G(k):= |P_j(k)| |k|^{-r-\sigma-2p}$.
The summability
of $\sum_{l\in (\Z^n)^p} F(l)$ is implied by the fact that $b\in
\mathcal{S}((\Z^n)^p)$. The summability of ${\sum}'_{k\in \Z^n} G(k)$
is a consequence
of the fact that $\sigma>n+j-r-2p$. Finally, as a product of two
summable series,
${\sum_{k,l}} F(l) G(k)$ is a summable series, which proves that
${\sum_{k,l}}h(s,k,l)$ is also absolutely summable.
\end{proof}

\begin{definition}
Let $f$ be a function on $D\times (\Z^n)^p$ where $D$ is an open
neighborhood of $0$ in $\C$.

We say that
$f$ satisfies (H1) if and only if there exists $\rho>0$ such that

\hspace{1cm} (i) for any $l$, $s\mapsto f(s,l)$ extends as a
holomorphic function on $U_\rho$,
where $U_\rho$ is the
open disk of center 0 and radius $\rho$,

\hspace{1cm}(ii) the series $\sum_{l\in (\Z^n)^p}
\norm{H(\cdot,l)}_{\infty,\rho}$
is summable,where $\norm{H(\cdot,l)}_{\infty,\rho}:=\sup_{s\in
U_{\rho}}|H(s,l)|$.\\
 We say that
$f$ satisfies (H2) if and only if there exists $\rho>0$ such that

\hspace{1cm}
(i) for any $l$, $s\mapsto f(s,l)$ extends as a holomorphic function
on
$U_\rho-\{0\}$,

\hspace{1cm}
(ii) for any $\delta$ such that $0<\delta<\rho$, the series
$\sum_{l\in (\Z^n)^p}
\norm{H(\cdot,l)}_{\infty,\delta,\rho}$ is summable, where
$\norm{H(\cdot,l)}_{\infty,\delta,\rho}:=\sup_{\delta<|s|<\rho}|H(s,l)|$.
\end{definition}

\begin{remark}
Note that (H1) implies (H2). Moreover, if $f$ satisfies (H1)
(resp. (H2) for $\rho>0$, then it is straightforward to check that
$f:s\mapsto \sum_{l\in (\Z^n)^p} f(s,l)$
extends as an holomorphic function on $U_\rho$ (resp. on $U_\rho
\setminus \set{0}$).
\end{remark}

\begin{corollary}\label{res-somH} With the same notations of Lemma
\ref{abs-som},
suppose that $r+2p-j>n$, then, the function
$H(s,l):={\sum}'_{k\in \Z^n} h(s,k,l)$ satisfies (H1).
\end{corollary}

\begin{proof} $(i)$ Let's fix $\rho>0$ such that $\rho < r+2p-j-n$.
Since $r+2p-j>n$, $U_\rho$ is inside the half-plane of absolute
convergence of
the series defined by $H(s,l)$. Thus, $s\mapsto H(s,l)$ is
holomorphic on $U_\rho$.\\
$(ii)$ Since $\big||k|^{-s}\big|\leq |k|^{\rho}$ for all $s\in
U_\rho$ and
$k\in\Z^n \setminus \set{0}$, we get as in the above proof
$$
|h(s,k,l)|\leq |b(l)\, P_j(k)| \, |k|^{-r+\rho-2p} \,
P(|l_1|,\cdots,|l_p|).
$$
Since $\rho < r+2p-j-n$, the series ${\sum}'_{k\in \Z^n} |P_j(k)|
|k|^{-r+\rho-2p}$ is
summable.

Thus, $\norm{H(\cdot,l)}_{\infty,\rho} \leq K \, F(l)$ where $K :=
{{\sum_k}}'|P_j(k)| |k|^{-r+\rho-2p}<\infty$. We have already seen
that the series
$\sum_l F(l)$ is summable, so we get the result.
\end{proof}

We note that if $f$ and $g$ both satisfy (H1) (or (H2)),
then so does $f+g$. In the
following, we will use the equivalence relation
$$
f\sim g \Longleftrightarrow f-g \text{ satisfies (H1)}.
$$

\begin{lemma}\label{res-som}
Let $f$ and $g$ be two functions on $D\times (\Z^n)^p$ where $D$ is
an open
neighborhood of $0$ in $\C$, such that $f\sim g$ and such that $g$
satisfies (H2). Then
$$
\underset{s=0}{\Res}\sum_{l \in (\Z^n)^p} f(s,l)=\sum_{l \in (\Z^n)^p}
\underset{s=0}{\Res}\ g(s,l)\, .
$$
\end{lemma}
\begin{proof} Since $f\sim g$, $f$ satisfies (H2) for a certain
$\rho>0$.
Let's fix $\eta$ such that $0<\eta<\rho$ and define $C_\eta$ as the
circle of center 0
and radius $\eta$. We have
$$
\underset{s=0}{\Res}\ g(s,l) = \underset{s=0}{\Res}\ f(s,l) = \tfrac
{1}{2\pi
i}\oint_{C_\eta} f(s,l)\, ds = \int_I u(t,l) dt\, .
$$
where $I=[0,2\pi]$ and $u(t,l):=\tfrac {1}{2\pi} \eta e^{it}
f(\eta\,e^{i t},l) $. The
fact that $f$ satisfies (H2) entails that the series $\sum_{l\in
(\Z^n)^p}
\norm{f(\cdot,l)}_{\infty,C_\eta}$ is summable. Thus, since
$\norm{u(\cdot,l)}_{\infty} = \tfrac {1}{2\pi} \eta
\norm{f(\cdot,l)}_{\infty,C_\eta}$, the series $\sum_{l\in (\Z^n)^p}
\norm{u(\cdot,l)}_{\infty}$ is summable, so, as a consequence,
$\int_I \sum_{l\in
(\Z^n)^p}   u(t,l) dt = \sum_{l\in (\Z^n)^p}  \int_I u(t,l)dt$ which
gives the result.
\end{proof}

\subsection{Computation of residues of zeta functions}

Since, we will have to compute residues of series, let us introduce
the following

\begin{definition}
\begin{align*}
\zeta(s)&:=\sum_{n=1}^{\infty} n^{-s},\\
Z_{n}(s)&:={\sum_{k\in\Z^{n}}}'\,\,\vert k\vert^{-s},\\
\zeta_{p_{1},\dots,p_{n}}(s)&:={\sum_{k\in\Z^{n}}}'\,\,
\frac{k_1^{p_{1}}\cdots
k_n^{p_{n}}}{\vert k\vert^{s}}\,\text{ , for } p_{i}\in \N,
\end{align*}
\end{definition}
\noindent where $\zeta(s)$ is the Riemann zeta function (see
\cite{HW} or
\cite{Edery}).

By the symmetry $k\rightarrow -k$, it is clear that these functions
$\zeta_{p_{1},\dots,p_{n}}$ all vanish for odd values of $p_{i}$.

Let us now compute
$\zeta_{0,\cdots,0,1_{i},0\cdots,0,1_{j},0\cdots,0}(s)$ in terms of
$Z_{n}(s)$:\\
Since $\zeta_{0,\cdots,0,1_{i},0\cdots,0,1_{j},0\cdots,0}(s)
=A_{i}(s)\,\delta_{ij}$,
exchanging the components $k_{i}$ and $k_{j}$, we get
\begin{align*}
\zeta_{0,\cdots,0,1_{i},0\cdots,0,1_{j},0\cdots,0}(s)
=\tfrac{\delta_{ij}}{n}\,Z_{n}(s-2).
\end{align*}
Similarly,
\begin{align*}
{\sum}'_{\Z^{n}}\,\tfrac{k_{1}^{2}k_{2}^{2}}{\vert k\vert^{s+8}}
=\tfrac{1}{n(n-1)}
Z_{n}(s+4)- \tfrac{1}{n-1}{\sum}'_{\Z^{n}}\, \tfrac{k_{1}^4}{\vert
k\vert^{s+8}}
\end{align*}
but it is difficult to write explicitly
$\zeta_{p_{1},\dots,p_{n}}(s)$ in terms of
$Z_{n}(s-4)$ and other $Z_{n}(s-m)$ when at least four indices
$p_{i}$ are non zero.

When all $p_{i}$ are even, $\zeta_{p_{1},\dots,p_{n}}(s)$ is a
nonzero series of
fractions $\tfrac{P(k)}{\vert k\vert ^s}$ where $P$ is a homogeneous
polynomial of
degree $p_{1}+\cdots +p_{n}$. Theorem \ref{res-int} now gives us the
following

\begin{prop}
    \label{calculres}
$\zeta_{p_{1},\dots,p_{n}}$ has a meromorphic extension to the whole
plane with a
unique pole at $n+p_{1} +\cdots +p_{n}$. This pole is simple and the
residue at this
pole is
\begin{align}
    \label{formule1}
\underset{s=n+p_{1} +\cdots
+p_{n}}{\Res} \,\zeta_{p_{1},\dots,p_{n}}(s)= 2 \,
\tfrac{\Gamma(\tfrac{p_{1}+1}{2}) \cdots
\Gamma(\tfrac{p_{n}+1}{2})}
{\Gamma
(\tfrac{n+p_{1}+ \cdots + p_{n}}{2})}
\end{align}
when all $p_{i}$ are even or this residue is zero otherwise.\\
In particular, for $n=2$,
\begin{align}
\underset{s=0}{\Res} \,\,{\sum_{k\in\Z^2}}'\,\tfrac{k_{i}k_{j}} {\vert
k\vert^{s+4}}=\delta_{ij}\,\pi\, , \label{formulen=2}
\end{align}
and for $n=4$,
\begin{align}
&\underset{s=0}{\Res} \,\,{\sum_{k\in\Z^{4}}}'\,\tfrac{k_{i}k_{j}}
{\vert k\vert^{s+6}}=\delta_{ij}\tfrac{\pi^2}{2}\, ,\nonumber\\
&\label{formule2} \underset{s=0}{\Res}
\,{\sum_{k\in\Z^{4}}}'\,\tfrac{k_{i}k_{j}k_{l}k_{m}}
{\vert k\vert^{s+8}}=(\delta_{ij}\delta_{lm}+\delta_{il}\delta_{jm}
+\delta_{im}\delta_{jl})\,\tfrac{\pi^2}{12}\, .
\end{align}
\end{prop}
\begin{proof}
Equation (\ref{formule1}) follows from Theorem (\ref{res-int})
$$
\underset{s=n+p_{1} +\cdots
+p_{n}}{\Res}\, \, \zeta_{p_{1},\dots,p_{n}}(s)=
\int_{k \in
S^{n-1}}k_{1}^{p_{1}} \cdots k_{n}^{p_{n}}\, dS(k)
$$
and standard formulae (see for instance \cite[VIII,1;22]{Schwartz}).
Equation
(\ref{formulen=2}) is a straightforward consequence of Equation
(\ref{formule1}).
Equation (\ref{formule2}) can be checked for the cases $i=j\neq l=m$
and
$i=j=l=m$.
\end{proof}
 Note that $Z_n(s)$ is an Epstein zeta function associated to the
quadratic form $q(x):=x_1^2+...+x_n^2$, so $Z_n$ satisfies the
following functional equation
$$
Z_n(s)= \pi^{s-n/2} \Gamma (n/2 -s/2)\Gamma (s/2)^{-1}\,
Z_n(n-s).
$$
Since $\pi^{s-n/2} \Gamma (n/2 -s/2) \,\Gamma (s/2)^{-1}=0$
for any negative even integer $n$ and $Z_n(s)$ is meromorphic on $\C$
with only one pole at $s=n$ with residue $2 \pi^{n/2} \Gamma
(n/2)^{-1}$ according to previous proposition, so we get $Z_n(0)=
-1$. We have proved that
\begin{align}
\label{formule}
    \underset{s=0}{\Res} \,\, Z_{n}(s+n)&=2\pi^{n/2} \, \Gamma
(n/2)^{-1},\\
Z_n(0)&= -1.
\label{Zn0}
\end{align}

\subsection{Meromorphic continuation of a class of zeta functions}

Let $n,q\in \N$, $q\geq2$, and $p=(p_1,\dots,p_{q-1}) \in
\N_0^{q-1}$.\\
Set $I:=\{ i \mid p_i\neq 0\}$ and assume that $I\neq \emptyset$ and
$${\cal I}:=\{\alpha =(\alpha_i)_{i\in I} \mid \forall i\in I ~
\alpha_i=(\alpha_{i,1},\dots, \alpha_{i,p_i})\in
\N_0^{p_i}\}=\prod_{i\in I} \N_0^{p_i}.$$

We will use in the sequel also the following notations:

\hspace{1cm} - for $x=(x_1,\dots,x_t) \in \R^t$ recall that
$|x|_1=|x_1|+\dots+|x_t|$ and $|x|=\sqrt{x_1^2+\dots+x_t^2}$;

\hspace{1cm} - for all $\alpha =(\alpha_i)_{i\in I}
  \in {\cal I} =\prod_{i\in I} \N_0^{p_i}$,
$$|\alpha|_1=\sum_{i\in I} |\alpha_i|_1 =\sum_{i\in I}
\sum_{j=1}^{p_i} |\alpha_{i,j}| {\mbox { and }}
\genfrac(){0pt}{1}{1/2}{\alpha} =\prod_{i\in I}
\genfrac(){0pt}{1}{1/2}{\alpha_{i}}=
\prod_{i\in I} \prod_{j=1}^{p_i}
\genfrac(){0pt}{1}{1/2}{\alpha_{i,j}}.$$

\subsubsection{A family of polynomials}
In this paragraph we define a family of polynomials which plays an
important role later.

Consider first the variables:

- for $X_1,\dots, X_n$ we set $X=(X_1,\dots,X_n)$;

- for any $i=1,\dots,2q$, we consider the variables $Y_{i,1},\dots,
  Y_{i,n}$ and set
$Y_i:=(Y_{i,1},\dots, Y_{i,n})$ and $Y:=(Y_1,\dots,Y_{2q})$;

- for $Y=(Y_1,\dots,Y_{2q})$, we set for any $1\leq j\leq q$,
$\wt Y_j:= Y_1+\cdots+ Y_j+  Y_{q+1}+\cdots + Y_{q+j}$.

We define for all $\alpha =(\alpha_i)_{i\in I}\in {\cal I}
=\prod_{i\in I} \N_0^{p_i}$ the polynomial
\begin{equation}
\label{palphaxy}
P_\alpha(X,Y):= \prod_{i\in I} \prod_{j=1}^{p_i}
(2\langle X, \wt Y_i\rangle + |\wt
    Y_i|^2)^{\alpha_{i,j}}.
\end{equation}

It is clear that $P_\alpha(X,Y) \in \Z[X,Y]$, deg$_X P_\alpha \leq
|\alpha|_1$ and deg$_Y P_\alpha \leq 2 |\alpha|_1$.

Let us fix a polynomial $Q\in \R[X_1,\cdots,X_n]$ and
note $d:= \deg Q$.
For $\a\in  {\cal I}$, we want to expand $P_\alpha(X,Y) \, Q(X)$ in
homogeneous polynomials in $X$ and $Y$ so defining
$$
L(\a):=\set{\beta\in\N_0^{(2q+1)n}\, :\,
|\beta|_1-d_\beta \leq 2|\a|_1 \text{ and } d_\beta\leq |\a|_1+d}
$$
where
$d_\beta := \sum_1^n \beta_i$, we set
$$
\genfrac(){0pt}{1}{1/2}{\alpha} P_\alpha(X,Y) \, Q(X) =:
\sum_{\beta\in L(\a)} c_{\a,\beta} \,  X^\beta Y^\beta
$$
where $c_{\a,\beta}\in \R$, $X^\beta:= X_1^{\beta_{1}} \cdots
X_n^{\beta_{n}}$ and $Y^{\beta}:=
Y_{1,1}^{\beta_{n+1}}\cdots Y_{2q,n}^{\beta_{(2q+1)n}}$. By
definition, $X^\beta$ is a
homogeneous polynomial of degree in $X$ equals to $d_\beta$.
We note $$M_{\a,\beta}(Y):=c_{\a,\beta} \, Y^\beta.$$

\subsubsection{Residues of a class of zeta functions}

In this section we will prove the following result, used in
Proposition \ref{zeta(0)} for
the computation of the spectrum dimension of the noncommutative
torus:

\begin{theorem}
\label{zetageneral}
(i) Let $\tfrac{1}{2\pi}\Th$ be a
diophantine matrix, and $\wt a \in \mathcal{S}
\big((\Z^{n})^{2q}\big)$. Then
$$
s\mapsto f(s):= \sum_{l\in [(\Z^n)^{q}]^2} \wt a_{l}\ {\sum_{k\in
\Z^n}}'\, \prod_{i=1}^{q-1}|k+\wt l_i|^{p_i} |k|^{-s}\, Q(k)\,
e^{ik.\Th \sum_1^{q} l_j}
$$
has a meromorphic continuation to the whole complex
 plane $\C$ with at most simple
possible poles at the points $s=n+d+|p|_1-m$ where $m\in \N_0$.

(ii) Let $m\in \N_0$ and set
$I(m):= \set{(\a,\beta)\in \mathcal{I}\times \N_0^{(2q+1)n} \, :
\, \beta\in L(\a) \text{ and } m=2|\a|_1 -d_\beta +d } $.
Then $I(m)$ is a finite set and $s=n+d+|p|_1-m$ is a pole of $f$
if and only if
$$
C(f,m):= \sum_{l\in Z} \wt a_l
\sum_{(\a,\beta)\in I(m)} M_{\a,\beta}(l)  \int_{u\in S^{n-1}}
u^\beta\, dS(u) \neq 0,
$$
with $Z:=\{l \,: \, \sum_1^{q} l_j=0 \}$ and the convention
$\sum_{\emptyset} =0$.
In that case $s=n+d+|p|_1-m$ is a simple pole of residue
$\underset{s= n+d+|p|_1 -m}{\Res} \, f(s) = C (f,m)$.
\end{theorem}

In order to prove the theorem above we need the following

\begin{lemma}
\label{zetageneral-lem}
For all $N\in \N$ we have
$$
 \prod_{i=1}^{q-1} |k+\wt l_i|^{p_i}=
\sum_{\alpha =(\alpha_i)_{i\in I}
  \in \prod_{i\in I}\{0,\dots,N\}^{p_i}}
  \genfrac(){0pt}{1}{1/2}{\alpha}\,
\tfrac{P_\alpha (k,l)}{|k|^{2|\alpha|_1-|p|_1}}
+\mathcal{O}_N(|k|^{|p|_1- (N+1)/2})
$$
uniformly in $k\in \Z^n$ and $l\in
(\Z^n)^{2q}$ verifying $|k| > U(l):=36\,
(\sum_{i=1,\, i\neq q}^{2q-1}|l_{i}|)^4$.
\end{lemma}
\begin{proof}
For $i=1,\dots,q-1$, we have uniformly in $k\in \Z^n$ and $l\in
(\Z^n)^{2q}$ verifying $|k| > U(l)$,
\begin{equation}
\label{devjustification}
\tfrac{\big|2\langle k, \wt l_i \rangle+|\wt l_i|^2\big|}{|k|^2}
\leq\tfrac{\sqrt{U(l)}}{2|k|} < \tfrac{1}{2\sqrt{|k|}}.
\end{equation}
In that case,
\begin{eqnarray*}
|k+\wt l_i|&=& \big(|k|^2+2\langle k, \wt l_i\rangle
  + |\wt l_i|^2\big)^{1/2} =
|k| \big(1+ \tfrac{2\langle k, \wt l_i\rangle
  + |\wt l_i|^2}{|k|^2}\big)^{1/2} =
 \sum_{u=0}^{\infty}  \genfrac(){0pt}{1}{1/2}{u} \,
\tfrac{1}{|k|^{2u-1}}P^i_u(k,l)
\end{eqnarray*}
where for all $i=1,\dots, q-1$ and for all $u\in \N_0$,
\begin{equation*}
P^i_u(k,l):=\big(2\langle k, \wt l_i\rangle + |\wt l_i|^2\big)^u,
\end{equation*}
with the convention $P^i_0(k,l):=1$.

In particular $P^i_u(k,l)\in \Z[k,l]$,
$\deg_{k} P^i_u\leq u$ and $\deg_{l} P^i_u\leq 2u$.
Inequality (\ref{devjustification}) implies that for all
$i=1,\dots,q-1$
and for all $u\in \N$,
$$
\tfrac{1}{|k|^{2u}}\,|P^i_u(k,l)|\leq \big(2\sqrt{|k|}\big)^{-u}
$$
uniformly in $k\in \Z^n$ and $l\in
(\Z^n)^{2q}$ verifying $|k| > U(l)$.

Let $N\in \N$. We deduce from the previous that for any
$k\in \Z^n$ and $l\in
(\Z^n)^{2q}$ verifying $|k| > U(l)$ and
for all $i=1,\dots,q-1$, we have
\begin{eqnarray*}
|k+\wt l_i|&=& \sum_{u=0}^{N}
\genfrac(){0pt}{1}{1/2}{u} \, \tfrac{1}{|k|^{2u-1}}P^i_u(k,l)+
\mathcal{O}\big(\sum_{u>N}|k|\,|\genfrac(){0pt}{1}{1/2}{u}|\,
(2\sqrt{|k|})^{-u}\big)\\
&=& \sum_{u=0}^{N} \genfrac(){0pt}{1}{1/2}{u} \,
\tfrac{1}{|k|^{2u-1}}P^i_u(k,l)+\mathcal{O}_N
\big(\tfrac{1}{|k|^{(N-1)/2}}\big).
\end{eqnarray*}
It follows that for any $N\in \N$, we have uniformly in
$k\in \Z^n$ and $l\in (\Z^n)^{2q}$ verifying $|k| > U(l)$ and
for all $i\in I$,
$$
|k+\wt l_i|^{p_i}=\sum_{\alpha_i \in
  \{0,\dots,N\}^{p_i}} \genfrac(){0pt}{1}{1/2}{\alpha_i} \,
\tfrac{1}{|k|^{2|\alpha_i|_1-p_i}}P^i_{\alpha_i} (k,l)
+\mathcal{O}_N \left(\tfrac{1}{|k|^{(N+1)/2-p_i}}\right)
$$
where $P^i_{\alpha_i} (k,l)=\prod_{j=1}^{p_i} P^i_{\alpha_{i,j}}(k,l)$
for all $\alpha_i =(\alpha_{i,1},\dots,\alpha_{i,p_i})\in
  \{0,\dots,N\}^{p_i}$ and
$$
\prod_{i\in I} |k+\wt l_i|^{p_i}=\sum_{\alpha=(\alpha_i) \in
  \prod_{i\in I} \{0,\dots,N\}^{p_i}} \genfrac(){0pt}{1}{1/2}{\alpha}
\,
\tfrac{1}{|k|^{2|\alpha|_1-|p|_1}}P_{\alpha} (k,l)+\mathcal{O}_N
\big(\tfrac{1}{|k|^{(N+1)/2 -|p|_1}}\big)
$$
where $P_{\alpha } (k,l)=\prod_{i\in I} P^i_{\alpha_{i}}(k,l)=
\prod_{i\in I} \prod_{j=1}^{p_i} P^i_{\alpha_{i,j}}(k,l)$.
\end{proof}

\medskip

{\it Proof of Theorem \ref{zetageneral}.}

$(i)$ All $n$, $q$, $p=(p_1,\dots,p_{q-1})$ and $\wt a
\in {\cal S}\left((\Z^n)^{2q}\right)$ are fixed as above and we
define formally for any $l \in (\Z^n)^{2q}$
\begin{equation}
\label{fls}
F(l,s):= {\sum_{k\in \Z^n}}' \,
\prod_{i=1}^{q-1} |k+\wt l_i|^{p_i}\, Q(k)\,
e^{ik.\Th \sum_1^{q}l_j}\,|k|^{-s}.
\end{equation}
Thus, still formally,
\begin{equation}\label{fsexpfls}
f(s):=\sum_{l\in (\Z^n)^{2q}} \wt a_l\ F(l,s).
\end{equation}
It is clear that $F(l,s)$ converges absolutely in the half
plane $\{\sigma=\Re(s) >n+d+|p|_1\}$ where $d=\deg Q$.

Let $N\in \N$. Lemma \ref{zetageneral-lem} implies
that for any $l\in (\Z^n)^{2q}$ and for $s\in \C$ such
that $\sigma >n+|p|_1+d$,
\begin{align*}
F(l,s)&= {\sum_{|k|\leq U(l)}}' \,
\prod_{i=1}^{q-1} |k+\wt l_i|^{p_i}\, Q(k)\,
e^{ik.\Th \sum_1^{q}l_j}\,|k|^{-s} \\
& \quad \quad  +\sum_{\alpha =(\alpha_i)_{i\in I}
  \in \prod_{i\in
I}\{0,\dots,N\}^{p_i}}\genfrac(){0pt}{1}{1/2}{\alpha}
\sum_{|k| > U(l)} \tfrac{1}{|k|^{s+2|\alpha|_1-|p|_1}}P_\alpha (k,l)
Q(k)\, e^{ik.\Th \sum_1^{q}l_j}
+ G_N(l,s).
\end{align*}
where $s\mapsto G_N(l,s)$ is a holomorphic function in the half-plane
$D_N:=\{\sigma > n+d+|p|_1-\tfrac{N+1}{2}\}$ and verifies in it the
bound
$G_N(l,s) \ll_{N,\sigma} 1$ uniformly in $l$.

It follows that
\begin{equation}
\label{flsexpzeta}
F(l,s)= \sum_{\alpha =(\alpha_i)_{i\in I}
  \in \prod_{i\in I}\{0,\dots,N\}^{p_i}} H_{\a}(l,s)+ R_N(l,s),
\end{equation}
where
\begin{eqnarray*}
H_{\a}(l,s)&:=&{\sum_{k\in \Z^n}}'\,
\genfrac(){0pt}{1}{1/2}{\alpha} \,
\tfrac {1}{|k|^{s+2|\alpha|_1-|p|_1}}P_\alpha (k,l)\, Q(k)\,
e^{ik.\Th \sum_1^{q}l_j},\\
R_N(l,s)&:=& {\sum_{|k|\leq U(l)}}' \,
\prod_{i=1}^{q-1} |k+\wt l_i|^{p_i}\, Q(k)\,
e^{ik.\Th \sum_1^{q}l_j}\,|k|^{-s}\\
& & \quad -{\sum_{|k|\leq U(l)}}' \quad\sum_{\alpha =
(\alpha_i)_{i\in I} \in \prod_{i\in I}\{0,\dots,N\}^{p_i}}
  \genfrac(){0pt}{1}{1/2}{\alpha} \tfrac{P_\alpha (k,l)}
{|k|^{s+2|\alpha|_1-|p|_1}}Q(k)\,
  e^{ik.\Th \sum_1^{q}l_j}
+ G_N(l,s).
\end{eqnarray*}
In particular there exists $A(N)>0$ such that
$s\mapsto R_N(l,s)$ extends holomorphically to the half-plane
$D_N$ and verifies in it the bound
$R_N(l,s) \ll_{N,\sigma} 1 +|l|^{A(N)}$ uniformly in $l$.

Let us note formally
$$
h_\a(s):= \sum_l \wt a_l\, H_\a(l,s).
$$
Equation (\ref{flsexpzeta}) and $R_N(l,s) \ll_{N,\sigma} 1
+|l|^{A(N)}$ imply that
\begin{equation}
\label{fssimN}
f(s) \sim_N \sum_{\alpha =(\alpha_i)_{i\in I}
  \in \prod_{i\in I}\{0,\dots,N\}^{p_i}} h_\a(s),
\end{equation}
where $\sim_N$ means modulo a holomorphic function in $D_N$.

Recall the decomposition
$\genfrac(){0pt}{1}{1/2}{\alpha} \,P_\a(k,l) \, Q(k)=
\sum_{\beta\in L(\a)} M_{\a,\beta}(l) \, k^\beta$ and we
decompose similarly
$h_{\a}(s) =\sum_{\beta\in L(\a)} h_{\a,\beta}(s).$
Theorem \ref{analytic} now implies that for all
$\alpha =(\alpha_i)_{i\in I}  \in \prod_{i\in I}\{0,\dots,N\}^{p_i}$
and $\beta\in L(\a)$,

\quad - the map $s\mapsto h_{\a,\beta}(s)$ has a meromorphic
continuation to the whole complex plane $\C$ with only one
simple possible pole at $s=n+ |p|_1 - 2|\a|_1 +d_\beta$,

\quad - the residue at this point is equal to
\begin{equation}
\label{res-halphaj}
\underset {s=n+ |p|_1 - 2|\a|_1 +d_\beta}{\Res}\,
h_{\a,\beta}(s) =
\sum_{l\in \mathcal{Z}} \wt a_l\, M_{\a,\beta}(l) \int_{u\in S^{n-1}}
u^\beta dS(u)
\end{equation}
where $\mathcal{Z}:=\{l\in (\Z)^{n})^{2q} \, : \, \sum_1^{q} l_j =0
\}$.
If the right hand side is zero, $h_{\a,\beta}(s)$ is holomorphic on
$\C$.

By (\ref{fssimN}), we deduce therefore that
$f(s)$ has a meromorphic continuation on the halfplane $D_N$,
with only simple possible poles in the set
$
\set{n+|p|_1 + k  \,: \,-2N|p|_1\leq k \leq d}.
$
Taking now $N\to \infty$ yields the result.

$(ii)$ Let $m\in \N_0$ and set
$I(m):= \set{(\a,\beta)\in \mathcal{I}\times \N_0^{(2q+1)n} \,:
\, \beta\in L(\a) \text{ and } m=2|\a|_1 -d_\beta +d } $.
If $(\a,\beta)\in I(m)$, then $|\a|_1 \leq m$ and
$|\beta|_1\leq 3m+d$, so $I(m)$ is finite.

With a chosen $N$ such that $2N|p|_1+d>m$, we get by (\ref{fssimN})
and (\ref{res-halphaj})
$$
\underset{s= n+d+|p|_1 -m}{\Res} \,f(s) =
\sum_{l\in \mathcal{Z}} \wt a_l
\sum_{(\a,\beta)\in I(m)} M_{\a,\beta}(l)  \int_{u\in S^{n-1}}
u^\beta\, dS(u)=C(f,m)
$$
with the convention $\sum_{\emptyset} =0$. Thus, $n+d+|p|_1 - m$ is a
pole of $f$ if and only if
$C(f,m)\neq 0$.
{\qed}

\section{Noncommutative integration on a simple spectral triple}

In this section, we revisit the notion of noncommutative integral
pioneered by Alain Connes, paying particular attention to the
reality (Tomita--Takesaki) operator $J$ and to kernels of
perturbed Dirac operators by symmetrized one-forms.

\subsection{Kernel dimension}

We will have to compare here the kernels of $\DD$ and $\DD_{A}$ which
are both finite dimensional:

\begin{lemma}
\label{compres} Let $(\A,\H,\DD)$ be a spectral triple with a reality
operator $J$ and chirality $\chi$. If $A \in \Omega^1_{\DD}$ is a
one-form, the fluctuated Dirac operator $$\DD_A:= \DD+A+ \epsilon
JAJ^{-1}$$
(where $\DD J=\epsilon \,J \DD$, $\epsilon =\pm 1$) is an operator
with compact resolvent, and in particular its kernel $\Ker \DD_A$ is
a finite dimensional space. This space is invariant by $J$ and $\chi$.
\end{lemma}

\begin{proof}
Let $T$ be a bounded operator and let $z$ be in the resolvent of
$\DD+T$ and $z'$ be in the resolvent of $\DD$. Then
$$
(\DD+T-z)^{-1}=(\DD-z')^{-1} \, [1-(T+z'-z)(\DD+T-z)^{-1}].
$$
Since $(\DD-z')^{-1}$ is compact by hypothesis and since the term in
bracket is
bounded, $\DD+T$ has a compact resolvent. Applying this to
$T=A+\epsilon JAJ^{-1}$,
$\DD_A$ has a finite dimensional kernel (see for instance
\cite[Theorem 6.29]{Kato}).

Since according to the dimension, $J^2=\pm 1$, $J$ commutes or
anticommutes with
$\chi$,  $\chi$ commutes with the elements in the algebra $\A$ and
$\DD \chi=-\chi
\DD$ (see \cite{ConnesReality} or \cite[p. 405]{Polaris}), we get
$\DD_A \chi=-\chi
\DD_A$ and $\DD_A J=\pm J \DD_A$ which gives the result.
\end{proof}

\subsection{Pseudodifferential operators}

Let $(A,\DD, \H)$ be a given real regular spectral
triple of dimension $n$.

We note
\begin{align*}
&P_0 \text{ the projection on } \Ker \DD\, ,
P_A \text{  the projection on } \Ker \DD_A\, ,\\
&D:= \DD + P_0\, , D_A:= \DD_A + P_A\, .
\end{align*}
$P_0$ and $P_A$ are thus finite-rank selfadjoint bounded operators.
We remark that
$D$ and $D_A$ are selfadjoint invertible operators with compact
inverses.

\begin{remark}
Since we only need to compute the residues and the value at 0 of the
$\zeta_{D}$,
$\zeta_{D_A}$ functions, it is not necessary to define the operators
$\DD^{-1}$ or
$\DD_A^{-1}$ and the associated zeta functions. However, we can
remark that all the
work presented here could be done using the process of Higson in
\cite{Higson}
which proves that we can add any smoothing operator to $\DD$ or
$\DD_A$ such that
the result is invertible without changing anything to the
computation of residues.
\end{remark}

Define for any $\a\in \R$
\begin{align*}
OP^0&:=\{T \, : \, t\mapsto F_t(T) \in
C^\infty\big(\R,\B(\H)\big)\},\\
OP^\a&:= \set{T \,: \, T |D|^{-\a} \in  OP^0}.
\end{align*}
where $F_t(T):=e^{it|D|}\,T\,e^{-it|D|} =
e^{it|\DD|}\,T\,e^{-it|\DD|}$ since $\vert D\vert=\vert \DD \vert
+P_0$.  Define
\begin{align*}
\delta(T)&:=[|D|,T],\\
\nabla(T)&:=[\DD^2,T],\\
\sigma_s(T)&:=|D|^{s} T |D|^{-s}, \, s\in \C.
\end{align*}
It has been shown in \cite{CM} that
$OP^0=\bigcap_{p\geq 0}\Dom(\delta^p)$.
In particular, $OP^0$ is a subalgebra of
$\B(\H)$ (while elements of $OP^{\a}$
are not necessarily bounded for $\a>0$) and
$\A\subseteq OP^0$, $J\A J^{-1}\subseteq
OP^0$, $[\DD,\A] \subseteq OP^0$. Note that $P_0 \in OP^{-\infty}$
and $\delta(OP^0)\subseteq OP^0$.

For any $t>0$, $\DD^{t}$ and and $|\DD|^t$ are in $OP^t$ and
for any $\a\in \R, D^\a$ and $|D|^{\a}$ are in $OP^\a$.
By hypothesis, $|D|^{-n} \in \L^{(1,\infty)}(\H)$ so
for any $\a>n$, $OP^{-\a}\subseteq \L^1(\H)$.

\begin{lemma} \cite{CM}
    \label{propOP}

(i) For any $T\in OP^0$ and $s\in \C$, $\sigma_s(T)\in OP^0$.

(ii) For any $\a, \beta \in \R$, $OP^{\a}OP^{\beta}\subseteq
OP^{\a+\beta}$.

(iii) If $\a\leq\beta$, $OP^{\a}\subseteq OP^{\beta}$.

(iv) For any $\a$, $\delta(OP^{\a})\subseteq OP^{\a}$.

(v) For any $\a$ and $T\in OP^{\a}$, $\nabla(T) \in OP^{\a+1}$.
\end{lemma}
\begin{proof} See the appendix.
\end{proof}

\begin{remark}
\label{oprem}
Any operator in $OP^\a$, where $\a\in \R$,
extends as a continuous linear operator from $\Dom |D|^{\a+1}$ to
$\Dom |D|$ where the $\Dom |D|^\a$ spaces have their natural norms
(see \cite{CM,Higson}).
\end{remark}

We now introduce a definition of pseudodifferential operators
in a slightly different way
than in \cite{CM,Cgeom,Higson} which in particular pays attention
to the reality operator $J$ and the kernel of $\DD$
and allows $\DD$ and $|D|^{-1}$
to be a pseudodifferential operators. It is more in the spirit of
\cite{CC1}.

\medskip

\begin{definition}
\label{defpseudo} Let us define $\DD(\A)$ as the polynomial algebra
generated
by $\A$, $J\A J^{-1}$, $\DD$ and $|\DD|$.

A pseudodifferential operator is an operator $T$ such that there
exists $d\in \Z$ such that
for any $N\in \N$, there exist $p\in \N_0$, $P\in \DD(\A)$ and $R\in
OP^{-N}$ ($p$, $P$ and
$R$ may depend on $N$) such that $P\,D^{-2p}\in OP^d$ and
$$
T=P\,D^{-2p}+R\, .
$$
Define $\Psi(\A)$ as the set of pseudodifferential operators and
$\Psi(\A)^k:=\Psi(\A)\cap OP^k$.
\end{definition}
Note that if $A$ is a 1-form, $A$ and $JAJ^{-1}$ are in $\DD(\A)$ and
moreover
$\DD(\A)\subseteq \cup_{p \in \N_0} OP^p$. Since $|\DD|\in \DD(\A)$
by construction and
$P_0$ is a pseudodifferential operator, for any
$p\in \Z$, $|D|^{p}$ is a pseudodifferential operator (in $OP^{p}$.)
Let us
remark also that $\DD(\A)\subseteq\Psi(\A) \subseteq \cup_{k\in \Z}
OP^{k}$.

\begin{lemma} \cite{CM,Cgeom}
\label{pdoalg}
The set of all pseudodifferential operators $\Psi(\A)$ is an algebra.
Moreover, if
$T\in \Psi(\A)^d$ and $T\in \Psi(\A)^{d'}$, then $TT'\in
\Psi(\A)^{d+d'}$.
\end{lemma}

\begin{proof} See the appendix.
\end{proof}

Due to the little difference of behavior between scalar and
nonscalar pseudodifferential operators (i.e. when coefficients
like $[\DD, a]$, $a\in \A$ appears in $P$ of Definition
\ref{defpseudo}), it is convenient to also introduce

\begin{definition}
\label{defpseudo1} Let $\DD_{1}(\A)$ be the algebra generated by
$\A$, $J\A J^{-1}$ and $\DD$, and $\Psi_{1}(\A)$ be the set of
pseudodifferential operators
constructed as before with $\DD_{1}(\A)$ instead of $\DD(\A)$. Note
that
$\Psi_{1}(\A)$ is subalgebra of $\Psi(\A)$.
\end{definition}

Remark that $\Psi_1(\A)$ does not
necessarily contain operators such as $|D|^k$ where $k\in \Z$ is odd.
This algebra
is similar to the one defined in \cite{CC1}.

\subsection{Zeta functions and dimension spectrum}

For any operator $B$ and if $X$ is either $D$ or
$D_A$, we define
\begin{align*}
{\zeta}_X^B(s)&:= \Tr\big(B|X|^{-s}\big ),\\
\zeta_X(s)&:= \Tr \big(|X|^{-s}\big).
\end{align*}

\medskip

{\it The dimension spectrum} $Sd(\A,\H,\DD)$ of a
spectral triple has been defined in
\cite{Cgeom, CM}. It is extended here to pay attention
to the operator $J$ and to our
definition of pseudodifferential operator.

\begin{definition}
The spectrum dimension of the spectral triple is
the subset $Sd(\A,\H,\DD)$ of all poles of the functions
$\zeta_D^P := s\mapsto \Tr \big(P |D|^{-s}\big)$
where $P$ is any pseudodifferential operator in $OP^0$.
The spectral triple $(\A,\H,\DD)$ is simple when these
poles are all simple.
\end{definition}

\begin{remark}\label{remark-spectrum} If $Sp(\A,\H,\DD)$ denotes
the set of all poles of
the functions $s\mapsto \Tr \big(P |D|^{-s}\big)$ where $P$
is any pseudodifferential operator, then,
$Sd(\A,\H,\DD) \subseteq Sp(\A,\H,\DD)$.

When $Sp(\A,\H,\DD)=\Z$, $Sd(\A,\H,\DD) = \set{n-k \ : \ k\in \N_0}$:
indeed, if  $P$ is a pseudodifferential operator
in $OP^0$, and $q\in \N$ is such that $q>n$, $P|D|^{-s}$
is in $OP^{-\Re(s)}$ so
is trace-class for $s$ in a neighborhood of $q$; as a
consequence, $q$ cannot be a pole of $s\mapsto \Tr
\big(P|D|^{-s}\big)$.
\end{remark}

\begin{remark} $Sp(\A,\H,\DD)$ is also the set of
poles of functions $s\mapsto \Tr \big(B |D|^{-s-2p}\big)$
where $p\in \N_0$ and $B\in \DD(\A)$.
\end{remark}

\subsection{The noncommutative integral $\ncint$}

We already defined the one parameter group
$\sigma_z(T):=|D|^{z}T|D|^{-z}, \, z\in \C$.

Introducing the notation (recall that $\nabla (T)=[\DD^2,T]$) for an
operator $T$,
$$
\eps(T):=\nabla(T)D^{-2},
$$
we get from \cite[(2.44)]{CC1} the following expansion for
$T\in OP^q$
\begin{equation}\label{one-par}
\sigma_{z}(T)\sim \sum_{r=0}^{N} g(z,r) \,\eps^r(T)  \mod OP^{-N-1+q}
\end{equation}
where $g(z,r):=
\tfrac{1}{r!}(\tfrac{z}{2})\cdots(\tfrac{z}{2}-(r-1))=\genfrac(){0pt}{1}{z/2}{r}$
with the
convention $g(z,0):=1$.

\noindent We define the noncommutative integral by
$$
\ncint T:=\underset{s=0}{\Res}\ \zeta_D^{T}(s)
=\underset{s=0}{\Res}\ \Tr\,\big(T|D|^{-s}\big).
$$

\begin{prop}
\label{tracenc}
\cite{CM}
If the spectral triple is simple, $\ncint$ is a trace on $\Psi(\A)$.
\end{prop}
\begin{proof}
See the appendix.
\end{proof}

\section{Residues of $\zeta_{D_A}$ for a spectral triple with
simple dimension spectrum}

We fix a regular spectral triple $(\A,\H,\DD)$ of
dimension $n$ and a self-adjoint 1-form $A$.

Recall that
\begin{align*}
\DD_A &:=\DD+\wt A \text{  where } \wt A:= A +\eps JAJ^{-1} ,\\
 D_A &:= \DD_A + P_A
\end{align*}
where $P_A$ is the projection on $\Ker \DD_A$. Remark that
$\wt A \in\DD(\A)\cap OP^0$ and $\DD_A\in \DD(\A)\cap OP^1$.

We note
$$
V_A:= P_A - P_0.
$$
As the following lemma shows, $V_A$ is a smoothing operator:

\begin{lemma}
\label{finiterank}
(i) $\bigcap_{k\geq 1} \Dom (\DD_A)^{k} \subseteq \bigcap_{k\geq 1}
\Dom |D|^k$.

(ii) $\Ker \DD_A \subseteq \bigcap_{k\geq 1} \Dom |D|^k$.

(iii) For any $\a, \beta \in \R$, $|D|^\beta P_A |D|^\a$ is bounded.

(iv) $P_A \in OP^{-\infty}$.
\end{lemma}
\begin{proof}
$(i)$ Let us define for any $p\in \N$, $R_p := (\DD_A)^p -\DD^p$, so
$R_p \in OP^{p-1}$ and $R_p \big(\Dom |D|^p\big)\subseteq \Dom |D|$
(see Remark \ref{oprem}).

Let us fix $k\in \N$, $k\geq 2$. Since $\Dom \DD_A = \Dom \DD =\Dom
|D|$, we have
$$
\Dom (\DD_A)^k = \set{\phi \in \Dom |D| \ : \ (\DD^j + R_j)\,\phi
\in \Dom |D| \ , \ \forall j\ \  1\leq j\leq k-1  }.
$$
Let $\phi \in \Dom (\DD_A)^k$. We prove by recurrence that for any
$j\in\set{1,\cdots,k-1}$, $\phi \in \Dom |D|^{j+1}$:

We have $\phi\in \Dom |D|$ and $(\DD + R_1)\, \phi \in \Dom |D|$.
Thus, since $R_1\,\phi \in \Dom |D|$, $\DD \phi \in \Dom |D|$, which
proves that $\phi \in \Dom |D|^2$.
Hence, case $j=1$ is done.

Suppose now that $\phi \in \Dom |D|^{j+1}$ for a $j \in
\set{1,\cdots,k-2}$.
Since $(\DD^{j+1} + R_{j+1})\, \phi \in
\Dom |D|$, and $R_{j+1}\, \phi \in \Dom |D|$, we get $\DD^{j+1}\,
\phi \in \Dom |D|$,
which proves that $\phi\in \Dom |D|^{j+2}$.

Finally, if we set $j=k-1$, we get $\phi \in \Dom |D|^{k}$,
so $\Dom (\DD_A)^k \subseteq \Dom |D|^k$.

$(ii)$ follows from $\Ker \DD_A \subseteq \bigcap_{k\geq 1} \Dom
(\DD_A)^k$ and $(i)$.

$(iii)$ Let us first check that $|D|^\a P_A$ is bounded. We define
$D_0$
as the operator with domain $\Dom D_0 = \Ima P_A \cap \Dom |D|^\a$
and such that $D_0\, \phi = |D|^\a\, \phi.$ Since $\Dom D_0$ is
finite dimensional,
$D_0$ extends as a bounded operator on $\H$ with finite rank.
We have
$$
\sup_{\phi \in \Dom |D|^\a P_A,\ \norm{\phi}\leq 1} \norm{|D|^\a
P_A\, \phi} \leq
\sup_{\phi \in \Dom D_0,\ \norm{\phi}\leq 1}
\norm{|D|^\a\, \phi} = \norm{D_0}<\infty
$$
so $|D|^\a P_A$ is bounded. We can remark that by $(ii)$, $\Dom D_0 =
\Ima P_A$ and
$\Dom |D|^\a P_A = \H$.

Let us prove now that $P_A |D|^\a$ is bounded:
Let $\phi\in \Dom P_A |D|^\a = \Dom |D|^\a$. By $(ii)$, we have $\Ima
P_A \subseteq \Dom |D|^\a$
so we get
\begin{align*}
\norm{P_A |D|^\a\,\phi} & \leq
\sup_{\psi \in \Ima P_A,\ \norm{\psi}\leq 1} |<\psi,|D|^\a\, \phi>|
\leq \sup_{\psi \in
\Ima P_A,\ \norm{\psi}\leq 1} |<|D|^\a\psi,\phi>| \\
&\leq \sup_{\psi \in
\Ima P_A,\ \norm{\psi}\leq 1} \norm{|D|^\a\psi}\norm{\phi} =
\norm{D_0} \norm{\phi}.
\end{align*}

$(iv)$ For any $k\in \N_0$ and $t\in \R$, $\delta^k(P_A)|D|^t$ is a
linear combination of terms of
the form $|D|^\beta P_A |D|^\a$, so the result follows from $(iii)$.
\end{proof}

\begin{remark}
We will see later on the noncommutative torus example how
important is the difference
between $\DD_{A}$ and $\DD+A$. In particular, the inclusion
$\Ker \DD \subseteq \Ker \DD +
A$ is not satisfied since $A$ does not preserve $\Ker \DD$
contrarily to $\wt A$.
\end{remark}

The coefficient of the nonconstant term $\Lambda^k$ ($k>0$)
in the expansion (\ref{formuleaction}) of the spectral action
$S(\DD_A,\Phi,\Lambda)$ is equal to the
residue of $\zeta_{D_A}(s)$ at $k$. We will see in this section
how we can compute these
residues in term of noncommutative integral of certain operators.

Define for any operator $T$, $p\in \N$, $s\in \C$,
$$
K_p(T,s):=(-\tfrac{s}{2})^p\int_{0\leq t_1\leq\cdots\leq t_p\leq 1}
\sigma_{-st_1}(T)\cdots\sigma_{-st_p}(T)\, dt
$$
with $dt:=dt_1\cdots dt_p$.

Remark that if $T\in OP^\a$, then $\sigma_z(T)\in OP^\a$ for $z\in \C$
and $K_p(T,s) \in OP^{\a p}$.

Let us define
\begin{align*}
X &:= \DD_{A}^2-\DD^2 =\wt A \DD + \DD \wt A + \wt A^2 ,\\
X_V &:= X+V_A,
\end{align*}
thus $X\in \DD_1(\A)\cap OP^1$ and by Lemma \ref{finiterank},
\begin{equation}\label{xvsim}
X_V \sim X \mod OP^{-\infty}.
\end{equation}

We will use
$$
Y:=\log(D_A^2) -\log (D^2)
$$
which makes sense since $D_A^2 = \DD_A^2 + P_A$ is invertible for any
$A$.

By definition of $X_V$, we get
$$
Y= \log (D^2 + X_V) -\log (D^2).
$$

\begin{lemma}
    \label{2dev}
\cite{CC1}

(i) $Y$ is a pseudodifferential operator in $OP^{-1}$ with the
following
expansion for any $N\in\N$
$$
Y \sim \sum_{p=1}^N\sum_{k_1,\cdots,k_p
=0}^{N-p}\tfrac{(-1)^{|k|_1+p+1}}{|k|_1+p}
\nabla^{k_p}(X\nabla^{k_{p-1}}(\cdots
X\nabla^{k_1}(X)\cdots)) D^{-2(|k|_1+p)} \mod OP^{-N-1}.
$$

(ii) For any $N\in\N$ and $s\in \C$,
\begin{align}
\label{expansion} |D_A|^{-s} \sim |D|^{-s} + \sum_{p=1}^N K_p(Y,s)
|D|^{-s} \mod
OP^{-N-1-\Re(s)}.
\end{align}
\end{lemma}
\begin{proof}
$(i)$ We follow \cite[Lemma 2.2]{CC1}. By functional calculus,
$Y=\int_0^\infty I(\la)\, d\la$, where
$$
I(\la)\sim\sum_{p=1}^N(-1)^{p+1}\big((D^2+\la)^{-1}X_V\big)^{p}
(D^2+\la)^{-1} \mod OP^{-N-3}.
$$
By (\ref{xvsim}), $\big((D^2+\la)^{-1}X_V\big)^{p} \sim
\big((D^2+\la)^{-1}X\big)^{p} \mod OP^{-\infty}$ and we get
$$
I(\la)\sim\sum_{p=1}^N(-1)^{p+1}\big((D^2+\la)^{-1}X\big)^{p}
(D^2+\la)^{-1} \mod OP^{-N-3}.
$$
We set $A_p(X):=\big((D^2+\la)^{-1}X\big)^{p}(D^2+\la)^{-1}$
and $L:=(D^2+\la)^{-1}\in OP^{-2}$ for a fixed $\la$.
Since $[D^2 + \la,X]\sim \nabla(X) \mod OP^{-\infty}$,
a recurrence proves that if $T$ is an operator
in $OP^{r}$, then, for $q\in \N_0$,
$$
A_1(T)=L T L \sim \sum_{k=0}^q (-1)^k\nabla^k(T) L^{k+2} \mod
OP^{r-q-5}.
$$
With $A_p(X)=LX A_{p-1}(X)$,
another recurrence gives, for any $q\in \N_0$,
$$
A_p(X)\sim \sum_{k_1,\cdots,k_p =0}^q (-1)^{|k|_1}\nabla^{k_p}
(X \nabla^{k_{p-1}}(\cdots X\nabla^{k_1}(X)\cdots)) L^{|k|_1+p+1}
\mod OP^{-q-p-3},
$$
which entails that
$$
I(\la)\sim\sum_{p=1}^N(-1)^{p+1}\sum_{k_1,\cdots,k_p
=0}^{N-p}(-1)^{|k|_1}\nabla^{k_p}(X\nabla^{k_{p-1}}
(\cdots X\nabla^{k_1}(X)\cdots))
L^{|k|_1+p+1} \mod OP^{-N-3}.
$$

With $\int_{0}^\infty (D^2+\la)^{-(|k|_1+p+1)}d\la =
\tfrac{1}{|k|_1+p}
D^{-2(|k|_1+p)}$, we get the result provided we control
the remainders. Such a control is given in \cite[(2.27)]{CC1}.

$(ii)$ We have $|D_A|^{-s}=e^{B-(s/2)Y}e^{-B}\, |D|^{-s}$
where $B:= (-s/2)\log(D^2)$. Following \cite[Theorem 2.4]{CC1},
we get
\begin{equation}
\label{egalite-DAs}
|D_A|^{-s} = |D|^{-s} + \sum_{p=1}^\infty K_p(Y,s)|D|^{-s}\, .
\end{equation}
and each $K_p(Y,s)$ is in $OP^{-p}$.
\end{proof}

\begin{corollary}
\label{eps-pdo} For any $p\in\N$ and $r_1,\cdots,r_p \in \N_0$,
$\eps^{r_1}(Y)\cdots \eps^{r_p}(Y) \in \Psi_{1}(\A)$.
\end{corollary}
\begin{proof}
If for any $q\in \N$ and $k=(k_1,\cdots,k_q)\in \N_0^q$,
$$
\Ga_q^k(X):=\tfrac{(-1)^{|k|_1+q+1}}{|k|_1+q} \nabla^{k_q}
(X\nabla^{k_{q-1}}(\cdots X\nabla^{k_1}(X)\cdots)),
$$
then, $\Ga_q^k(X) \in OP^{|k|_1+q}$. For any $N\in\N$,
\begin{equation}\label{Ydev}
Y \sim \sum_{q=1}^N\sum_{k_1,\cdots,k_q =0}^{N-q} \Ga_q^k(X)
D^{-2(|k|_1+q)} \mod OP^{-N-1}.
\end{equation}
Note that the $\Ga_q^k(X)$ are in $\DD_{1}(\A)$, which, with
(\ref{Ydev}) proves that $Y$
and thus $\eps^r(Y)= \nabla^r(Y)D^{-2r}$, are also in $\Psi_{1}(\A)$.
\end{proof}

We remark, as in \cite{Cours}, that the fluctuations leave invariant
the first term of the
spectral action (\ref{formuleaction}). This is a generalization of
the fact that in the commutative case, the noncommutative integral
depends only on the principal symbol of the
Dirac operator $\DD$ and this symbol is stable by adding a gauge
potential like in $\DD+A$.
Note however that the symmetrized gauge potential
$A+\epsilon JAJ^{-1}$ is always zero in
this case for any selfadjoint one-form $A$.

\begin{lemma} If the spectral triple is simple,
formula (\ref{constant}) can be extended as
\begin{align}
\zeta_{D_{A}}(0)-\zeta_{D}(0)=\sum_{q=1}^{n} \tfrac{(-1)^{q}}{q}
\ncint (\wt AD^{-1})^{q}. \label{termconstanttilde}
\end{align}
\end{lemma}
\begin{proof}
Since the spectral triple is simple, equation (\ref{egalite-DAs})
entails that
$$
\zeta_{D_A}(0)-\zeta_{D}(0) = \Tr (K_1(Y,s)|D|^{-s})_{|s=0} \, .
$$
Thus, with (\ref{one-par}), we get
$\zeta_{D_A}(0)-\zeta_{D}(0) = -\half \ncint Y$. Replacing $A$ by
$\wt A$, the same proof as in \cite{CC1} gives
\begin{align*}
-\half \ncint Y = \sum_{q=1}^{n} \tfrac{(-1)^{q}}{q} \ncint (\wt
AD^{-1})^{q}.
\tag*{\qed}
\end{align*}
\hideqed
\end{proof}

\begin{lemma}
    \label{Res-zeta-n-k}
For any $k\in \N_0$,
$$
\underset{s=n-k}{\Res} \, \zeta_{D_A}(s)=
\underset{s=n-k}{\Res} \,\zeta_{D}(s) +
\sum_{p=1}^k \sum_{r_1,\cdots, r_p =0}^{k-p}
\underset{s=n-k}{\Res} \, h(s,r,p) \, \Tr\big(\eps^{r_1}(Y)
\cdots\eps^{r_p}(Y) |D|^{-s}\big),
$$
where
$$
h(s,r,p):=(-s/2)^p\int_{0\leq t_1\leq \cdots \leq t_p\leq 1}
g(-st_1,r_1)\cdots
g(-st_p,r_p) \, dt\, .
$$
\end{lemma}
\begin{proof}
By Lemma \ref{2dev} $(ii)$, $|D_A|^{-s} \sim |D|^{-s} + \sum_{p=1}^k
K_p(Y,s)
|D|^{-s} \mod OP^{-(k+1)-\Re(s)}$, where the convention
$\sum_{\emptyset}=0$ is used.
Thus, we get for $s$ in a neighborhood of $n-k$,
$$
|D_A|^{-s}-|D|^{-s} - \sum_{p=1}^k K_p(Y,s) |D|^{-s} \in
OP^{-(k+1)-\Re(s)}\subseteq
\L^1(\H)
$$
which gives
\begin{equation}
    \label{res-n-k-interm}
\underset{s=n-k}{\Res} \, \zeta_{D_A}(s)= \underset{s=n-k}{\Res}
\,\zeta_{D}(s) + \sum_{p=1}^k \underset{s=n-k}{\Res} \,\Tr
\big(K_p(Y,s) |D|^{-s}\big).
\end{equation}
Let us fix $1\leq p\leq k$ and $N\in \N$. By (\ref{one-par}) we get
\begin{align}
    \label{K_p}
K_p(Y,s)\sim (-\tfrac s2)^p \int_{0\leq t_1\leq \cdots t_p \leq 1}
\sum_{r_1,\cdots,r_p =0}^N g(-st_1,r_1)&\cdots g(-st_p,r_p)
\nonumber\\
&\eps^{r_1}(Y)\cdots \eps^{r_p}(Y)\, dt \, \mod OP^{-N-p-1}.
\end{align}

If we now take $N=k-p$, we get for $s$ in a neighborhood of $n-k$
$$
K_p(Y,s)|D|^{-s} - \sum_{r_1,\cdots,r_p
=0}^{k-p}h(s,r,p)\,\eps^{r_1}(Y)\cdots
\eps^{r_p}(Y)|D|^{-s} \in OP^{-k-1-\Re(s)} \subseteq \L^1(\H)
$$
so (\ref{res-n-k-interm}) gives the result.
\end{proof}

Our operators $|D_A|^k$ are pseudodifferential operators:
\begin{lemma}
For any $k\in \Z$, $\vert D_{A} \vert^k \in \Psi^k(\A)$.
\end{lemma}
\begin{proof}
Using (\ref{K_p}), we see that $K_p(Y,s)$ is a pseudodifferential
operator in $OP^{-p}$, so (\ref{expansion}) proves that $|D_A|^k$
is a pseudodifferential operator in $OP^k$.
\end{proof}

The following result is quite important since it shows that one
can use $\ncint$ for $D$ or $D_{A}$:
\begin{prop}
\label{ncintfluctuated}
If the spectral triple is simple,
$\underset{s=0}{\Res} \, \Tr \big(P |D_A|^{-s}\big) = \ncint P$
for any pseudodifferential operator $P$. In particular, for any $k\in
\N_0$
$$
\ncint |D_A|^{-(n-k)}=\underset{s=n-k}{\Res} \,\zeta_{D_A}(s) .
$$
\end{prop}
\begin{proof}
Suppose $P\in OP^{k}$ with $k\in \Z$ and let us fix $p\geq 1$.
With (\ref{K_p}), we see that for any $N\in \N$,
$$
PK_p(Y,s)|D|^{-s}\sim \sum_{r_1,\cdots,r_p =0}^N h(s,r,p) \,
P\eps^{r_1}(Y)\cdots \eps^{r_p}(Y)|D|^{-s} \mod OP^{-N-p-1+k-\Re(s)}.
$$
Thus if we take $N=n-p+k$, we get
$$
 \underset{s=0}{\Res} \,\Tr \big(P K_p (Y,s) |D|^{-s}\big) =
 \sum_{r_1,\cdots,r_p =0}^{n-p+k} \underset{s=0}{\Res} \,\,
h(s,r,p) \, \Tr \big(P\eps^{r_1}(Y)\cdots \eps^{r_p}(Y) |D|^{-s}\big).
$$
Since $s=0$ is a zero of the analytic function $s\mapsto h(s,r,p)$
and $s\mapsto \Tr
P\eps^{r_1}(Y)\cdots \eps^{r_p}(Y)|D|^{-s}$ has only simple poles
by hypothesis, we see
that $\underset{s=0}{\Res} \, h(s,r,p) \, \Tr
\big(P\eps^{r_1}(Y)\cdots \eps^{r_p}(Y)
|D|^{-s}\big)=0$ and
\begin{equation}
\label{res0K_p}
\underset{s=0}{\Res} \, \Tr \big(P K_p (Y,s) |D|^{-s}\big)=0.
\end{equation}
Using (\ref{expansion}), $ P|D_A|^{-s} \sim P |D|^{-s} +
\sum_{p=1}^{k+n} PK_p(Y,s) |D|^{-s} \mod OP^{-n-1-\Re(s)} $ and thus,
\begin{equation}\label{res0PD_A}
\underset{s=0}{\Res} \, \Tr (P|D_A|^{-s}) =\ncint P +
\sum_{p=1}^{k+n} \,
\underset{s=0}{\Res} \, \Tr \big(PK_p(Y,s) |D|^{-s}\big).
\end{equation}
The result now follows from (\ref{res0K_p}) and (\ref{res0PD_A}).
To get the last equality,
one uses the pseudodifferential operator $|D_A|^{-(n-k)}$.
\end{proof}

\begin{prop}
\label{invariance1} If the spectral triple is simple, then
\begin{align}
\ncint {|D_{A}|}^{-n}&=\ncint |D|^{-n}.
\end{align}
\end{prop}
\begin{proof} Lemma \ref{Res-zeta-n-k} and previous proposition for
$k=0$.
\end{proof}

\begin{lemma}
    \label{residus-particuliers}
If the spectral triple is simple,
\begin{align*}
\hspace{-2.2cm} (i) & \quad \ncint |D_A|^{-(n-1)}= \ncint |D|^{-(n-1)}
-(\tfrac{n-1}{2})\ncint X|D|^{-n-1}.\\
\hspace{-2.2cm} (ii)  & \quad \ncint |D_A|^{-(n-2)}= \ncint
|D|^{-(n-2)}+\tfrac{n-2}{2}\big(-\ncint X|D|^{-n} + \tfrac{n}{4}
\ncint X^2
|D|^{-2-n} \big).
\end{align*}
\end{lemma}

\begin{proof}
$(i)$ By (\ref{expansion}),
$$
\underset{s=n-1}{\Res} \, \zeta_{D_A}(s) -\zeta_{D}(s)=
\underset{s=n-1}{\Res} \,(-s/2)
\Tr \big(Y |D|^{-s}\big) = -\tfrac{n-1}{2} \, \underset{s=0}{\Res} \,
\Tr
\big(Y|D|^{-(n-1)}|D|^{-s} \big)
$$
where for the last equality we use the simple dimension spectrum
hypothesis. Lemma \ref{2dev} $(i)$ yields
$Y\sim XD^{-2} \mod OP^{-2}$ and $Y|D|^{-(n-1)}\sim
X|D|^{-n-1} \mod OP^{-n-1}\subseteq \L^1(\H)$. Thus,
$$
\underset{s=0}{\Res} \, \Tr \big( Y|D|^{-(n-1)}|D|^{-s}\big) =
\underset{s=0}{\Res} \,
\Tr \big(X|D|^{-n-1} |D|^{-s}\big) = \ncint  X |D|^{-n-1}.
$$
$(ii)$ Lemma \ref{Res-zeta-n-k} $(ii)$ gives
$$
\underset{s=n-2}{\Res} \, \zeta_{D_A}(s) =
\underset{s=n-2}{\Res} \, \zeta_{D}(s) +
\underset{s=n-2}{\Res} \, \sum_{r=0}^1 h(s,r,1) \,
\Tr \big(\eps^r(Y)|D|^{-s}\big) +
h(s,0,2) \,\Tr \big(Y^2 |D|^{-s}\big).
$$
We have $h(s,0,1)=-\tfrac s2$, $h(s,1,1)=
\half(\tfrac s2)^2$ and $h(s,0,2)= \half (\tfrac
s2)^2$. Using again Lemma \ref{2dev} $(i)$,
$$
Y\sim XD^{-2}-\half \nabla(X)D^{-4} -\half X^2 D^{-4} \mod OP^{-3}.
$$
Thus,
$$
\underset{s=n-2}{\Res} \, \Tr \big(Y|D|^{-s}\big) =
\ncint X|D|^{-n} -\half \ncint
(\nabla(X)+X^2)|D|^{-2-n}.
$$
Moreover, using $\ncint \nabla(X)|D|^{-k}=0$ for any $k\geq 0$
since $\ncint$ is a trace,
$$
\underset{s=n-2}{\Res} \, \Tr \big(\eps(Y)|D|^{-s}\big) =
\underset{s=n-2}{\Res} \, \Tr
\big(\nabla(X)D^{-4}|D|^{-s}\big) = \ncint \nabla(X)|D|^{-2-n}=0.
$$
Similarly, since $Y\sim XD^{-2}$ mod $OP^{-2}$ and
$Y^2\sim X^2D^{-4} \mod OP^{-3}$, we get
$$
\underset{s=n-2}{\Res} \, \Tr \big(Y^2 |D|^{-s}\big) =
\underset{s=n-2}{\Res} \, \Tr
\big(X^2D^{-4}|D|^{-s}\big) = \ncint X^2 |D|^{-2-n}.
$$
Thus,
\begin{align*}
\underset{s=n-2}{\Res} \, \zeta_{D_A}(s) =
\underset{s=n-2}{\Res} \,\zeta_{D}(s) +
&(-\tfrac {n-2}{2})(\ncint X|D|^{-n}
-\half \ncint (\nabla(X)+X^2)|D|^{-2-n})\\
& \quad +\half(\tfrac {n-2}{2})^2
\ncint \nabla(X)|D|^{-2-n}+\half(\tfrac {n-2}{2})^2
\ncint X^2 |D|^{-2-n}.
\end{align*}
Finally,
$$
\underset{s=n-2}{\Res} \, \zeta_{D_A}(s) =
\underset{s=n-2}{\Res} \,\zeta_{D}(s) +
(-\tfrac {n-2}{2}) \big(\ncint X|D|^{-n} -
\half \ncint X^2|D|^{-2-n}\big)+\half(\tfrac
{n-2}{2})^2 \ncint X^2 |D|^{-2-n}
$$
and the result follows from Proposition \ref{ncintfluctuated}.
\end{proof}

\begin{corollary}
    \label{res-n-2-A}
If the spectral triple is simple and satisfies $\ncint |D|^{-(n-2)}=
\ncint \wt A \DD |D|^{-n} = \ncint  \DD \wt A |D|^{-n}=0$, then
$$
 \ncint |D_A|^{-(n-2)} = \tfrac{n(n-2)}{4}
 \big(\ncint \wt A \DD \wt A \DD
|D|^{-n-2}+\tfrac{n-2}{n}\ncint \wt A ^2|D|^{-n}\big).
$$
\end{corollary}
\begin{proof}
By previous lemma,
$$
\underset{s=n-2}{\Res} \, \zeta_{D_A}(s) =
\tfrac{n-2}{2}\big( -\ncint \wt A^2 |D|^{-n}
+\tfrac{n}{4}\ncint (\,\wt A\DD \wt A \DD+ \DD \wt A \DD \wt A+
\wt A \DD^2 \wt A + \DD \wt A^2 \DD \,) |D|^{-n-2} \big).
$$
Since $\nabla(\wt A) \in OP^1$, the trace property of $\ncint$ yields
the result.
\end{proof}

\section{The noncommutative torus}
\subsection{Notations}

Let $\Coo(\T^n_\Th)$ be the smooth noncommutative
$n$-torus associated to a non-zero skew-symmetric deformation matrix
$\Th \in
M_n(\R)$ (see \cite{ConnesTorus}, \cite{RieffelRot}). This means that
$\Coo(\T^n_\Th)$ is the algebra generated by $n$
unitaries $u_i$, $i=1,\dots,n$ subject to the relations
\begin{equation}
\label{rel}
u_i\,u_j=e^{i\Th_{ij}}\,u_j\,u_i,
\end{equation}
and with
Schwartz coefficients: an
element $a\in\Coo(\T_\Th^n)$ can be written as
$a=\sum_{k\in\Z^n}a_k\,U_k$, where $\{a_k\}\in\SS(\Z^n)$ with
the Weyl elements defined by $U_k:=e^{-\frac i2 k.\chi
k}\,u_1^{k_1}\cdots
u_n^{k_n}$, $k\in\Z^n$, relation
\eqref{rel} reads
\begin{equation}
\label{rel1}
U_{k}U_{q}=e^{-\frac i2 k.\Theta q} \,U_{k+q}, \text{ and }
U_{k}U_{q}=e^{-i k.\Theta q} \,U_{q}U_{k}
\end{equation}
where $\chi$ is
the matrix restriction of $\Theta$ to its upper triangular part.
Thus unitary operators $U_{k}$ satisfy $U_{k}^*=U_{-k}$ and
$[U_{k},U_{l}]=-2i\sin(\frac 12 k.\Th l)\,U_{k+l}$.

Let $\tau$ be the trace on $\Coo(\T^n_\Th)$ defined by
$\tau\big( \sum_{k\in\Z^n}a_k\,U_k \big):=a_0$
and $\H_{\tau}$ be the GNS Hilbert space obtained
by completion of $ \Coo(\T_\Th^n)$
with respect of the norm induced by the scalar product
$\langle a,b\rangle:=\tau(a^*b)$.
On $\H_{\tau}=\set{\sum_{k\in\Z^n}a_k\,U_k \, : \, \{a_{k}\}_{k} \in
l^2(\Z^n) }$, we consider the left and right regular
representations of
$\Coo(\T_\Th^n)$ by bounded operators, that we denote respectively
by $L(.)$ and $R(.)$.

Let also $\delta_\mu$, $\mu\in \set{1,\dots,n}$, be the $n$ (pairwise
commuting)
canonical derivations, defined by
\begin{equation}
\delta_\mu(U_k):=ik_\mu U_k. \label{dUk}
\end{equation}

We need to fix notations: let $\A_{\Th}:=C^{\infty}(\T_{\Th}^n)$
acting on $\H:=\H_{\tau}\otimes \C^{2^m}$ with $n=2m$ or $n=2m+1$
(i.e., $m=\lfloor \tfrac n2 \rfloor$ is the integer part of $\tfrac
n2$),
the square integrable sections of the trivial spin bundle over $\T^n$.

Each element of $\A_{\Th}$ is represented on $\H$ as
$L(a)\otimes1_{2^m}$ where $L$ (resp. $R$) is the left (resp. right)
multiplication. The Tomita conjugation $J_{0}(a):=a^*$
satisfies $[J_{0},\delta_{\mu}]=0$ and we define
$J:=J_{0}\otimes C_{0}$
where $C_{0}$ is an operator on $\C^{2^m}$.
The Dirac operator is given by
\begin{align}
\label{defDirac}
\DD:=-i\,\delta_{\mu}\otimes \gamma^{\mu},
\end{align}
where we use hermitian Dirac matrices $\gamma$. It is defined and
symmetric on the
dense subset of $\H$ given by $C^{\infty}(\T_{\Th}^n) \otimes
\C^{2^{m}}$. We still note $\DD$ its selfadjoint extension. This
implies
\begin{align}
    \label{CGamma}
C_{0}\ga^{\alpha}=-\eps \ga^\alpha C_{0},
\end{align}
and
$$
\DD\ U_k \otimes e_i = k_\mu U_k \otimes \gamma^\mu e_i ,
$$
where $(e_i)$ is the canonical basis of $\C^{2^m}$. Moreover,
$C_{0}^2=\pm 1_{2^m}$
depending on the parity of $m$. Finally, one introduces the chirality
(which in the
even case is $\chi:=id \otimes (-i)^{m} \gamma^1 \cdots \gamma^{n}$)
and this yields
that $(\A_{\Th},\H,\DD,J,\chi)$ satisfies all axioms of a spectral
triple, see
\cite{Book,Polaris}.

The perturbed Dirac operator $V_{u}\DD V_{u}^*$ by
the unitary
$$
V_{u}:=\big(L(u)\otimes 1_{2^m}\big)J\big(L(u)\otimes
1_{2^m}\big)J^{-1},
$$
defined for every unitary $u \in \A$,
$uu^{*}=u^{*}u=U_{0}$,
must satisfy condition (\ref{Jcom}) (which is equivalent
to $\H$ being endowed with a structure of $\A_{\Th}$-bimodule). This
yields the necessity of a symmetrized covariant Dirac operator:
$$
\DD_{A}:=\DD + A + \epsilon J\,A\,J^{-1}$$
since
$V_{u}\DD V_{u}^{*}=\DD_{L(u)\otimes 1_{2^m}[\DD,L(u^{*})
\otimes 1_{2^m}]}$:
in fact, for $a \in \A_{\Th}$, using $J_{0}L(a)J_{0}^{-1}=R(a^*)$,
we get $$\epsilon J\big(L(a)\otimes
\gamma^{\alpha}\big)J^{-1}=-R(a^*)\otimes \gamma^{\alpha}$$
and that the representation $L$ and the
anti-representation $R$ are $\C$-linear, commute and satisfy
$$
[\delta_{\alpha},L(a)]=L(\delta_{\alpha}a),\quad
[\delta_{\alpha},R(a)]=R(\delta_{\alpha}a).
$$
This induces some covariance property for the Dirac operator:
one checks that for all $k \in \Z^{n}$,
\begin{align}
\label{puregauge1}
L(U_{k})\otimes 1_{2^m}[\DD,L(U_{k}^{*})\otimes 1_{2^m}]&=1\otimes
(-k_{\mu}\ga^{\mu}),
\end{align}
so with (\ref{CGamma}), we get $U_{k}[\DD,U_{k}^{*}]+\epsilon
JU_{k}[\DD,U_{k}^{*}]J^{-1}=0$ and
\begin{align}
\label{covariance}
V_{U_{k}} \,\DD \, V_{U_{k}}^{*}=
\DD=\DD_{L(U_{k})\otimes 1_{2^m}[\DD,L(U_{k}^{*})\otimes 1_{2^m}]}.
\end{align}
Moreover, we get the gauge transformation:
\begin{align}
\label{gaugeDirac}
V_{u} \DD_{A} V_{u}^{*}= \DD_{\gamma_{u}(A)}
\end{align}
where the gauged transform one-form of $A$ is
\begin{align}
\label{gaugetransform}
\gamma_{u}(A):=u[\DD,u^{*}]+uAu^{*},
\end{align}
with the shorthand
$L(u)\otimes 1_{2^m} \longrightarrow u$.

As a consequence, the spectral action is gauge invariant:
$$
\SS(\DD_{A},\Phi,\Lambda)=\SS(\DD_{\gamma_{u}(A)},\Phi,\Lambda).
$$

An arbitrary selfadjoint one-form $A$, can be written as
\begin{equation}
\label{connection}
A = L(-iA_{\alpha})\otimes\gamma^{\alpha},\,\, A_{\alpha}
=-A_{\alpha}^* \in
\A_{\Th},
\end{equation}
thus
\begin{equation}
\label{dirac}
\DD_{A}=-i\,\big(\delta_{\alpha}+L(A_{\alpha})-R(A_{\alpha})\big)
 \otimes \gamma^{\alpha}.
\end{equation}
Defining $$\tilde A_{\alpha}:=L(A_{\alpha})-R(A_{\alpha}),$$
we get
$\DD_{A}^2=-g^{{\alpha}_{1} {\alpha}_{2}}(\delta_{{\alpha}_{1}}+\tilde
A_{{\alpha}_{1}})(\delta_{{\alpha}_{2}}+\tilde
A_{{\alpha}_{2}})\otimes 1_{2^m} - \tfrac 12
\Omega_{{\alpha}_{1} {\alpha}_{2}}\otimes \gamma^{{\alpha}_{1}
{\alpha}_{2}}
$
where
\begin{align*}
\gamma^{{\alpha}_{1} {\alpha}_{2}}
&:=\tfrac 12(\gamma^{{\alpha}_{1}}\gamma^{{\alpha}_{2}}
-\gamma^{{\alpha}_{2}}\gamma^{{\alpha}_{1}}) ,\\
\Omega_{{\alpha}_{1} {\alpha}_{2}}
&:=[\delta_{{\alpha}_{1}}+
\tilde A_{{\alpha}_{1}},\delta_{{\alpha}_{2}}
+\tilde A_{{\alpha}_{2}}]\,
=L(F_{{\alpha}_{1} {\alpha}_{2}}) - R(F_{{\alpha}_{1} {\alpha}_{2}})
\end{align*}
with
\begin{align}
    \label{Fmunu}
F_{{\alpha}_{1} {\alpha}_{2}}:=\delta_{{\alpha}_{1}}(A_{{\alpha}_{2}})
-\delta_{{\alpha}_{2}}(A_{{\alpha}_{1}})+[A_{{\alpha}_{1}},A_{{\alpha}_{2}}].
\end{align}
In summary,
\begin{align}
\label{D2}
\DD_{A}^2=-\delta^{{\alpha}_{1} {\alpha}_{2}}
\Big(
\delta_{{\alpha}_{1}}+L(A_{{\alpha}_{1}})-R(A_{{\alpha}_{1}})\Big)
\Big(\delta_{{\alpha}_{2}}+L(A_{{\alpha}_{2}})-R(A_{{\alpha}_{2}})\Big)
\otimes 1_{2^m}
\nonumber\\
-\tfrac 12\,\big(L(F_{{\alpha}_{1} {\alpha}_{2}}) - R(F_{{\alpha}_{1}
{\alpha}_{2}})\big)
\otimes \gamma^{{\alpha}_{1} {\alpha}_{2}}.
\end{align}

\subsection{Kernels and dimension spectrum}
We now compute the kernel of the perturbed Dirac operator:
\begin{prop}
    \label{noyaux}
(i) $\Ker \DD=U_0\otimes \C^{2^m}$, so $\dim \Ker \DD =
2^m$.

(ii) For any selfadjoint one-form $A$, $\Ker \DD \subseteq \Ker
\DD_A$.

(iii) For any unitary $ u\in \A$, $\Ker \DD_{\gamma_{u}(A)}=V_{u}\,
\Ker \DD_{A}$.
\end{prop}
\begin{proof}
$(i)$ Let $\psi =\sum_{k,j} c_{k,j} \, U_k \otimes e_j \in \Ker \DD$.
Thus, $0=\DD^2
\psi =\sum_{k,i} c_{k,j} |k|^2\, U_k \otimes e_j$ which entails that
$c_{k,j}|k|^2=0$
for any $k \in \Z^n$ and $1\leq j\leq 2^m$. The result follows.

$(ii)$ Let $\psi \in \Ker \DD$. So, $\psi = U_0 \otimes v$
with $v\in \C^{2^m}$ and from (\ref{dirac}), we get
\begin{align*}
\DD_A \psi &= \DD \psi + (A+\epsilon J AJ^{-1})\psi = (A+\epsilon J
AJ^{-1})\psi=
-i[A_\a,U_0]\otimes \ga^\a v = 0
\end{align*}
since $U_{0}$ is the unit of the algebra, which proves that $\psi \in
\Ker \DD_A$.

$(iii)$ This is a direct consequence of (\ref{gaugeDirac}).
\end{proof}

\begin{corollary}
Let $A$ be a selfadjoint one-form. Then
$\Ker \DD_A=\Ker \DD$ in the following cases:

 (i) $A_{u}:=L(u)\otimes
1_{2^m}[\DD,L(u^*)\otimes 1_{2^m}]$ when $u$ is a unitary in $\A$.

 (ii) $\vert \vert A \vert \vert <\tfrac12$.

 (iii) The matrix $\tfrac{1}{2\pi}\Th$ has only integral coefficients.
 \end{corollary}

\begin{proof}
$(i)$ This follows from previous result because
$V_{u} (U_{0}\otimes v)= U_{0} \otimes v$ for any $v\in \C^{2^m}$.

$(ii)$ Let $\psi=\sum_{k,j}c_{k,j}\, U_{k}\otimes e_{j} $ be in $ \Ker
\DD_{A}$ (so $\sum_{k,j} \vert c_{k,j}\vert^2<
\infty$) and $\phi:=\sum_{j}c_{0,j}\, U_{0}\otimes e_{j}$. Thus
$\psi':=\psi-\phi \in \text{Ker }\DD_{A}$ since $\phi \in \Ker \DD
\subseteq \Ker \DD_A$ and
$$
\vert \vert \sum_{0\neq k \in \Z^n,\,j}
c_{k,j}\,k_{\alpha}\,U_{k}\otimes
\gamma^{\alpha}e_{j}\vert \vert^2=\vert \vert \DD
\psi'\vert\vert^2=\vert \vert -(A + \epsilon
JAJ^{-1})\psi'\vert \vert^2 \leq 4\vert\vert A \vert \vert^{2}\vert
\vert \psi' \vert
\vert^{2} <\vert \vert \psi' \vert \vert^{2}.
$$
Defining $X_{k}:=\sum_{\alpha}k_{\alpha}\gamma_{\alpha}$,
$X_{k}^{2}=\sum_{\alpha}\vert k_{\alpha}\vert^{2}\, 1_{2^{m}}$ is
invertible and the vectors $\set{U_{k}\otimes X_{k}e_{j}}_{0\neq k\in
\Z^{n},\,j}$ are orthogonal in $\H$, so
$$
\sum_{0\neq k\in \Z^{n},\,j}\big( \sum_{\alpha} \vert k_{\alpha}\vert
^{2} \big)\, \vert
c_{k,j}\vert^{2} < \sum_{0\neq k\in \Z^{n},\,j}\vert c_{k,j}\vert^{2}
$$
which is possible only if $c_{k,j}=0, \, \forall k,\,j$ that is
$\psi'=0$ et
$\psi=\phi \in \text{Ker }\DD$.

$(iii)$ This is a consequence of the fact that the algebra is
commutative, thus $A+\epsilon JAJ^{-1}=0$.
\end{proof}

Note that if $\wt A_{u}:=A_{u}+\epsilon JA_{u}J^{-1}$, then by
(\ref{puregauge1}), $\wt A_{U_{k}}=0$ for all
$k \in \Z^n$ and $\norm{A_{U_{k}}}=\vert k\vert$,
but for an arbitrary unitary $u\in \A$, $\wt A_{u}\ne 0$ so
$\DD_{A_{u}}\ne \DD$.

Naturally the above result is also a direct consequence of the fact
that the eigenspace of an isolated eigenvalue of an operator is not
modified by small perturbations. However, it is interesting to
compute the last result directly to emphasize the difficulty of the
general case:

Let $\psi=\sum_{l\in \Z^n, 1\leq j \leq 2^m}c_{l,j}\, U_{l}\otimes
e_{j}\in \Ker
\DD_A$, so $\sum_{l\in \Z^n, 1\leq j \leq 2^m} \vert c_{l,j}\vert^2
<Ê \infty$. We
have to show that $\psi\in$ Ker $\DD$ that is $c_{l,j}=0$ when $l\ne
0$.

Taking the scalar product of $\langle U_{k} \otimes e_{i}\vert$ with
$$
0=\DD_{A}\psi=\sum_{l,\,\a,\,j} c_{l,\,j}(l^{\a}U_{l}-i[A_{\a},U_{l}]
)\otimes
\gamma^{\a}e_{j},
$$
we obtain
$$
0=\sum_{l,\,\a,\,j} c_{l,\,j} \big(l^{\a}\delta_{k,l}-i\langle
U_{k},[A_{\a},U_{l}]\rangle \big)\langle e_{i},\gamma^{\a}e_{j}
\rangle.
$$
If $A_{\a}=\sum_{\a,l}a_{\a,l}\, U_{l} \otimes \gamma^{\a}$ with
$\set{a_{\a,l}}_{l}
\in \SS(\Z^n)$, note that $[U_{l},U_{m}]=-2i \sin(\tfrac 12 l.\Th m)
\, U_{l+m}$ and
$$
\langle U_{k},[A_{\a},U_{l}]\rangle = \sum_{l'\in
\Z^{n}}a_{\a,l'}(-2i \sin (\tfrac 12
l'.\Th l) \langle U_{k}, U_{l'+l}\rangle=-2i\, a_{\a,k-l} \,
\sin(\tfrac 12 k.\Th l).
$$
Thus
\begin{align}
    \label{contraintenoyau}
0=\sum_{l\in \Z^{n}}\sum_{\a=1}^{n}\sum_{j=1}^{2^{m}} c_{l,\,j}
\big(l^{\a}\delta_{k,l} -2a_{\a,k-l} \, \sin(\tfrac 12 k.\Th l) \big)
\, \langle
e_{i},\gamma^{\a}e_{j} \rangle, \quad \forall k\in \Z^n, \, \forall
i, 1\leq i \leq
2^{m}.
\end{align}

\medskip

{\it We conjecture that $\Ker \DD=\Ker \DD_A$ at least for generic
$\Th$'s}:

the constraints (\ref{contraintenoyau}) should imply $c_{l,j} = 0$
for all $j$ and all $l \neq 0$ meaning $\psi \in \Ker \DD$. When
$\tfrac{1}{2\pi}\Th$ has only integer coefficients, the sin part of
these constraints disappears giving the result.
\medskip

\begin{lemma}
\label{spectrumset}
If $\tfrac{1}{2\pi} \Th$ is diophantine,
$Sp\big(\Coo(\T^n_\Th),\H,\DD\big)=\Z$ and all these poles are simple.
\end{lemma}

\begin{proof}
Let $B\in \DD(\A)$ and $p\in \N_0$. Suppose that $B$ is of the form
$$
B= a_r b_r
\DD^{q_{r-1}}|\DD|^{p_{r-1}} a_{r-1}b_{r-1}\cdots
\DD^{q_1}|\DD|^{p_1} a_1 b_1
$$
where $r\in \N$, $a_i \in \A$, $b_i\in J\A J^{-1}$, $q_i, p_i \in
\N_0$.
We note $a_i=:\sum_l a_{i,l}\,U_l$ and
$b_i=:\sum_i b_{i,l} \,U_l$. With the shorthand
$k_{\mu_1,\mu_{q_i}}:=k_{\mu_1}\cdots
k_{\mu_{q_i}}$ and $\ga^{\mu_1,\mu_{q_i}}=\ga^{\mu_1}\cdots
\ga^{\mu_{q_i}}$, we get
$$
\DD^{q_1}|\DD|^{p_1}  a_1 b_1 \, U_k \otimes e_j =  \sum_{l_1,\,l'_1}
a_{1,l_1} b_{1,l'_1}
U_{l_1}U_k U_{l'_1}
\,|k+l_1+l'_1|^{p_1}\,(k+l_1+l'_1)_{\mu_1,\mu_{q_1}} \otimes
\ga^{\mu_1,\mu_{q_1}} e_j
$$
which gives, after $r$ iterations,
$$
B U_k \otimes e_j = \sum_{l,l'} \wt a_{l} \wt b_{l} U_{l_r}\cdots
U_{l_1} U_k
U_{l'_1}\cdots U_{l'_r} \prod_{i=1}^{r-1} |k+\wh l_i+\wh
l'_i|^{p_i}(k+\wh l_{i} +\wh
l'_{i})_{\mu^{i}_1,\mu^i_{q_i}} \otimes
\ga^{\mu^{r-1}_1,\mu^{r-1}_{q_{r-1}}}\cdots
\ga^{\mu^1_1,\mu^1_{q_1}} e_j
$$
where $\wt a_l : = a_{1,l_1}\cdots a_{r,l_r}$ and $\wt b_{l'} : =
b_{1,l'_1}\cdots
b_{r,l'_r}$.

Let us note $F_\mu(k,l,l'):=\prod_{i=1}^{r-1}|k+\wh l_i+\wh
l'_i|^{p_i}
(k+\wh l_{i} +\wh l'_{i})_{\mu^{i}_1,\mu^i_{q_i}}$ and $\ga^\mu
:=\ga^{\mu^{r-1}_1,\mu^{r-1}_{q_{r-1}}}\cdots
\ga^{\mu^1_1,\mu^1_{q_1}}$. Thus, with
the shortcut $\sim_c$ meaning modulo a constant function towards the
variable $s$,
$$
\Tr \big(B|D|^{-2p-s}\big) \sim_c {\sum_k}' \, \sum_{l,l'} \wt a_l
\wt b_{l'} \,
\tau\big(U_{-k}U_{l_r}\cdots U_{l_1} U_k U_{l'_1}\cdots U_{l'_r}\big)
\tfrac{F_\mu(k,l,l')}{|k|^{s+2p}} \Tr (\ga^\mu)\, .
$$
Since $U_{l_r}\cdots U_{l_1} U_k = U_k U_{l_r}\cdots U_{l_1}
e^{-i\sum_1^r l_i .\Th
k}$ we get
$$\tau\big(U_{-k}U_{l_r}\cdots U_{l_1} U_k U_{l'_1}\cdots
U_{l'_r}\big)=
\delta_{\sum_1^r l_i+l'_i,0} \, e^{i\phi(l,l')} \, e^{-i\sum_1^r
l_i.\Th k}$$ where $\phi$
is a real valued function. Thus,
\begin{align*}
\Tr \big(B |D|^{-2p-s} \big)&\sim_c {\sum_k}' \, \sum_{l,l'}
e^{i\phi(l,l')}\,\delta_{\sum_1^r l_i+l'_i,0}\, \wt a_l \wt b_{l'}\,
\tfrac{F_\mu(k,l,l')\,e^{-i\sum_1^r l_i.\Th k}}{|k|^{s+2p}} \Tr
(\ga^\mu) \\
&\sim_c f_\mu(s)\Tr (\ga^\mu).
\end{align*}
The function $f_\mu(s)$ can be decomposed has a linear combination of
zeta function
of type described in Theorem \ref{zetageneral} (or, if $r=1$ or all
the $p_i$ are zero,
in Theorem \ref{analytic}).
Thus, $s\mapsto \Tr \big(B |D|^{-2p-s}\big)$
has only poles in $\Z$ and each pole is simple.
Finally, by linearity, we get the result.
\end{proof}
The dimension spectrum of the noncommutative torus is simple:
\begin{prop} \label{zeta(0)}

(i) If $\tfrac{1}{2\pi} \Th$ is diophantine,
the spectrum dimension of $\big(\Coo(\T^n_\Th),\H,\DD\big)$ is
equal to the set $\set{n-k \, :\,  k\in \N_0}$ and all these poles
are simple.

(ii) $\zeta_D(0)=0.$
\end{prop}
\begin{proof}
$(i)$ Lemma \ref{spectrumset} and Remark \ref{remark-spectrum}.

$(ii)$  $\zeta_D(s)={\sum}_{k\in \Z^n}
\sum_{1\leq j\leq 2^m} \<
U_k\otimes e_j, |D|^{-s}U_k\otimes e_{j}>=2^m( {\sum}'_{k\in\Z^n}
\frac{1}{|k|^{s}} + 1) =2^m(\,Z_n(s)+1).$ By (\ref{Zn0}), we get the
result.
\end{proof}

We have computed $\zeta_D(0)$ relatively easily but the
main difficulty of the present work is precisely to calculate
$\zeta_{D_A}(0)$.

\subsection{Noncommutative integral computations}

We fix a self-adjoint 1-form $A$ on the noncommutative torus of
dimension $n$.

\begin{prop}
\label{invariance} If $\tfrac{1}{2\pi}\Th$ is diophantine, then
the first elements of the expansion (\ref{formuleaction}) are
given by
\begin{align}
\ncint {|D_{A}|}^{-n}\,&=\ncint |D|^{-n}=
2^{m+1}\pi^{n/2}\,\Gamma(\tfrac{n}{2})^{-1}.\\
\ncint \vert D_{A}\vert^{n-k}&=0 \text{ for k odd}.\nonumber\\
\ncint \vert D_{A}\vert^{n-2}&=0.\nonumber
\end{align}
\end{prop}
We need few technical lemmas:
\begin{lemma}
\label{traceAD}
On the noncommutative torus, for any $t\in \R$,
$$
\ncint \wt A \DD |D|^{-t}= \ncint \DD \wt A |D|^{-t} =0.
$$
\end{lemma}

\begin{proof}
Using notations of (\ref{connection}), we have
\begin{align*}
\Tr (\wt A \DD |D|^{-s})&\sim_c {\sum}_j {\sum}'_k \langle U_k\otimes
e_j,-i k_\mu|k|^{-s}
[A_\a,U_k] \otimes \ga^\a \ga^\mu e_j \rangle \\
&\sim_c -i\Tr(\ga^\a\ga^\mu) \, {\sum}'_k k_\mu
|k|^{-s} \langle U_k,[A_\a,U_k] \rangle=0
\end{align*}
since $\langle U_k,[A_\a,U_k] \rangle = 0$. Similarly
\begin{align*}
\Tr ( \DD \wt A  |D|^{-s})&\sim_c {\sum}_j {\sum}'_k \langle
U_k\otimes
e_j,|k|^{-s}{\sum}_l a_{\a,l}\,2 \sin \tfrac{k. \Th l}{2} (l+k)_\mu
U_{l+k}
\otimes \ga^\mu \ga^\a e_j \rangle\\
&\sim_c 2\Tr(\ga^\mu \ga^\a){\sum}'_k {\sum}_l a_{\a,l}\sin \tfrac{k.
\Th
l}{2}\,(l+k)_\mu \,|k|^{-s}\langle U_k,U_{l+k} \rangle =0.
\tag*{\qed}
\end{align*}
\hideqed
\end{proof}

Any element $h$ in the algebra generated by $\A$ and $[\DD,\A]$ can
be written as a
linear combination of terms of the form ${a_1}^{p_1}\cdots
{a_n}^{p_r}$ where
$a_i$ are elements of $\A$ or $[\DD,\A]$. Such a term can be written
as a series $b:=\sum a_{1,\a_1,l_1}\cdots a_{q,\a_q,l_q} U_{l_1}\cdots
U_{l_q} \otimes \ga^{\a_1}\cdots \ga^{\a_q}$ where
$a_{i,\a_i}$ are Schwartz sequences and when $a_i=:\sum_l a_l U_l \in
\A$, we set $a_{i,\a,l}=a_{i,l}$ with $\ga^\a =1$. We define
$$
L(b):= \tau \big({\sum}_l a_{1,\a_1,l_1}\cdots a_{q,\a_q,l_q} U_{l_1}
\cdots U_{l_q}\big) \Tr (\ga^{\a_1}\cdots \ga^{\a_q}).
$$
By linearity, $L$ is defined as a linear form on the whole algebra
generated by $\A$
and $[\DD,\A]$.

\begin{lemma}
\label{tracehD}
If $h$ is an element of the algebra generated by $\A$ and $[\DD,\A]$,
$$
\Tr \big(h |D|^{-s}\big) \sim_c L(h)\,  Z_n(s).
$$
In particular, $\Tr \big(h |D|^{-s}\big)$ has at most one pole at
$s=n$.
\end{lemma}
\begin{proof} We get with $b$ of the form $\sum
a_{1,\a_1,l_1}\cdots a_{q,\a_q,l_q} U_{l_1}\cdots U_{l_q} \otimes
\ga^{\a_1}\cdots
\ga^{\a_q}$,
\begin{align*}
\Tr\big(b|D|^{-s}\big)&\sim_c {\sum_{k\in\Z^n}}' \langle  U_k, \sum_l
a_{1,\a_1,l_1}\cdots a_{q,\a_q,l_q} U_{l_1}\cdots U_{l_q}U_k \rangle
\, \Tr
(\ga^{\a_1}\cdots \ga^{\a_q})|k|^{-s} \\
&\sim_c \tau(\sum_l a_{1,\a_1,l_1}\cdots a_{q,\a_q,l_q} U_{l_1}\cdots
U_{l_q})\Tr(\ga^{\a_1}\cdots \ga^{\a_q})\, Z_n(s) =L(b) \,Z_n(s).
\end{align*}
The results follows now from linearity of the trace.
\end{proof}

\begin{lemma}
\label{traceJAJA}
If $\tfrac{1}{2\pi}\Th$ is diophantine, the function
$s\mapsto\Tr \big( \eps JAJ^{-1} A |D|^{-s} \big)$ extends
meromorphically on the whole plane with only one
possible pole at $s=n$. Moreover, this pole is simple and
$$
\underset{s=n}{\Res}\, \Tr \big(\eps JAJ^{-1} A |D|^{-s}\big) =
a_{\a,0}\,a^\a_{0}\
2^{m+1}\pi^{n/2}\,\Ga(n/2)^{-1}.
$$
\end{lemma}
\begin{proof} With $A=L(-i A_\a)\otimes \ga^\a$, we get
$\epsilon J A J^{-1}=R(i A_\a)\otimes \ga^\a$, and by multiplication
$\eps JAJ^{-1} A=R(A_\beta)
L(A_\a)\otimes \ga^{\beta}\ga^\a$. Thus,
\begin{align*}
\Tr\big(\eps JAJ^{-1} A |D|^{-s}\big)&\sim_c {\sum_{k\in\Z^n}}'
\langle U_k,A_\a
U_k A_\beta \rangle \,|k|^{-s}\Tr (\ga^{\beta}\ga^{\a}) \\
&\sim_c{\sum_{k\in\Z^n}}' \,\sum_{l}
a_{\a,l}\,a_{\beta,-l}\,e^{ik.\Th l}\,
|k|^{-s}\Tr (\ga^{\beta}\ga^{\a})\\
&\sim_c 2^m {\sum_{k\in\Z^n}}' \, \sum_{l}
a_{\a,l}\,a^\a_{-l}\,e^{ik.\Th l}\,
|k|^{-s}.
\end{align*}
Theorem \ref{analytic} $(ii)$ entails that ${\sum}'_{k\in\Z^n} \,
\sum_{l} a_{\a,l} \,a^{\a}_{-l}\,e^{ik.\Th l}\, |k|^{-s}$ extends
meromorphically
on the whole plane $\C$ with only one possible pole at $s=n$.
Moreover, this pole is
simple and we have
$$
\underset{s=n}{\Res}\, {\sum_{k\in\Z^n}}' \,\sum_{l} a_{\a,l}
\,a^\a_{-l}\,e^{ik.\Th l}\,
|k|^{-s} =  a_{\a,0}\,a^\a_{0} \, \underset{s=n}{\Res}\, Z_n(s).
$$
Equation (\ref{formule}) now gives the result.
\end{proof}

\begin{lemma}
    \label{traceXD}
If $\tfrac{1}{2\pi}\Th$ is diophantine, then for any $t\in \R$,
$$
\ncint X|D|^{-t} = \delta_{t,n}\, 2^{m+1}\big(-\sum_l
a_{\a,l}\,a^\a_{-l}+
\,a_{\a,0}\,a^\a_{0}\big)\ 2\pi^{n/2}\,\Ga(n/2)^{-1}  .
$$
where $X=\wt A\DD + \DD \wt A + {\wt A}^2$ and $A=:-i\sum_{l}
a_{\a,l}\,U_l\otimes \ga^\a$.
\end{lemma}

\begin{proof} By Lemma \ref{traceAD},
we get $\ncint X|D|^{-t}=\Res_{s=0} \Tr({\wt A}^2 |D|^{-s-t})$. Since
$A$ and
$\eps JAJ^{-1}$ commute, we have $\wt A ^2 = A^2 + JA^2J^{-1} + 2\eps
JAJ^{-1}A$.
Thus,
$$
\Tr({\wt A}^2 |D|^{-s-t})=\Tr( A^2 |D|^{-s-t})+\Tr( JA^2J^{-1}
|D|^{-s-t})+2\Tr ( \eps JAJ^{-1}A |D|^{-s-t}).
$$
Since $|D|$ and $J$
commute, we have with Lemma \ref{tracehD},
$$
 \Tr \big({\wt A}^2 |D|^{-s-t}\big)\sim_c 2L (A^2) \, Z_n(s+t) +
 2 \Tr \big(\eps JAJ^{-1}A |D|^{-s-t}\big).
$$
Thus Lemma \ref{traceJAJA} entails that $\Tr({\wt A}^2 |D|^{-s-t})$
is holomorphic at 0 if $t\neq n$. When $t=n$,
\begin{equation}
    \label{TrA^2}
\underset{s=0}{\Res}\, \Tr\big({\wt A}^2 |D|^{-s-t}\big) =
2^{m+1}\big(-\sum_l a_{\a,l} \, a^\a_{-l}+
\,a_{\a,0}\,a^\a_{0}\ \big)\, 2\pi^{n/2}\,\Ga(n/2)^{-1},
\end{equation}
which gives the result.
\end{proof}

\begin{lemma}
\label{traceAA}
If $\tfrac{1}{2\pi}\Th$ is diophantine, then
$$
\ncint \wt A \DD \wt A \DD |D|^{-2-n}=-\tfrac{n-2}{n}
\ncint \wt A^2 |D|^{-n}.
$$
\end{lemma}
\begin{proof}
With $\DD J = \eps J \DD$, we get
$$
\ncint \wt A \DD \wt A \DD |D|^{-2-n} = 2 \ncint  A \DD  A \DD
|D|^{-2-n} + 2
\ncint \eps JAJ^{-1} \DD A \DD |D|^{-2-n}.
$$
Let us first compute $\ncint  A \DD  A \DD |D|^{-2-n}$. We have, with
$A=:-i
L(A_\a)\otimes \ga^\a=: -i\sum_{l} a_{\a,l} U_l \otimes \ga^\a$,
$$
\Tr \big(A\DD A \DD |D|^{-s-2-n}\big) \sim_c -{{\sum_{k}}}'
\sum_{l_1,l_2}
a_{\a_2,l_2}\,a_{\a_1,l_1}
\,\tau(U_{-k} U_{l_2} U_{l_1} U_k) \,
\tfrac{k_{\mu_1}(k+l_1)_{\mu_2}}{|k|^{s+2+n}}
\Tr(\ga^{\a,\mu})
$$
where $\ga^{\a,\mu}:= \ga^{\a_2}\ga^{\mu_2}  \ga^{\a_1}\ga^{\mu_1}$.
Thus,
$$
\ncint  A \DD  A \DD |D|^{-2-n} = -\sum_{l} a_{\a_2,-l}\,a_{\a_1,l} \,
\underset{s=0}{\Res}\,\big({{\sum_{k}}}'
\tfrac{k_{\mu_1}k_{\mu_2}}{|k|^{s+2+n}} \big)
\Tr(\ga^{\a,\mu}).
$$
We have also, with $\eps JAJ^{-1} = iR(A_\a)\otimes \ga^{a}$,
$$
\Tr \big(\eps JAJ^{-1}\DD A \DD |D|^{-s-2-n}\big) \sim_c
{{\sum_{k}}}' \sum_{l_1,l_2}
a_{\a_2,l_2}a_{\a_1,l_1} \tau(U_{-k} U_{l_1} U_k U_{l_2})
\tfrac{k_{\mu_1}(k+l_1)_{\mu_2}}{|k|^{s+2+n}} \Tr(\ga^{\a,\mu}).
$$
which gives
$$
\ncint  \eps JAJ^{-1} \DD  A \DD |D|^{-2-n} =a_{\a_2,0}a_{\a_1,0} \,
\underset{s=0}{\Res}\,\big({{\sum_{k}}}'
\tfrac{k_{\mu_1}k_{\mu_2}}{|k|^{s+2+n}} \big)
\Tr(\ga^{\a,\mu}).
$$
Thus,
$$
\half \ncint \wt A \DD \wt A \DD |D|^{-2-n} =
\big(a_{\a_2,0}a_{\a_1,0}-\sum_{l}
a_{\a_2,-l}a_{\a_1,l}\big) \Res_{s=0}\big({{\sum_{k}}}'
\tfrac{k_{\mu_1}k_{\mu_2}}{|k|^{s+2+n}} \big) \Tr(\ga^{\a,\mu}).
$$
With ${\sum}'_k \tfrac{k_{\mu_1}k_{\mu_2}}{|k|^{s+2+n}} =
\tfrac{\delta_{\mu_1\mu_2}}{n}Z_n(s+n)$ and
$C_n:=\Res_{s=0} Z_n(s+n) = 2\pi^{n/2} \Ga(n/2)^{-1}$ we obtain
$$
\half \ncint \wt A \DD \wt A \DD |D|^{-2-n} =
\big(a_{\a_2,0}a_{\a_1,0}-\sum_{l}
a_{\a_2,-l}a_{\a_1,l}\big)\tfrac{C_n}{n}
\Tr(\ga^{\a_2}\ga^{\mu}\ga^{\a_1}\ga_\mu).
$$
Since  $\Tr(\ga^{\a_2}\ga^{\mu}\ga^{\a_1}\ga_\mu)=
2^m(2-n)\delta^{\a_2,\a_1}$,
we get
$$
\half \ncint \wt A \DD \wt A \DD |D|^{-2-n} =
2^m\big(-a_{\a,0}\,a^\a_0+\sum_{l}
a_{\a,-l}\,a^\a_l\big)\tfrac{C_n (n-2)}{n}.
$$
Equation (\ref{TrA^2}) now proves the lemma.
\end{proof}

\begin{lemma}
\label{ncint-odd-pdo}
If $\tfrac{1}{2\pi}\Th$ is diophantine, then
for any $P\in \Psi_{1}(\A)$ and $q \in \N$, $q$ odd,
$$
\ncint P |D|^{-(n-q)} = 0.
$$
\end{lemma}
\begin{proof}
There exist $B\in \DD_{1}(\A)$ and $p\in \N_0$ such that
$P= BD^{-2p}+R$ where $R$ is in $OP^{-q-1}$.
As a consequence, $\ncint P |D|^{-(n-q)} = \ncint
B|D|^{-n-2p+q}$. Assume $B= a_r b_r
\DD^{q_{r-1}}a_{r-1}b_{r-1}\cdots \DD^{q_1} a_1 b_1 $ where
$r\in \N$, $a_i \in \A$,
$b_i\in J\A J^{-1}$, $q_i\in \N$. If we prove that
$\ncint B|D|^{-n-2p+q} =0$, then
the general case will follow by linearity. We note
$a_i=:\sum_l a_{i,l}\,U_l$ and
$b_i=:\sum_l b_{i,l} \,U_l$. With the shorthand
$k_{\mu_1,\mu_{q_i}}:=k_{\mu_1}\cdots
k_{\mu_{q_i}}$ and $\ga^{\mu_1,\mu_{q_i}}=\ga^{\mu_1}\cdots
\ga^{\mu_{q_i}}$, we get
$$
\DD^{q_1}  a_1 b_1 U_k \otimes e_j =
\sum_{l_1,l'_1} \,a_{1,l_1}\, b_{1,l'_1}\, U_{l_1}U_k
U_{l'_1} \,(k+l_1+l'_1)_{\mu_1,\mu_{q_1}}
\otimes \ga^{\mu_1,\mu_{q_1}} e_j
$$
which gives, after iteration,
$$
B\, U_k \otimes e_j = \sum_{l,l'} \wt a_{l} \wt b_{l} U_{l_r}
\cdots U_{l_1} U_k
U_{l'_1}\cdots U_{l'_r} \prod_{i=1}^{r-1} (k+\wh l_{i} +\wh
l'_{i})_{\mu^{i}_1,\mu^i_{q_i}} \otimes \ga^{\mu^{r-1}_1,
\mu^{r-1}_{q_{r-1}}}\cdots \ga^{\mu^1_1,\mu^1_{q_1}} e_j
$$
where $\wt a_l : = a_{1,l_1}\cdots a_{r,l_r}$ and
$\wt b_{l'} : = b_{1,l'_1}\cdots
b_{r,l'_r}$. Let's note $Q_\mu(k,l,l'):=\prod_{i=1}^{r-1}
(k+\wh l_{i} +\wh l'_{i})_{\mu^{i}_1,\mu^i_{q_i}}$ and
$\ga^\mu :=\ga^{\mu^{r-1}_1,\mu^{r-1}_{q_{r-1}}}\cdots
\ga^{\mu^1_1,\mu^1_{q_1}}$. Thus,
$$
\ncint B\,|D|^{-n-2p+q} =\underset{s=0}{\Res}\,  {\sum_k}' \,
\sum_{l,l'} \wt a_l
\,\wt b_{l'} \, \tau\big(U_{-k}U_{l_r}\cdots U_{l_1} U_k
U_{l'_1}\cdots U_{l'_r}\big)
\,\tfrac{Q_\mu(k,l,l')}{|k|^{s+2p+n-q}} \,\Tr (\ga^\mu)\, .
$$
Since $U_{l_r}\cdots U_{l_1} U_k = U_k U_{l_r}
\cdots U_{l_1} e^{-i\sum_1^r l_i .\Th
k}$, we get
$$\tau\big(U_{-k}U_{l_r}\cdots U_{l_1} U_k U_{l'_1}\cdots
U_{l'_r}\big)=
\delta_{\sum_1^r l_i+l'_i,0} \,e^{i\phi(l,l')} \,
e^{-i\sum_1^r l_i.\Th k}$$ where $\phi$
is a real valued function. Thus,
\begin{align*}
\ncint B\,|D|^{-n-2p+q} &=\underset{s=0}{\Res}\,
{\sum_k}'\,\sum_{l,l'}
e^{i\phi(l,l')}\,\delta_{\sum_1^r l_i+l'_i,0}\, \wt a_l \,\wt b_{l'}
\,\tfrac{Q_\mu(k,l,l')e^{-i\sum_1^r l_i.\Th k}}{|k|^{s+2p+n-q}}
\Tr (\ga^\mu) \\
&=:\underset{s=0}{\Res}\, f_\mu(s)\Tr (\ga^\mu).
\end{align*}

We decompose $Q_{\mu}(k,l,l')$ as a sum
$\sum_{h=0}^r M_{h,\mu}(l,l') \, Q_{h,\mu}(k)$
where $Q_{h,\mu}$ is a homogeneous polynomial in $(k_1,\cdots,k_n)$
and $M_{h,\mu}(l,l')$ is a polynomial in
$\big((l_1)_1,\cdots,(l_{r})_n,(l'_1)_1,\cdots,(l'_{r})_n \big)$.

Similarly, we decompose $f_{\mu}(s)$ as $\sum_{h=0}^{r}
f_{h,\mu}(s)$. Theorem
\ref{analytic} $(ii)$ entails that $ f_{h,\mu}(s)$ extends
meromorphically to the whole complex plane $\C$ with only one
possible pole for $s+2p+n-q=n+d$ where
$d:=\text{deg } Q_{h,\mu}$. In other words, if $d+q-2p\neq 0$,
$f_{h,\mu}(s)$ is holomorphic at $s=0$. Suppose now $d+q-2p =0$
(note that this implies that $d$ is odd, since
$q$ is odd by hypothesis), then, by Theorem \ref{analytic} $(ii)$
$$
\underset{s=0}{\Res}\ f_{h,\mu}(s) = V \int_{u\in S^{n-1}}
Q_{h,\mu}(u)\,dS(u)
$$
where $V:=\sum_{l,l'\in Z} M_{h,\mu}(l,l')\,e^{i \phi(l,l')}\,
\delta_{\sum_1^r l_i+l'_i,0}\, \wt a_{l} \,\wt b_{l'}$ and
$Z:=\set{l,l' \, :\, \sum_{i=1}^{r} l_i=0}$.
Since $d$ is odd, $Q_{h,\mu}(-u)=-Q_{h,\mu}(u)$ and $\int_{u\in
S^{n-1}}Q_{h,\mu}(u)\,dS(u)=0$. Thus, $\underset{s=0}{\Res}
\ f_{h,\mu}(s)=0$ in any case, which gives the result.
\end{proof}

As we have seen, the crucial point of the preceding lemma is the
decomposition of the numerator of the series $f_\mu(s)$ as polynomials
in $k$. This has been possible because we restricted our
pseudodifferential
operators to $\Psi_1(\A)$.

\bigskip

{\it Proof of Proposition \ref{invariance}.}
The top element follows from Proposition \ref{invariance1} and
according to (\ref{formule}),
\begin{align*}
\ncint |D| ^{-n}= \underset{s=0}{\Res}\
\Tr\big(|D|^{-s-n}\big)=2^m\,\underset{s=0}{\Res}\,
Z_n(s+n)=\tfrac{2^{m+1}\pi^{n/2}}{\Gamma(n/2)}\, .
\end{align*}

For the second equality, we get from Lemmas \ref{tracehD}
and \ref{Res-zeta-n-k}
$$
\underset{s=n-k}{\Res}\,  \zeta_{D_A}(s)= \sum_{p=1}^k
\sum_{r_1,\cdots, r_p =0}^{k-p} h(n-k,r,p)
\ncint \eps^{r_1}(Y)\cdots\eps^{r_p}(Y) |D|^{-(n-k)}.
$$
Corollary  \ref{eps-pdo} and Lemma \ref{ncint-odd-pdo} imply that
$\ncint
\eps^{r_1}(Y)\cdots\eps^{r_p}(Y) |D|^{-(n-k)} =0$, which gives the
result.

Last equality follows from Lemma \ref{traceAA} and Corollary
\ref{res-n-2-A}.
\qed

\section{The spectral action}

Here is the main result of this section.

\begin{theorem}
\label{main}
Consider the $n$-NC-torus $\big(\Coo(\T^n_\Th),\H,\DD\big)$ where
$n\in \N$ and
$\tfrac{1}{2\pi}\Th$ is a real $n\times n$ skew-symmetric real diophantine matrix, and a selfadjoint one-form
$A=L(-iA_{\alpha})\otimes
\ga^{\alpha}$. Then, the full spectral action of
$\DD_{A}=\DD +A
+¾\epsilon JAJ^{-1}$ is

(i) for $n=2$,
$$
\SS(\DD_{A},\Phi,\Lambda)=4\pi\,\Phi_{2} \, \Lambda^{2} +
\mathcal{O}(\Lambda^{-2}),
$$
(ii) for $n=4$,
$$
\SS(\DD_{A},\Phi,\Lambda)= 8\pi^2\,\Phi_{4} \, \Lambda^{4}
-\tfrac{4\pi^{2}}{3}\
\Phi(0) \,\tau(F_{\mu\nu}F^{\mu\nu})+  \mathcal{O}(\Lambda^{-2}),
$$
(iii) More generally, in
$$
\SS(\DD_{A},\Phi,\Lambda) \, = \,\sum_{k=0}^n \Phi_{n-k}\,
c_{n-k}(A) \,\Lambda^{n-k} +\mathcal{O}(\Lambda^{-1}),
$$
$c_{n-2}(A)=0$, $c_{n-k}(A)=0$ for $k$ odd. In particular, $c_0(A)=0$
when $n$ is odd.
\end{theorem}

\quad

This result (for $n=4$) has also been obtained in \cite{GIVas} using
the heat kernel method. It is however interesting to get the result
via direct computations of
(\ref{formuleaction}) since it shows how this formula is efficient.
As we will see, the computation of all the noncommutative integrals
require a lot of technical steps.
One of the main points, namely to isolate where the Diophantine
condition on $\Th$ is assumed, is outlined here.

\begin{remark}
Note that all terms must be gauge invariants, namely, according
to (\ref{gaugetransform}), invariant by
$A_{\alpha}\longrightarrow \gamma_{u}(A_{\alpha})=
uA_{\alpha}u^{*}+u\delta_{\alpha}(u^{*})$. A particular case is
$u=U_{k}$ where
$U_{k}\delta_{\alpha}(U_{k}^{*})=-ik_{\alpha} U_{0}$.

In the same way, note that there is no contradiction with
the commutative case where, for any selfadjoint one-form
$A$, $\DD_{A}=\DD$ (so $A$ is equivalent to 0!), since
we assume in Theorem \ref{main} that
$\Th$ is diophantine, so $\A$ cannot be commutative.
\end{remark}

\begin{conjecture}
The constant term of the
spectral action of $\DD_{A}$ on the noncommutative n-torus is
proportional to the constant term of the spectral action
of $\DD+A$ on the commutative n-torus.
\end{conjecture}

\begin{remark}
The appearance of a Diophantine condition for $\Th$ has been
characterized in dimension 2 by Connes \cite[Prop. 49]{NCDG} where in
this case, $\Th=\th\genfrac{(}{)}{0pt}{1}{\,0 \quad1}{-1\,\,\,\, 0 }$
with $\th \in \R$. In fact, the Hochschild cohomology
$H(\A_{\Th},{\A_{\Th}}^*)$ satisfies dim
$H^j(\A_{\Th},{\A_{\Th}}^*)=2$ (or $1$) for $j=1$ (or $j=2$) if and
only if the irrational number $\th$ satisfies a Diophantine condition
like $\vert 1-e^{i2\pi n \th} \vert^{-1} =\mathcal{O}(n^k)$ for some
$k$.

Recall that when the matrix $\Th$ is quite irrational (see \cite[Cor.
2.12]{Polaris}), then the C$^*$-algebra generated by $\A_{\Th}$ is
simple.
\end{remark}

\begin{remark}
It is possible to generalize above theorem to the case $\DD=-i\,{g^{\mu}}_{\nu} \, \delta_\mu \otimes \ga^\nu$ instead of \eqref{defDirac} when $g$ is a positive definite constant matrix. The formulae in Theorem \ref{main}  are still valid. 
\end{remark}

\subsection{Computations of $\ncint$}
In order to get this theorem, let us prove a few technical lemmas.

We suppose
from now on that $\Th$ is a skew-symmetric matrix in
$\mathcal{M}_n(\R)$. No other
hypothesis is assumed for $\Th$, except when it is explicitly stated.

When $A$ is a selfadjoint one-form, we define for $n\in N$, $q\in \N$,
$2\leq q \leq n$ and
$\sigma\in \{-,+\}^q$
\begin{align*}
\mathbb{A}^{+}&:=A\DD D^{-2},\\
\mathbb{A}^{-}&:= \epsilon JAJ^{-1} \DD D^{-2},\\
\mathbb{A}^{\sigma}&:=\mathbb{A}^{\sigma_q}\cdots
\mathbb{A}^{\sigma_1} .
\end{align*}

\begin{lemma}
\label{ncadmoins1}
We have for any $q\in \N$,
$$
\ncint (\wt A D^{-1})^q = \ncint (\wt A \DD D^{-2})^q =
\sum_{\sigma\in \set{+,-}^{q}}
\ncint \mathbb{A}^{\sigma}.
$$
\end{lemma}
\begin{proof} Since $P_0 \in OP^{-\infty}$, $D^{-1} = \DD D^{-2} \mod
OP^{-\infty}$
and $\ncint (\wt A D^{-1})^q = \ncint (\wt A \DD D^{-2})^q$.
\end{proof}

\begin{lemma}
\label{symetrie}
Let $A$ be a selfadjoint one-form, $n\in \N$ and $q\in \N$ with
$2\leq q \leq n$ and
$\sigma\in \{-,+\}^q$. Then
$$
\ncint \mathbb{A}^{\sigma} =\ncint \mathbb{A}^{-\sigma}.
$$
\end{lemma}
\begin{proof}Let us first check that $JP_0 = P_0 J$. Since $\DD J =
\eps J\DD$, we get $\DD J P_0=0$ so $JP_0 = P_0 J P_0$. Since $J$ is
an antiunitary operator,
we get $P_0 J = P_0 J P_0$  and finally, $P_0 J = J P_0$.
As a consequence, we get $JD^2 = D^2 J$,
$J \DD D^{-2}=\eps \DD D^{-2} J$, $J\mathbb{A}^{+}J^{-1} =
\mathbb{A}^{-}$ and
$J\mathbb{A}^{-}J^{-1} = \mathbb{A}^{+}.$

In summary, $J\mathbb{A}^{\sigma_i}J^{-1} = \mathbb{A}^{-\sigma_i}$.

The trace property of $\ncint$ now gives
\begin{align*}
\ncint \mathbb{A}^{\sigma} &= \ncint \mathbb{A}^{\sigma_q}\cdots
\mathbb{A}^{\sigma_1} = \ncint J
\mathbb{A}^{\sigma_q}J^{-1}\cdots J\mathbb{A}^{\sigma_1}J^{-1} \ncint
\mathbb{A}^{-\sigma_q} \cdots
\mathbb{A}^{-\sigma_1}=\ncint \mathbb{A}^{-\sigma}.
\tag*{\qed}
\end{align*}
\hideqed
\end{proof}
\begin{definition}

In \cite{CC1} has been introduced the vanishing tadpole hypothesis:
\begin{align}
    \label{vanishtad}
    \ncint A D^{-1}=0, \text{ for all } A\in \Omega_{\DD}^{1}(\A).
\end{align}
\end{definition}
By the following lemma, this condition is satisfied for the
noncommutative torus, a fact more or less already known within the
noncommutative community \cite{Walter}.

\begin{lemma}
\label{tadpole}
Let $n\in \N$, $A= L(-iA_{\a})\otimes \gamma^{\a}=-i\sum_{l\in
\Z^n}a_{\alpha,l} \,
U_{l}\otimes \ga^{\alpha}$, $A_{\a}\in \A_{\Th}$,
$\set{a_{\alpha,l}}_{l}\in \SS
(\Z^n)$, be a hermitian one-form.
    Then, \\
    (i) $\ncint A^p D^{-q}  =  \ncint (\epsilon JAJ^{-1})^p
D^{-q}=0$ for $p\geq0$ and $1\leq q < n$ (case $p=q=1$ is tadpole
hypothesis.)\\
    (ii) If $\tfrac{1}{2\pi} \Th$ is diophantine, then $\ncint
BD^{-q}=0$ for $1\leq q < n$ and any $B$ in the algebra generated by
$\A$, $[\DD,\A]$, $J\A J^{-1}$ and $J[\DD,\A]J^{-1}$. 
\end{lemma}

\begin{proof}
$(i)$ Let us compute $$\ncint A^p(\epsilon
JAJ^{-1})^{p'}D^{-q}.$$ 
With $A=L(-i A_\a)\otimes \ga^\a$ and
$\epsilon J A
J^{-1}=R(i A_\a)\otimes \ga^\a$, we get
$$
A^p=L(-i A_{\a_1})\cdots L(-i A_{\a_p}) \otimes \ga^{\a_1}\cdots
\ga^{\a_p}
$$
and
$$
(\epsilon J A J^{-1})^{p'}=R(i A_{\a'_1})\cdots R(i A_{\a'_{p'}})
\otimes
\ga^{\a'_1}\cdots \ga^{\a'_{p'}}.
$$
We note $\wt a_{\a,l}:= a_{\a_1,l_1}\cdots a_{\a_p,l_p}$. Since
$$
L(-i A_{\a_1})\cdots L(-i A_{\a_p})R(i A_{\a'_1})\cdots R(i
A_{\a'_{p'}}) U_k= (-i)^p
\, i^{p'} \sum_{l,l'} \wt a_{\a,l} \, \wt a_{\a',l'} \, U_{l_1}\cdots
U_{l_p} U_k
U_{l'_{p'}}\cdots U_{l'_1},
$$
and
$$
U_{l_1}\cdots U_{l_p} U_k= U_k U_{l_1}\cdots U_{l_p} \, e^{-i(\sum_i
l_i).\Th k},
$$
we get, with
\begin{align*}
&U_{l,l'}:=U_{l_1}\cdots U_{l_p}U_{l'_{p'}}\cdots U_{l'_1},\\
&g_{\mu,\a,\a'}(s,k,l,l'):= e^{ik. \Th \sum_j l_j} \,
\tfrac{k_{\mu_{1}}\ldots
k_{\mu_{q}}}{\vert k \vert^{s+2q}} \, \wt a_{\a,l} \,  \wt a_{\a',l'},
\\
&\ga^{\a,\a',\mu}:=\ga^{\a_1}\cdots\ga^{\a_{p}}\ga^{\a'_1}\cdots
\ga^{\a'_{p'}} \ga^{\mu_1}\cdots \ga^{\mu_q},
\end{align*}
$$
A^p(\epsilon JAJ^{-1})^{p'}D^{-q}|D|^{-s} U_k\otimes e_i\sim_c (-i)^p
\, i^{p'}
\sum_{l,l'} g_{\mu,\a,\a'}(s,k,l,l') \, U_k U_{l,l'} \otimes
\ga^{\a,\a',\mu} e_i.
$$
Thus, $\ncint A^p(\epsilon JAJ^{-1})^{p'}D^{-q} =
\underset{s=0}{\Res}\ f(s)$
where
\begin{align*}
f(s):&=\Tr\big(A^p(\epsilon JAJ^{-1})^{p'}D^{-q}|D|^{-s}\big)\\
&\sim_c (-i)^p \, i^{p'} {\sum_{k\in\Z^{n}}}' \langle U_{k}\otimes
e_{i},\sum_{l,l'}
g_{\mu,\a,\a'}(s,k,l,l')  U_k
U_{l,l'} \otimes \ga^{\a,\a',\mu} e_i \rangle\\
&\sim_c (-i)^p \, i^{p'} {\sum_{k\in\Z^{n}}}' \, \,\tau\big(
\sum_{l,l'}
g_{\mu,\a,\a'}(s,k,l,l')
U_{l,l'} \big) \Tr(\ga^{\mu,\a,\a'})\\
&\sim_c (-i)^p \, i^{p'} {\sum_{k\in\Z^{n}}}' \sum_{l,l'}
g_{\mu,\a,\a'}(s,k,l,l') \, \tau
\big(U_{l,l'} \big) \Tr(\ga^{\mu,\a,\a'}).
\end{align*}
It is straightforward to check that the series
${\sum}'_{k,l,l'} g_{\mu,\a,\a'}(s,k,l,l') \, \tau\big( U_{l,l'}
\big)$
is absolutely summable if $\Re(s)>R$ for a $R>0$. Thus, we can
exchange the summation
on $k$ and $l,l'$, which gives
$$
f(s)\sim_c (-i)^p \, i^{p'} \sum_{l,l'}  {\sum_{k\in\Z^{n}}}'
g_{\mu,\a,\a'}(s,k,l,l') \,
\tau \big( U_{l,l'} \big) \Tr(\ga^{\mu,\a,\a'}).
$$
If we suppose now that $p'=0$, we see that,
$$
f(s)\sim_c  (-i)^p \sum_{l}  {\sum_{k\in\Z^{n}}}'   \,
\tfrac{k_{\mu_{1}}\ldots
k_{\mu_{q}}}{\vert k \vert^{s+2q}} \, \wt a_{\a,l} \,
\delta_{\sum_{i=1}^p l_i,0}
\Tr(\ga^{\mu,\a,\a'})
$$
which is, by Proposition \ref{calculres}, analytic at 0. In
particular, for
$p=q=1$, we see that $\ncint A D^{-1} =0$, i.e. the
vanishing tadpole
hypothesis is satisfied. Similarly, if we suppose $p=0$,
we get
$$
f(s)\sim_c  (-i)^{p'} \sum_{l'}  {\sum_{k\in\Z^{n}}}'   \,
\tfrac{k_{\mu_{1}}\ldots
k_{\mu_{q}}}{\vert k \vert^{s+2q}} \, \wt a_{\a,l'} \,
\delta_{\sum_{i=1}^{p'} {l'}_i,0}
\Tr(\ga^{\mu,\a,\a'})
$$
which is holomorphic at 0.

$(ii)$ Adapting the proof of Lemma \ref{ncint-odd-pdo} to our setting
(taking $q_i=0$, and 
adding gamma matrices components), we see that
$$
\ncint B\,D^{-q} =\underset{s=0}{\Res}\, {\sum_k}'\,\sum_{l,l'}
e^{i\phi(l,l')}\,\delta_{\sum_1^r l_i+l'_i,0}\, \wt a_{\a,l} \,\wt
b_{\beta,l'}
\,\tfrac{k_{\mu_1}\cdots k_{\mu_q}\,e^{-i\sum_1^r l_i.\Th
k}}{|k|^{s+2q}}
\Tr (\ga^{(\mu,\a,\beta)})
$$
where $\ga^{(\mu,\a,\beta)}$ is a complicated product of gamma matrices.
By Theorem \ref{analytic} $(ii)$, since we suppose here that
$\tfrac{1}{2\pi} \Th$ is
diophantine, this residue is 0.
\end{proof}

\subsubsection{Even dimensional case}
\begin{corollary}

Same hypothesis as in Lemma \ref{tadpole}.

(i) Case $n=2$:
    \begin{align*}
\ncint A^q D^{-q}= -\delta_{q,2}\,4\pi  \,\tau\big(A_{\a} A^{\a}
\big) \,.
\end{align*}
    (ii) Case $n=4$: with the shorthand
$\delta_{\mu_1,\ldots,\mu_4}:=
\delta_{\mu_1\mu_2}\delta_{\mu_3\mu_4}+\delta_{\mu_1\mu_3}\delta_{\mu_2\mu_4}
+\delta_{\mu_1\mu_4}\delta_{\mu_2\mu_3}$,
\begin{align*}
\ncint A^q D^{-q}= \delta_{q,4}\,\tfrac{\pi^2}{12}
\tau\big(A_{\a_1}\cdots A_{\a_4}
\big)\Tr(\ga^{\a_1}\cdots\ga^{\a_{4}}\ga^{\mu_1}\cdots\ga^{\mu_4})
\delta_{\mu_1,\ldots,\mu_4} \,.
\end{align*}
\end{corollary}

\begin{proof}
$(i,ii)$ The same computation as in Lemma \ref{tadpole} $(i)$ (with
$p'=0$, $p=q=n$) gives
\begin{align*}
\ncint A^n D^{-n}=\underset{s=0}{\Res}(-i)^{n}
\big({\sum_{k\in\Z^{n}}}'\tfrac{k_{\mu_{1}}\ldots k_{\mu_{n}}}{\vert k
    \vert^{s+2n}}\big) \, \tau\big(\sum_{l\in (\Z^n)^n}
\wt a_{\a,l} U_{l_1}\cdots U_{l_{n}}
    \big) \,
\Tr(\ga^{\a_1}\cdots\ga^{\a_{n}}\ga^{\mu_1}\cdots\ga^{\mu_n})
\end{align*}
and the result follows from Proposition \ref{calculres}.
\end{proof}

We will use few notations:

 If $n\in \N$, $q\geq 2$, $l:=(l_1,\cdots,l_{q-1})\in
(\Z^n)^{q-1}$, $\a:=(\a_1,\cdots,\a_q)\in \{1,\cdots,n\}^q$, $k\in
\Z^n \backslash
\{0\}$, $\sigma\in \{-,+\}^q$, $(a_i)_{1\leq i\leq n}\in
(\mathcal{S}(\Z^n))^n$,
\begin{align*}
&l_q:=-\sum_{1\leq j\leq q-1} l_j \, , \quad
\lambda_\sigma:=(-i)^q\prod_{j=1\ldots
q}\sigma_j \, , \quad
\wt a_{\a,l}:= a_{\a_1,l_1}\ldots a_{\a_q,l_q}\,,\\
& \phi_\sigma(k,l):=\sum_{1\leq j\leq q-1} (\sigma_j-\sigma_q)\,
k.\Th l_j +
\sum_{2\leq
j\leq q-1} \sigma_j \, (l_1+\ldots +l_{j-1}).\Th l_j \, ,\\
& g_\mu(s,k,l):=\tfrac{k_{\mu_1}(k+l_1)_{\mu_2}\ldots (k+l_1+\ldots
+l_{q-1})_{\mu_q}}{|k|^{s+2}|k+l_1|^2\ldots|k+l_1+\ldots+l_{q-1}|^2}
\, ,
\end{align*}
with the convention $\sum_{2\leq j\leq q-1} = 0$ when $q=2$, and
$g_\mu(s,k,l)=0$
whenever $\wh l_i=-k$ for a $1\leq i\leq q-1$.

\begin{lemma}\label{formegenerale}
Let $A= L(-iA_{\a})\otimes \gamma^{\a}=-i\sum_{l\in \Z^n}a_{\alpha,l}
\, U_{l}\otimes
\ga^{\alpha}$ where $A_{\a}=-A_{\a}^*\in \A_{\Th}$ and
$\set{a_{\alpha,l}}_{l}\in \SS
(\Z^n)$, with $n\in \N$, be a hermitian one-form, and let $2\leq q
\leq n$, $\sigma\in
\{-,+\}^q$.

Then, $\ncint \mathbb{A}^{\sigma}= \underset{s=0}{\Res}\ f(s)$ where
$$
f(s):=\sum_{l\in (\Z^{n})^{q-1}} {\sum_{k\in\Z^n}}' \,
\lambda_\sigma\  e^{\tfrac i2
\phi_\sigma(k,l)}\ g_\mu(s,k,l)\ \wt a_{\a,l}\
\Tr(\ga^{\a_q}\ga^{\mu_q}\cdots\ga^{\a_1}\ga^{\mu_1}).
$$
\end{lemma}
\begin{proof}
By definition, $\ncint \mathbb{A}^{\sigma}= \underset{s=0}{\Res}\
f(s)$ where
$$
\Tr(\mathbb{A}^{\sigma_q}\cdots \mathbb{A}^{\sigma_1}
|D|^{-s})\sim_c {\sum_{k\in\Z^n}}' \langle U_k\otimes
e^i ,|k|^{-s}\,\mathbb{A}^{\sigma_q}\cdots \mathbb{A}^{\sigma_1}
U_k\otimes e_i\rangle =: f(s).
$$
Let $r\in \Z^n$ and $v\in \C^{2^m}$. Since $A=L(-i A_\a)\otimes
\ga^\a$, and $\epsilon
JAJ^{-1}=R(i A_\a)\otimes \ga^{\a}$, we get
\begin{align*}
\mathbb{A}^{+}U_r\otimes v &= A\DD D^{-2} U_r\otimes
v=A\tfrac{r_\mu}{|r|^2+\delta_{r,0}} U_r \otimes
\ga^{\mu}v=-i\tfrac{r_\mu}{|r|^2+\delta_{r,0}} A_\a U_r \otimes
\ga^\a\ga^{\mu}v\,
,  \\
\mathbb{A}^{-}U_r\otimes v &=\epsilon JAJ^{-1}\DD D^{-2} U_r\otimes
v=\epsilon
JAJ^{-1}\tfrac{r_\mu}{|r|^2+\delta_{r,0}} U_r \otimes
\ga^{\mu}v=i\tfrac{r_\mu}{|r|^2+\delta_{r,0}}U_r A_\a
\otimes \ga^\a\ga^{\mu}v.
\end{align*}
With $U_l U_r=e^{\tfrac i2 r.\Th l} U_{r+l}$ and $U_r U_l=e^{-\tfrac
i2 r.\Th l}
U_{r+l}$, we obtain, for any $1\leq j\leq q$,
$$
\mathbb{A}^{\sigma_j}U_r\otimes v=\sum_{l\in \Z^n} (-\sigma_j)
\, i\, e^{\sigma_j\, \tfrac i2
r.\Th l}\, \tfrac{r_\mu}{|r|^2+\delta_{r,0}}\, a_{\a,l} \, U_{r+l}
\otimes
\ga^{\a}\ga^{\mu}v.
$$
We now apply $q$ times this formula to get
$$
|k|^{-s} \mathbb{A}^{\sigma_q}\cdots \mathbb{A}^{\sigma_1}
U_k\otimes e_i = \lambda_\sigma \sum_{l\in
(\Z^{n})^q} e^{\tfrac i2 \phi_\sigma(k,l)}\ g_\mu(s,k,l)\ \wt
a_{\a,l}\ U_{k+\sum_j
l_j} \otimes \ga^{\a_q}\ga^{\mu_q}\cdots\ga^{\a_1}\ga^{\mu_1}e_i
$$
with
\begin{align*}
\phi_\sigma(k,l)&:=\sigma_1\,  k.\Th l_1+\sigma_2 \, (k+l_1).\Th
l_2+\ldots
+\sigma_q \, (k+l_1+\ldots+l_{q-1}).\Th l_q.
\end{align*}
Thus,
\begin{align*}
f(s) &= {\sum_{k\in\Z^n}}' \, \tau \big(\lambda_\sigma
\sum_{l\in (\Z^{n})^q} e^{\tfrac i2
\phi_\sigma(k,l)}\ g_\mu(s,k,l)\ \wt a_{\a,l}\ U_{\sum_j
l_j}e^{\tfrac i2 k.\Th \sum_j
l_j}\big) \Tr(\ga^{\a_q}\ga^{\mu_q}\cdots\ga^{\a_1}\ga^{\mu_1})\\
&={\sum_{k\in\Z^n}}' \,  \lambda_\sigma \sum_{l\in (\Z^{n})^q}
e^{\tfrac i2
\phi_\sigma(k,l)}\ g_\mu(s,k,l)\ \wt a_{\a,l}\ \delta(\sum_j l_j)
\Tr(\ga^{\a_q}\ga^{\mu_q}\cdots\ga^{\a_1}\ga^{\mu_1})\\
&= {\sum_{k\in\Z^n}}' \,  \lambda_\sigma \sum_{l\in (\Z^{n})^{q-1}}
e^{\tfrac i2
\phi_\sigma(k,l)}\ g_\mu(s,k,l)\ \wt a_{\a,l}\
\Tr(\ga^{\a_q}\ga^{\mu_q}\cdots\ga^{\a_1}\ga^{\mu_1})
\end{align*}
where in the last sum $l_q$ is fixed to $-\sum_{1\leq j\leq q-1} l_j$
and thus,
$$
\phi_\sigma(k,l)=\sum_{1\leq j\leq q-1} (\sigma_j-\sigma_q) \, k.\Th
l_j + \sum_{2\leq
j\leq q-1} \sigma_j \, (l_1+\ldots +l_{j-1}).\Th l_j.
$$
By Lemma \ref{abs-som}, there exists a $R>0$ such that for any $s\in
\C$ with
$\Re(s)>R$, the family
$$\big(e^{\tfrac i2 \phi_\sigma(k,l)}\
g_\mu(s,k,l)\ \wt a_{\a,l}\big)_{(k,l)\in (\Z^n \setminus
\set{0})\times (\Z^{n})^{q-1}}$$
is absolutely summable as a linear combination of families of the
type considered in that lemma. As a consequence, we can
exchange the summations on $k$ and $l$, which gives the result.
\end{proof}

In the following, we will use the shorthand
$$
c:=\tfrac{4\pi^{2}}{3}.
$$

\begin{lemma}\label{Termaterm} Suppose $n=4$. Then, with the same
hypothesis of Lemma
\ref{formegenerale},
\begin{align*}
\hspace{-2cm}\text{(i)} \quad \quad &
\tfrac 12 \ncint (\mathbb A^+)^2= \tfrac 12 \ncint (\mathbb A^-)^2= c
\,
\sum_{l\in\Z^4} \, a_{\alpha_{1},l} \, a_{\alpha_{2},-l} \,
\big(l^{\alpha_{1}}l^{\alpha_{2}}
- \delta^{\alpha_{1}\alpha_{2}} \vert l \vert^2\big).\\
\hspace{-2cm}\text{(ii)} \quad  \quad & \hspace{-0.45cm}
-\tfrac 13\ \ncint (\mathbb A^+)^3=-\tfrac 13 \ncint (\mathbb
A^-)^3=4c
\,\sum_{l_i \in \Z^4}
a_{\a_3,-l_1-l_2}\,a^{\a_1}_{l_2}\,a_{\a_1,l_1}\ \sin \tfrac{l_1.\Th
l_2}{2}\,l_1^{\a_3}.\\
\hspace{-1cm}\text{(iii)} \quad \quad &
\tfrac 14 \ncint (\mathbb A^+)^4=\tfrac 14 \ncint (\mathbb A^-)^4=
2c\,
\sum_{l_i \in \Z^4} a_{\alpha_{1},-l_1-l_2-l_3}\,
a_{{\alpha_{2}},l_3} \, a^{\alpha_{1}}_{l_2}
\, a^{\alpha_{2}}_{l_1} \sin \tfrac{l_1 .\Th (l_2+l_3)}{2}\, \sin
\tfrac{ l_2 .\Th
l_3}{2}.
\end{align*}
\vspace{-0.3cm}

(iv) Suppose $\tfrac {1}{2\pi}\Th$ diophantine. Then the
crossed terms in
$\ncint (\mathbb A^+ + \mathbb A^-)^q$ vanish:  if $C$ is the set of
all $\sigma\in \{-,+\}^q$ with $2\leq q\leq 4$, such that there exist
$i,j$ satisfying
$\sigma_i \neq \sigma_j$, we have $ \sum_{\sigma \in C} \, \ncint
\mathbb{A}^{\sigma} =0.
$
\end{lemma}
\begin{proof}
$(i)$ Lemma \ref{formegenerale} entails that
$\ncint \mathbb{A}^{++}=
\underset{s=0}{\Res} \sum_{l\in \Z^n} - f(s,l)$ where
$$
f(s,l):=  {\sum_{k\in\Z^n}}'
\tfrac{k_{\mu_1}(k+l)_{\mu_2}}{|k|^{s+2}|k+l|^2}\ \wt
a_{\a,l}\ \Tr(\ga^{\a_2}\ga^{\mu_2}\ga^{\a_1}\ga^{\mu_1}) \, \text{
and
} \, \wt a_{\a,l}:=a_{\alpha_1,l}\, a_{\alpha_2,-l} \,.
$$
We will now reduce the computation of the residue of an expression
involving terms
like $\vert k+l\vert^{2}$ in the denominator to the computation of
residues of zeta
functions. To proceed, we use (\ref{trick-0}) into an expression like
the one
appearing in $f(s,l)$. We see that the last term on the righthandside
yields a
$Z_{n}(s)$ while the first one is less divergent by one power of $k$.
If this is not
enough, we repeat this operation for the new factor of $\vert
k+l\vert^{2}$ in the
denominator. For $f(s,l)$, which is quadratically divergent at $s=0$,
we have to repeat
this operation three times before ending with a convergent result.
All the remaining
terms are expressible in terms of $Z_{n}$ functions. We get, using
three times
(\ref{trick-0}),
\begin{equation}\label{trick-1}
\tfrac{1}{|k+l|^2}=\tfrac{1}{|k|^2} - \tfrac{2k.l+|l|^2}{|k|^4} +
\tfrac{(2k.l+|l|^2)^2}{|k|^6} - \tfrac{(2k.l+|l|^2)^3}{|k|^6|k+l|^2}
\, .
\end{equation}
Let us define
$$
f_{\a,\mu}(s,l):={\sum_{k\in\Z^n}}'
\tfrac{k_{\mu_1}(k+l)_{\mu_2}}{|k|^{s+2}|k+l|^2}\ \wt a_{\a,l}
$$
so that $f(s,l)=
f_{\a,\mu}(s,l)\Tr(\ga^{\a_2}\ga^{\mu_2}\ga^{\a_1}\ga^{\mu_1})$.
Equation (\ref{trick-1}) gives
$$
f_{\a,\mu}(s,l) = f_1(s,l) - f_2(s,l) + f_3(s,l) - r(s,l)
$$
with obvious identifications. Note that the function
$$
r(s,l) ={\sum_{k\in\Z^n}}'
\tfrac{k_{\mu_1}(k+l)_{\mu_2}(2kl+|l|^2)^3}{|k|^{s+8}|k+l|^2}\ \wt
a_{\a,l}
$$
is a linear combination of functions of the type $H(s,l)$ satisfying
the hypothesis of
Corollary \ref{res-somH}. Thus, $r(s,l)$ satisfies (H1) and with the
previously seen
equivalence relation modulo functions satisfying this hypothesis we
get
$f_{\a,\mu}(s,l) \sim f_1(s,l) - f_2(s,l) + f_3(s,l)$.

Let's now compute $f_1(s,l)$.
$$
f_1(s,l) = {\sum_{k\in\Z^n}}'
\tfrac{k_{\mu_1}(k+l)_{\mu_2}}{|k|^{s+4}}\ \wt a_{\a,l}
= \wt a_{\a,l}\, {\sum_{k\in\Z^n}}'
\tfrac{k_{\mu_1}k_{\mu_2}}{|k|^{s+4}} + 0.
$$
Proposition \ref{res-int} entails that $s\mapsto {\sum_{k\in\Z^n}}'
\tfrac{k_{\mu_1}k_{\mu_2}}{|k|^{s+4}}$ is holomorphic at 0. Thus,
$f_1(s,l)$ satisfies
(H1), and $f_{\a,\mu}(s,l) \sim - f_2(s,l) + f_3(s,l)$.

Let's now compute $f_2(s,l)$ modulo (H1). We get, using several times
Proposition
\ref{res-int},
\begin{align*}
f_2(s,l)&={\sum_{k\in\Z^n}}'
\tfrac{k_{\mu_1}(k+l)_{\mu_2}(2kl+|l|^2)}{|k|^{s+6}}\ \wt
a_{\a,l} = {\sum_{k\in\Z^n}}' \tfrac{(2kl) k_{\mu_1} k_{\mu_2}+(2kl)
k_{\mu_1}
l_{\mu_2}+|l|^2 k_{\mu_1}
k_{\mu_2}+l_{\mu_2}|l|^2k_{\mu_1}}{|k|^{s+6}}\ \wt
a_{\a,l}\\
&\sim 0 + {\sum_{k\in\Z^n}}' \tfrac{(2kl) k_{\mu_1}
l_{\mu_2}}{|k|^{s+6}}\ \wt
a_{\a,l}+{\sum_{k\in\Z^n}}' \tfrac{|l|^2 k_{\mu_1}
k_{\mu_2}}{|k|^{s+6}}\ \wt
a_{\a,l}+0\, .
\end{align*}
Recall that ${\sum}'_{k\in \Z^n} \tfrac{k_ik_j}{|k|^{s+6}} =
\tfrac{\delta_{ij}}{n}
Z_n(s+4)$. Thus,
$$
f_2(s,l) \sim 2 l^i l_{\mu_2} \wt a_{\a,l} \tfrac{\delta_{i\mu_1}}{n}
Z_n(s+4)+
|l|^2\,\wt a_{\a,l}\, \tfrac{\delta_{\mu_1\mu_2}}{n} Z_n(s+4).
$$
Finally, let us compute $f_3(s,l)$ modulo (H1) following the same
principles:
\begin{align*}
f_3(s,l)&={\sum_{k\in\Z^n}}'
\tfrac{k_{\mu_1}(k+l)_{\mu_2}(2kl+|l|^2)^2}{|k|^{s+8}}\
\wt a_{\a,l}\\&= {\sum_{k\in\Z^n}}' \tfrac{(2kl)^2 k_{\mu_1}
k_{\mu_2} + (2kl)^2
k_{\mu_1} l_{\mu_2} + |l|^4 k_{\mu_1}k_{\mu_2} + |l|^4
k_{\mu_1}l_{\mu_2} +
(4kl)|l|^2k_{\mu_1}k_{\mu_2}+(4kl)|l|^2k_{\mu_1}l_{\mu_2}}{|k|^{s+8}}\
\wt
a_{\a,l}\\&\sim 4 l^{i}l^{j} \, {\sum_{k\in\Z^n}}'\tfrac{k_i k_j
k_{\mu_1}
k_{\mu_2}}{|k|^{s+8}} \wt a_{\a,l}+ 0.
\end{align*}
In conclusion,
$$
f_{\a,\mu}(s,l) \sim - \tfrac 14 (2 l_{\mu_1} l_{\mu_2} +|l|^2\,\,
\delta_{\mu_1\mu_2}) \wt a_{\a,l} Z_n(s+4) + 4 l^{i}l^{j}\,\wt
a_{\a,l}
{\sum_{k\in\Z^n}}'\tfrac{k_i k_j k_{\mu_1} k_{\mu_2}}{|k|^{s+8}}=:
g_{\a,\mu}(s,l).
$$
Proposition (\ref{res-int}) entails that $Z_n(s+4)$ and $s\mapsto
{\sum_{k\in\Z^n}}'\tfrac{k_i k_j k_{\mu_1} k_{\mu_2}}{|k|^{s+8}}$
extend
holomorphically in a punctured open disk centered at 0. Thus,
$g_{\a,\mu}(s,l)$
satisfies (H2) and we can apply Lemma \ref{res-som} to get
$$
-\ncint (\mathbb{A}^+)^2= \underset{s=0}{\Res} \sum_{l\in \Z^n}
f(s,l)= \sum_{l\in
\Z^n}\underset{s=0}{\Res}\ g_{\a,\mu}(s,l)
\Tr(\ga^{\a_2}\ga^{\mu_2}\ga^{\a_1}\ga^{\mu_1})=:\sum_{l\in
\Z^n}\underset{s=0}{\Res}\
g(s,l).
$$
The problem is now reduced to the computation of
$\underset{s=0}{\Res}\ g(s,l)$.
Recall that Res$_{s=0} \, Z_{4}(s+4)=2\pi^2$ by (\ref{formule}) or
(\ref{formule1}),
and
$$
{\rm Res}_{s=0}\,{\sum_{k\in\Z^{n}}}'\,\tfrac{k_{i}k_{j}k_{l}k_{m}}
{\vert
k\vert^{s+8}}=(\delta_{ij}\delta_{lm}+\delta_{il}\delta_{jm}
+\delta_{im}\delta_{jl})\,\tfrac{\pi^2}{12}.
$$
Thus,
$$
\underset{s=0}{\Res}\ g_{\a,\mu}(s,l) = -\tfrac{\pi^2}{3}\,\wt
a_{\a,l}\,(l_{\mu_1}l_{\mu_2}+\tfrac 12 |l|^2\delta_{\mu_1\mu_2}).
$$
We will use
\begin{align}
    \label{Wick}
\Tr(\ga^{\mu_{1}}\cdots\ga^{\mu_{2j}})=\Tr(1)\, \sum_{\text{all
pairings of }\set{1\cdots 2j}} s(P) \, \delta_{\mu_{P_{1}}\mu_{P_{2}}}
 \delta_{\mu_{P_{3}}\mu_{P_{4}}}\cdots
 \delta_{\mu_{P_{2j-1}}\mu_{P_{2j}}}
\end{align}
where $s(P)$ is the signature of the permutation $P$ when
$P_{2m-1}<P_{2m}$ for $1 \leq m \leq n$. This gives
\begin{align}
\label{Wick1}
\Tr(\ga^{\alpha_{2}}\ga^{\mu_2}\ga^{\alpha_{1}}\ga^{\mu_1}) = 2^{m}
(\delta^{\alpha_{2}\mu_2}\delta^{\alpha_{1}
\mu_1}-\delta^{\alpha_{1}\alpha_{2}}\delta^{\mu_2\mu_1}+
\delta^{\alpha_{2}\mu_1}\delta^{\mu_2\alpha_{1}}).
\end{align}
Thus,
\begin{align*}
\underset{s=0}{\Res}\ g(s,l)&= -c\,\wt
a_{\a,l}\,(l_{\mu_1}l_{\mu_2}+\tfrac 12
|l|^2\delta_{\mu_1\mu_2}) (\delta^{\alpha_{2}\mu_2}\delta^{\alpha_{1}
\mu_1}-\delta^{\alpha_{1}\alpha_{2}}\delta^{\mu_2\mu_1}+
\delta^{\alpha_{2}\mu_1}\delta^{\mu_2\alpha_{1}})\\
&=-2c\,\wt a_{\a,l}\,(l^{\a_1}l^{\a_2}-\delta^{\a_1\a_2} |l|^2).
\end{align*}
Finally,
$$
\tfrac 12 \ncint (\mathbb{A}^+)^2= \tfrac 12 \ncint (\mathbb{A}^-)^2=
c \,
\sum_{l\in\Z^n} \, a_{\alpha_{1},l} \, a_{\alpha_{2},-l} \,
\big(l^{\alpha_{1}}l^{\alpha_{2}} - \delta^{\alpha_{1}\alpha_{2}}
\vert l \vert
^2\big).
$$

$(ii)$ Lemma \ref{formegenerale} entails that $\ncint
\mathbb{A}^{+++}=
\underset{s=0} {\Res} \sum_{(l_1,l_2)\in (\Z^n)^2} f(s,l)$ where
\begin{align*}
f(s,l)&:=  {\sum_{k\in\Z^n}}' i\,e^{\tfrac i2 l_1\Th
l_2}\,\tfrac{k_{\mu_1}(k+l_1)_{\mu_2}(k+\wh
l_2)_{\mu_3}}{|k|^{s+2}|k+l_1|^2|k+\wh
l_2|^2}\ \wt a_{\a,l}
\Tr(\ga^{\a_3}\ga^{\mu_3}\ga^{\a_2}\ga^{\mu_2}\ga^{\a_1}\ga^{\mu_1})\\
&=:f_{\a,\mu}(s,l)\Tr(\ga^{\a_3}\ga^{\mu_3}\ga^{\a_2}\ga^{\mu_2}\ga^{\a_1}\ga^{\mu_1}),
\end{align*}
and $\wt a_{\a,l}:=a_{\alpha_1,l_1}\, a_{\alpha_2,l_2}\,
a_{\alpha_3,-\wh l_2}$ with
$\wh l_2:=l_1+l_2$.

We use the same technique as in $(i)$:
\begin{align*}
\tfrac{1}{|k+l_1|^2}&=\tfrac{1}{|k|^2} -
\tfrac{2k.l_1+|l_1|^2}{|k|^4} +
\tfrac{(2k.l_1+|l_1|^2)^2}{|k|^4|k+l_1|^2}\, , \\
\tfrac{1}{|k+\wh l_2|^2}&=\frac{1}{|k|^2} - \tfrac{2k.\wh l_2+|\wh
l_2|^2}{|k|^4} +
\tfrac{(2k.\wh l_2+|\wh l_2|^2)^2}{|k|^4|k+\wh l_2|^2}
\end{align*}
and thus,
\begin{equation}\label{trick-2}
\tfrac{1}{|k+l_1|^2|k+\wh l_2|^2}=\tfrac{1}{|k|^4}
-\tfrac{2k.l_1}{|k|^6}-\tfrac{2k.\wh l_2}{|k|^6} +R(k,l)
\end{equation}
where the remain $R(k,l)$ is a term of order at most $-6$ in $k$.
Equation
(\ref{trick-2}) gives
$$
f_{\a,\mu}(s,l) = f_1(s,l) + r(s,l)
$$
where $r(s,l)$ corresponds to $R(k,l)$. Note that the function
$$
r(s,l) ={\sum_{k\in\Z^n}}'i\,e^{\tfrac i2 l_1\Th l_2}\,
\tfrac{k_{\mu_1}(k+l)_{\mu_2}(k+\wh l_2)_{\mu_3}R(k,l)}{|k|^{s+2}}\
\wt a_{\a,l}
$$
is a linear combination of functions of the type $H(s,l)$ satisfying
the hypothesis of
Corollary (\ref{res-somH}). Thus, $r(s,l)$ satisfies (H1) and
$f_{\a,\mu}(s,l) \sim
f_1(s,l)$.

Let us compute $f_1(s,l)$ modulo (H1)
\begin{align*}
f_1(s,l) &= {\sum_{k\in\Z^n}}'i\,e^{\tfrac i2 l_1\Th l_2}\,
\tfrac{k_{\mu_1}(k+l_1)_{\mu_2}(k+\wh l_2)_{\mu_3}}{|k|^{s+6}}\ \wt
a_{\a,l} -
{\sum_{k\in\Z^n}}' i\,e^{\tfrac i2 l_1\Th
l_2}\,\tfrac{k_{\mu_1}(k+l_1)_{\mu_2}(k+\wh
l_2)_{\mu_3}(2k. l_1+2k.\wh
l_2)}{|k|^{s+8}}\ \wt a_{\a,l}\\
&\sim {\sum_{k\in\Z^n}}'i\,e^{\tfrac i2 l_1\Th
l_2}\,\tfrac{k_{\mu_1}k_{\mu_2} \wh
{l_2}_{\mu_3}+k_{\mu_1}k_{\mu_3}{ l_1}_{\mu_2}}{|k|^{s+6}}\ \wt
a_{\a,l}-
{\sum_{k\in\Z^n}}' i\,e^{\tfrac i2 l_1\Th
l_2}\,\tfrac{k_{\mu_1}k_{\mu_2}k_{\mu_3}(2k.l_1+2k.\wh
l_2)}{|k|^{s+8}}\
\wt a_{\a,l}\\
&= i\,e^{\tfrac i2 l_1\Th l_2}\,\wt a_{\a,l}\big( ({l_{1}}_{\mu_2}
\delta_{\mu_1
\mu_3}+\wh {l_2}_{\mu_3} \delta_{\mu_1 \mu_2}) \, \,\tfrac 14
Z_{4}(s+4) -
2(l_{1}^{i}+\wh l_{2}^{i}){\sum_{k\in\Z^{n}}}'
\tfrac{k_{\mu_1}k_{\mu_2}
k_{\mu_3}k_{i}}{\vert k\vert^{s+8}}\big)\\&=:g_{\a,\mu}(s,l).
\end{align*}
Since $g_{\a,\mu}(s,l)$ satisfies (H2), we can apply Lemma
\ref{res-som} to
get
\begin{align*}
\ncint (\mathbb{A^+})^3&= \underset{s=0}{\Res} \sum_{(l_1,l_2)\in
(\Z^n)^2} f(s,l)\\
&= \sum_{(l_1,l_2)\in (\Z^n)^2}\underset{s=0}{\Res}\ g_{\a,\mu}(s,l)
\Tr(\ga^{\a_3}\ga^{\mu_3}\ga^{\a_2}\ga^{\mu_2}\ga^{\a_1}\ga^{\mu_1})=:\sum_{l}
X_l.
\end{align*}
Recall that $l_3:=-l_1-l_2=-\wh l_2$. By (\ref{formule1}) and
(\ref{formule2}),
\begin{align*} \underset{s=0}{\Res}\ g_{\a,\mu}(s,l) &i\,e^{\tfrac i2
l_1\Th l_2}\,\wt a_{\a,l}\big(
2(-l_{1}^{i}+l_{3}^{i})\tfrac{\pi^2}{12} (\delta_{\mu_1
\mu_2}\delta_{\mu_3
i}+\delta_{\mu_1 \mu_3}\delta_{\mu_2 i}+\delta_{\mu_1 i}\delta_{\mu_2
\mu_3})\\ &+
({l_{1}}_{\mu_2}\delta_{\mu_1 \mu_3} -{l_{3}}_{\mu_3}\delta_{\mu_1
\mu_2}) \tfrac
{\pi^2}{2}\big).
\end{align*}
We decompose $X_l$ in five terms:
$X_l= 2^m\ \tfrac{\pi^2}{2} \ i\,e^{\tfrac i2 l_{1}\Th l_{2}}\, \wt
a_{\a,l}\,(T_1+T_2+T_3+T_4+T_5)$
where
\begin{align*}
T_0 &:= \tfrac 13 (-l_1^{i}+l_3^{i})(\delta_{\mu\nu}\delta_{\rho
i}+\delta_{\mu\rho}\delta_{\nu i}+\delta_{\mu i}\delta_{\nu\rho})
+{l_1}_\nu
\delta_{\mu\rho} -{l_3}_\rho \delta_{\mu\nu},\\
T_1&:=(\delta^{\a_3\rho}\delta^{\a_2\nu}\delta^{\a_1\mu}
-\delta^{\a_3\rho}\delta^{\a_2\a_1}\delta^{\mu\nu}+
\delta^{\a_3\rho}\delta^{\a_2\mu}\delta^{\a_1\nu})T_0,\\
T_2&:=(-\delta^{\a_2\a_3}\delta^{\rho\nu}\delta^{\a_1\mu}
+\delta^{\a_2\a_3}\delta^{\a_1\rho}\delta^{\mu\nu}-
\delta^{\a_2\a_3}\delta^{\rho\mu}\delta^{\a_1\nu})T_0,\\
T_3&:=(\delta^{\a_3\nu}\delta^{\a_2\rho}\delta^{\a_1\mu}
-\delta^{\a_3\nu}\delta^{\a_1\rho}\delta^{\a_2\mu}+
\delta^{\a_3\nu}\delta^{\rho\mu}\delta^{\a_1\a_2})T_0,\\
T_4&:=(-\delta^{\a_1\a_3}\delta^{\a_2\rho}\delta^{\mu\nu}
+\delta^{\a_1\a_3}\delta^{\rho\nu}\delta^{\a_2\mu}-
\delta^{\a_1\a_3}\delta^{\rho\mu}\delta^{\a_2\nu})T_0,\\
T_5&:=(\delta^{\a_3\mu}\delta^{\a_2\rho}\delta^{\a_1\nu}
-\delta^{\a_3\mu}\delta^{\rho\nu}\delta^{\a_1\a_2}+
\delta^{\a_3\mu}\delta^{\a_1\rho}\delta^{\a_2\nu})T_0 .
\end{align*}
With the shorthand $p:=-l_1-2l_3$, $q:=2l_1+l_3$,
$r:=-p-q=-l_1+l_3$, we compute each
$T_i$, and find
\begin{align*}
3T_1&=\delta^{\a_1\a_2}(2-2^m) p^{\a_3}
+ \delta^{\a_3\a_1} q^{\a_2}-\delta^{\a_2\a_1}
q^{\a_3}+\delta^{\a_3\a_2}q^{\a_1} + \delta^{\a_3\a_2}r^{\a_1}
-\delta^{\a_2\a_1}r^{\a_3}+\delta^{\a_3\a_1}r^{\a_2},\\
3T_2&=(2^m-2)\delta^{\a_2\a_3}p^{\a_1} -2^m \delta^{\a_2\a_3}q^{\a_1}
-2^m
\delta^{\a_2\a_3}r^{\a_1},\\
3T_3&=\delta^{\a_1\a_3}p^{\a_2}-\delta^{\a_2\a_3}p^{\a_1}
+\delta^{\a_1\a_2}p^{\a_3}+2^m
\delta^{\a_2\a_1}q^{\a_3}+\delta^{\a_3\a_2}r^{\a_1}-\delta^{\a_3\a_1}r^{\a_2}
+\delta^{\a_1\a_2}r^{\a_3},\\
3T_4&=-\delta^{\a_1\a_3}2^m p^{\a_2}-\delta^{\a_1\a_3}2^m q^{\a_2}
+\delta^{\a_1\a_3}(2^m-2)r^{\a_2},\\
3T_5&=\delta^{\a_1\a_3}p^{\a_2}-\delta^{\a_1\a_2}p^{\a_3}
+\delta^{\a_3\a_2}p^{\a_1}+\delta^{\a_3\a_2}q^{\a_1}
-\delta^{\a_1\a_2}q^{\a_3}+\delta^{\a_3\a_1}q^{\a_2}
+(2-2^m)\delta^{\a_1\a_2}r^{\a_3}.
\end{align*}
Thus,
\begin{equation}
X_l = 2^m\ \tfrac{2\pi^2}{3} i\, e^{\tfrac i2 l_{1}.\Th l_{2}}\, \wt
a_{\a,l}
\,(q^{\a_3}\delta^{\a_1\a_2}
+r^{\a_2}\delta^{\a_1\a_3}+p^{\a_1}\delta^{\a_2\a_3})
\label{formuleA3-1}
\end{equation}
and
\begin{align*}
\ncint (\mathbb{A}^+)^3=i\,2 c\,  (S_1+S_2+S_3),
\end{align*}
where $S_1$, $S_2$ and $S_3$ correspond to respectively
$q^{\a_3}\delta^{\a_1\a_2}$,
$r^{\a_2}\delta^{\a_1\a_3}$ and $p^{\a_1}\delta^{\a_2\a_3}$. In
$S_1$, we permute the
$l_i$ variables the following way: $l_1\mapsto l_3$, $l_2\mapsto
l_1$, $l_3\mapsto
l_2$. Therefore, $l_3.\Th\, l_1 \mapsto l_3.\Th\, l_1$ and $q \mapsto
r$. With a
similar permutation of the $\a_i$, we see that $S_1=S_2$. We apply
the same principles
to prove that $S_1=S_3$ (using permutation $l_1\mapsto l_2$,
$l_2\mapsto l_3$,
$l_3\mapsto l_1$). Thus,
$$
\tfrac 13\ \ncint (\mathbb{A}^+)^3= i\ 2c\ \sum_{l_i} \wt a_{\a,l}\,
 e^{\tfrac i2 l_{1}.\Th l_{2}}\ (l_1-l_2)^{\a_3}
\delta^{\a_1\a_2}= S_4-S_5,
$$
where $S_4$ correspond to $l_1$ and $S_5$ to $l_2$. We permute the
$l_i$ variables in
$S_5$ the following way: $l_1\mapsto l_2$, $l_2\mapsto l_1$,
$l_3\mapsto l_3$, with a
similar permutation on the $\a_i$. Since $l_1.\Th\, l_2 \mapsto
-l_1.\Th\, l_2$, we
finally get
$$
\tfrac 13\ \ncint (\mathbb{A}^+)^3=-4c \sum_{l_i}
a_{\a_1,l_1}\,a_{\a_2,l_2}\,
a_{\a_3,-l_1-l_2}\ \sin \tfrac{l_1.\Th l_2}{2}\ l_1^{\a_3} \,
\delta^{\a_1\a_2}.
$$

$(iii)$ Lemma \ref{formegenerale} entails that
$\ncint \mathbb{A}^{++++}= \underset{s=0}{\Res} \sum_{(l_1,l_2,l_3)\in
(\Z^n)^3} f_{\mu,\a}(s,l) \Tr\ga^{\mu,\a}$
where
\begin{align*}
&\theta:=l_1.\Th l_2+ l_1.\Th l_3 + l_2. \Th l_3,\\
&\Tr\ga^{\mu,\a}:= \Tr(\ga^{\a_4}\ga^{\mu_4}\ga^{\a_3}\ga^{\mu_3}
\ga^{\a_2}\ga^{\mu2}\ga^{\a_1}\ga^{\mu_1}),\\
&f_{\mu,\a}(s,l):= {\sum_{k\in\Z^n}}'e^{\tfrac i2 \theta}\,
\tfrac{k_{\mu_1}(k+l_1)_{\mu_2}(k+\wh l_2)_{\mu_3}(k+\wh
l_3)_{\mu_4}}{|k|^{s+2}|k+l_1|^2|k+\wh l_2|^2|k+\wh l_3|^2}\ \wt
a_{\a,l},\\
&\wt a_{\a,l}:=a_{\alpha_1,l_1}\, a_{\alpha_2,l_2}\,
a_{\alpha_3,l_3}\,a_{\alpha_4,-l_1-l_2-l_3}.
\end{align*}

Using (\ref{trick-0}) and Corollary \ref{res-somH} successively, we
find
$$
f_{\mu,\a}(s,l)\sim  {\sum_{k\in\Z^{n}}}'e^{\tfrac i2 \theta}\,
\tfrac{k_{\mu_1}k_{\mu_2}k_{\mu_3}k_{\mu_4}}{\vert k\vert^{s+2}\vert
k+l_{1}\vert^{2}{\vert k+l_{1}+l_{2}\vert^2} \vert k+l_1+l_2+l_3
\vert^2} \, \wt
a_{\a,l} \sim {\sum_{k\in\Z^{n}}}' e^{\tfrac i2 \theta}\,
\tfrac{k_{\mu_1}k_{\mu_2}k_{\mu_3}k_{\mu_4}}{\vert k\vert^{s+8}} \,
\wt a_{\a,l}.
$$

Since the function ${\sum_{k\in\Z^{n}}}'e^{\tfrac i2 \theta}\,
\tfrac{k_{\mu_1}k_{\mu_2}k_{\mu_3}k_{\mu_4}}{\vert k\vert^{s+8}} \,
\wt a_{\a,l}$
satisfies (H2), Lemma \ref{res-som} entails that
$$
\ncint (\mathbb{A}^+)^4= \sum_{(l_1,l_2,l_3)\in (\Z^n)^3} e^{\tfrac
i2 \theta}\,\wt
a_{\a,l}\,\underset{s=0}{\Res}\ {\sum_{k\in\Z^{n}}}'
\tfrac{k_{\mu_1}k_{\mu_2}k_{\mu_3}k_{\mu_4}}{\vert k\vert^{s+8}}
\Tr\ga^{\mu,\a}
=:\sum_l X_l.
$$
Therefore, with (\ref{formule2}), we get
$X_l =\tfrac{\pi^2}{12}\, \wt a_{\a,l}\, e^{\tfrac i2 \th}\,
(A+B+C)$,
where
\begin{align*}
A&:=\Tr(\ga^{\a_4}\ga^{\mu_4}\ga^{\a_3}
\ga_{\mu_4}\ga^{\a_2}\ga^{\mu_2}\ga^{\a_1}\ga_{\mu_2}),\\
B&:=\Tr(\ga^{\a_4}\ga^{\mu_4}\ga^{\a_3}\ga^{\mu_2}\ga^{\a_2}\ga_{\mu_4}\ga^{\a_1}\ga_{\mu_2}),\\
C&:=\Tr(\ga^{\a_4}\ga^{\mu_4}\ga^{\a_3}
\ga_{\mu_2}\ga^{\a_2}\ga^{\mu_2}\ga^{\a_1}\ga_{\mu_4}).
\end{align*}
Using successively $\{\gamma^{\mu},\gamma^{\nu}\}=2\delta^{\mu\nu}$
and
$\gamma^\mu\gamma_\mu=2^m\ 1_{2^m}$, we see that
\begin{align*}
A&=C=4\ \Tr(\ga^{\a_4}\ga^{\a_3}\ga^{\a_2}\ga^{\a_1}),\\
B&=-4\ \big(\Tr(\ga^{\a_4}\ga^{\a_3}\ga^{\a_1}\ga^{\a_2}) +
\Tr(\ga^{\a_4}\ga^{\a_2}\ga^{\a_3}\ga^{\a_1})\big).
\end{align*}
Thus,
$
A+B+C=8\ 2^m \big( \delta^{\a_4\a_3}\delta^{\a_2\a_1}
+\delta^{\a_4\a_1}\delta^{\a_3\a_2}-2\delta^{\a_4\a_2}\delta^{\a_3\a_1}\big),
$
and
\begin{equation}
X_l =\tfrac{2\pi^2}{3}\ 2^m\ e^{\tfrac i2 \th}\, \wt a_{\a,l}\,  \big(
\delta^{\a_4\a_3}\delta^{\a_2\a_1} +\delta^{\a_4\a_1}\delta^{\a_3\a_2}
-2\delta^{\a_4\a_2}\delta^{\a_3\a_1}\big).\label{formuleA4c}
\end{equation}
By (\ref{formuleA4c}), we get
$$
\ncint (\mathbb{A}^+)^4 = 2c\ (-2T_1 +T_2+T_3),
$$
where
\begin{align*}
T_1&:=\sum_{l_1,\ldots,l_4}a_{\a_4,l_4}\,a_{\a_3,l_3}\, a_{\a_2,l_2}
\,a_{\a_1,l_1} \,
e^{\tfrac i2 \th}\
\delta_{0,\sum_i l_i}\ \delta^{\a_4\a_2}\,\delta^{\a_3\a_1},\\
T_2&:=\sum_{l_1,\ldots,l_4}a_{\a_4,l_4}\,
a_{\a_3,l_3}\,a_{\a_2,l_2}\,a_{\a_1,l_1} \,
e^{\tfrac i2 \th}\
\delta_{0,\sum_i l_i}\ \delta^{\a_4\a_3}\,\delta^{\a_2\a_1},\\
T_3&:=\sum_{l_1,\ldots,l_4}a_{\a_4,l_4}\,a_{\a_3,l_3}\,a_{\a_2,l_2}\,a_{\a_1,l_1}
\,
e^{\tfrac i2\th}\ \delta_{0,\sum_i l_i}\
\delta^{\a_4\a_1}\,\delta^{\a_3\a_2}.
\end{align*}
We now proceed to the following permutations
of the $l_i$ variables in the $T_1$ term :
$l_1\mapsto l_2$, $l_2 \mapsto l_1$,
$l_3 \mapsto l_4$, $l_4 \mapsto l_3$. While
$\sum_i l_i$ is invariant, $\th$ is modified :
$\th \mapsto l_2 .\Th l_1 + l_2 .\Th
l_4 + l_1 .\Th  l_4$. With $\delta_{0,\sum_i l_i}$
in factor, we can let $l_4$ be
$-l_1-l_2-l_3$, so that $\th \mapsto -\th$. We also
permute the $\a_i$ in the same
way. Thus,
$$
T_1=\sum_{l_1,\ldots,l_4}a_{\a_3,l_3}\,a_{\a_4,l_4} \, a_{\a_1,l_1}
\,a_{\a_2,l_2} \,
e^{-\tfrac i2 \th}\ \delta_{0,\sum_i l_i}\ \delta^{\a_3\a_1}
\,\delta^{\a_4\a_2}.
$$
Therefore,
\begin{equation}
2T_1 = 2\sum_{l_1,\ldots,l_4}a_{\a_4,l_4} \, a_{\a_3,l_3}\,
a_{\a_2,l_2}
\,a_{\a_1,l_1}\ \cos \tfrac{\th}{2}\ \delta_{0,\sum_i l_i}\
\delta^{\a_4\a_2} \,
\delta^{\a_3\a_1}\label{formuleT1}.
\end{equation}
The same principles are applied to $T_2$ and $T_3$.
Namely, the permutation
$l_1\mapsto l_1$, $l_2\mapsto l_3$, $l_3
\mapsto l_2$, $l_4 \mapsto l_4$ in $T_2$ and
the permutation $l_1\mapsto l_2$,
$l_2\mapsto l_3$, $l_3 \mapsto l_1$, $l_4 \mapsto
l_4$ in $T_3$ (the $\a_i$ variables are permuted the same way) give
\begin{align*}
T_2 &= \sum_{l_1,\ldots,l_4}a_{\a_4,l_4}a_{\a_3,l_3} a_{\a_2,l_2} \,
a_{\a_1,l_1} \,
e^{\tfrac i2 \phi}\
\delta_{0,\sum_i l_i}\ \delta^{\a_4\a_2}\,\delta^{\a_3\a_1},\\
T_3 &= \sum_{l_1,\ldots,l_4}a_{\a_4,l_4} \, a_{\a_3,l_3} a_{\a_2,l_2}
\, a_{\a_1,l_1}
\,e^{-\tfrac i2 \phi}\ \delta_{0,\sum_i l_i}\
\delta^{\a_4\a_2}\,\delta^{\a_3\a_1}
\end{align*}
where $\phi:=l_1 .\Th\ l_2 + l_1 .\Th\ l_3 - l_2 .\Th\ l_3$. Finally,
we get
\begin{align}
\ncint (\mathbb{A}^+)^4&=4c\ \sum_{l_1,\ldots,l_4}a_{{\a_{1}},l_4} \,
a_{{\a_{2}},l_3}\,
a^{\a_{1}}_{l_2} \, a^{\a_{2}}_{l_1} \, \delta_{0,\sum_i l_i}(\cos
\tfrac{\phi}{2}
-\cos \tfrac{\th}{2})\nonumber\\
&=8c\ \sum_{l_1,\ldots,l_3}a_{{\a_{1}},-l_1-l_2-l_3}\,
a_{{\a_{2}},l_3} \,
a^{\a_{1}}_{l_2} \, a^{\a_{2}}_{l_1} \sin \tfrac{l_1.\Th
(l_2+l_3)}{2}\ \sin \tfrac{
l_2 .\Th l_3}{2}.\label{formuleA4}
\end{align}

$(iv)$ Suppose $q=2$. By Lemma \ref{formegenerale}, we get
$$
\ncint \mathbb{A}^{\sigma}= \underset{s=0}{\Res} \sum_{l\in \Z^n}
\lambda_\sigma
f_{\a,\mu}(s,l) \Tr(\ga^{\a_2}\ga^{\mu_2}\ga^{\a_1}\ga^{\mu_1})
$$
where
$$
f_{\a,\mu}(s,l):=  {\sum_{k\in\Z^n}}'
\tfrac{k_{\mu_1}(k+l)_{\mu_2}}{|k|^{s+2}|k+l|^2}\,e^{i \eta \,  k.\Th
l}
\,\wt a_{\a,l}\,
$$
and $\eta:=\half (\sigma_1-\sigma_2) \in \{-1,1\}$. As in the proof of
$(i)$, since the presence of the phase does not change the
fact that $r(s,l)$ satisfies (H1), we get
$$
f_{\a,\mu}(s,l) \sim f_1(s,l) - f_2(s,l) + f_3(s,l)
$$
where
\begin{align*}
f_1(s,l) &= {\sum_{k\in\Z^n}}'
\tfrac{k_{\mu_1}(k+l)_{\mu_2}}{|k|^{s+4}}\,e^{i \eta \,
k.\Th l}\, \wt a_{\a,l},\\
f_2(s,l)&={\sum_{k\in\Z^n}}'
\tfrac{k_{\mu_1}(k+l)_{\mu_2}(2k.l+|l|^2)}{|k|^{s+6}}\,e^{i \eta \,
k.\Th l} \,  \wt
a_{\a,l} ,\\
f_3(s,l)&={\sum_{k\in\Z^n}}'
\tfrac{k_{\mu_1}(k+l)_{\mu_2}(2k.l+|l|^2)^2}{|k|^{s+8}}\,e^{i \eta \,
k.\Th l} \, \wt
a_{\a,l}.
\end{align*}
Suppose that $l=0$. Then $f_2(s,0)=f_3(s,0)=0$ and Proposition
\ref{res-int} entails
that
\begin{align*}
f_1(s,0)&= {\sum}_{k\in\Z^n}' \tfrac{k_{\mu_1}k_{\mu_2}}{|k|^{s+4}}\,
\wt a_{\a,0}
\end{align*}
is holomorphic at 0 and so is $f_{\a,\mu}(s,0)$.

Since $\tfrac {1}{2\pi}\Th$
is diophantine, Theorem \ref{analytic} $3$ gives us the result.

Suppose $q=3$. Then Lemma \ref{formegenerale} implies that
$$
\ncint \mathbb{A}^{\sigma}
= \underset{s=0}\Res\ {\sum}_{l\in (\Z^n)^{2}} \,  f_{\mu,\a}(s,l)\,
\Tr(\ga^{\mu_3}\ga^{\a_3}\cdots \ga^{\mu_1}\ga^{\a_1})
$$
where
$$
f_{\mu,\a}(s,l):= {\sum}'_{k\in \Z^n}\la_\sigma e^{ik.\Th(\eps_1
l_1+\eps_2 l_2)}
e^{\tfrac i2 \sigma_2 l_1.\Th
l_2}\tfrac{k_{\mu_1}(k+l_1)_{\mu_2}(k+l_1+
l_2)_{\mu_3}}{|k|^{s+2}|k+l_1|^2|k+l_1+l_{2}|^2}\,\wt a_{\a,l},
$$
and $\eps_i :=\half(\sigma_i - \sigma_3)\in \{-1,0,1\}$. By hypothesis
$(\eps_1,\eps_2)\neq (0,0)$. There are six possibilities for the
values of
$(\eps_1,\eps_2)$, corresponding to the six possibilities for the
values of $\sigma$:
$(-,-,+)$, $(-,+,+)$, $(+,-,+)$, $(+,+,-)$, $(-,+,-)$, and $(+,-,-)$.
As in $(ii)$, we
see that
\begin{align*}
f_{\mu,\a}(s,l)&\sim\big( {\sum_{k\in\Z^n}}' \tfrac{e^{ik.\Th(\eps_1
l_1+\eps_2
l_2)}k_{\mu_1}(k+l_1)_{\mu_2}(k+\wh l_2)_{\mu_3}}{|k|^{s+6}}\\
& \hspace{1cm}- {\sum_{k\in\Z^n}}'
\tfrac{e^{ik.\Th(\eps_1 l_1+\eps_2 l_2)}k_{\mu_1}(k+l_1)_{\mu_2}(k+\wh
l_2)_{\mu_3}(2k. l_1+2k.\wh l_2)}{|k|^{s+8}}\, \la_\sigma \,\wt
a_{\a,l}\,e^{\tfrac i2
\sigma_2 l_1.\Th l_2}.
\end{align*}
With $Z:=\{(l_1,l_2) \ : \ \eps_1 l_1 + \eps_2 l_2=0\}$, Theorem
\ref{analytic} $(iii)$
entails that $\sum_{l\in (\Z^n)^2\setminus Z}f_{\mu,\a}(s,l)$ is
holomorphic at 0.
To conclude we need to prove that
$$\sum_{\sigma} g(\sigma) := \sum_{\sigma} \sum_{l\in
Z}f_{\mu,\a}(s,l)\,
\Tr(\ga^{\mu_3}\ga^{\a_3}\cdots \ga^{\mu_1}\ga^{\a_1})$$ is
holomorphic at 0. By
definition, $\la_\sigma = i\sigma_1\sigma_2\sigma_3$
and as a consequence, we check that
\begin{align*}
g(-,-,+)=-g(+,+,-),\quad
g(+,-,+)=-g(+,-,-),\quad
g(-,+,+)=-g(-,+,-),
\end{align*}
which implies that $\sum_{\sigma} g(\sigma)=0$. The result follows.

Suppose finally that $q=4$. Again, Lemma \ref{formegenerale} implies
that
$$
\ncint \mathbb{A}^{\sigma} = \underset{s=0}\Res\ \sum_{l\in
(\Z^n)^{3}} f_{\mu,\a}(s,l)\,
\Tr(\ga^{\mu_4}\ga^{\a_4}\cdots \ga^{\mu_1}\ga^{\a_1})
$$
where
$$
f_{\mu,\a}(s,l):= {\sum_{k\in \Z^n}}'\la_\sigma\, e^{ik.\Th
\sum_{i=1}^{3}\eps_i l_i}
\, e^{\tfrac i2 (\sigma_2 l_1.\Th l_2+\sigma_3 (l_1+l_2).\Th
l_3)} \, \tfrac{k_{\mu_1}(k+l_1)_{\mu_2}(k+l_1+
l_2)_{\mu_3}(k+l_1+l_2+l_3)_{\mu_4}}
{|k|^{s+2}|k+l_1|^2|k+l_1+l_{2}|^2|k+l_1+l_2+l_3|^2}\, \wt
a_{\a,l}
$$
and $\eps_i :=\half(\sigma_i - \sigma_4)\in \{-1,0,1\}$. By hypothesis
$(\eps_1,\eps_2,\eps_3)\neq (0,0,0)$. There are fourteen
possibilities for the values
of $(\eps_1,\eps_2,\eps_3)$, corresponding to the fourteen
possibilities for the
values of $\sigma$: $(-,-,-,+)$, $(-,-,+,+)$, $(-,+,-,+)$,
$(+,-,-,+)$, $(-,+,+,+)$,
$(+,-,+,+)$, $(+,+,-,+)$, $(+,+,+,-)$, $(-,-,+,-)$, $(-,+,-,-)$,
$(+,-,-,-)$,
$(-,+,+,-)$, $(+,-,+,-)$ and $(+,+,-,-)$. As in $(ii)$, we see that,
with the
shorthand $\th_\sigma:=\sigma_2 l_1.\Th l_2+\sigma_3 (l_1+l_2).\Th
l_3$,
\begin{align*}
f_{\mu,\a}(s,l)\sim {\sum}'_{k\in \Z^n}\la_\sigma \, e^{ik.\Th
\sum_{i=1}^3\eps_i l_i}
\, e^{\tfrac i2 \th_\sigma} \,
\tfrac{k_{\mu_1}k_{\mu_2}k_{\mu_3}k_{\mu_4}}{|k|^{s+8}}\wt
a_{\a,l}=:g_{\mu,\a}(s,l)\, .
\end{align*}
With $Z_\sigma:=\{(l_1,l_2,l_3) \ : \ \sum_{i=1}^3\eps_i l_i=0\}$,
Theorem
\ref{analytic} $(iii)$, the series
$\sum_{l \in (\Z^n)^3 \setminus Z_\sigma}f_{\mu,\a}(s,l)$ is
holomorphic at 0. To conclude, we need to prove that
$$
\sum_{\sigma} g(\sigma) := \sum_{\sigma}\underset{s=0}\Res\
\sum_{l\in Z_\sigma}g_{\mu,\a}(s,l)\,
\Tr(\ga^{\mu_4}\ga^{\a_4}\cdots \ga^{\mu_1}\ga^{\a_1})=0.
$$
Let $C$ be the set of
the fourteen values of $\sigma$ and $C_7$ be the set of the seven
first values of
$\sigma$ given above. Lemma \ref{symetrie} implies
$$
\sum_{\sigma\in C} g(\sigma) = 2\sum_{\sigma\in C_7} g(\sigma).
$$
Thus, in the following, we restrict to these seven values. Let us note
$F_\mu(s):={\sum}'_{k\in
\Z^n}\tfrac{k_{\mu_1}k_{\mu_2}k_{\mu_3}k_{\mu_4}}{|k|^{s+8}}$
so that
$$
g(\sigma)=\underset{s=0} \Res \ F_\mu(s) \, \la_\sigma \, \sum_{l\in
Z_\sigma} e^{\tfrac i2
\th_\sigma} \, \wt a_{\a,l}\, \Tr(\ga^{\mu_4}\ga^{\a_4}\cdots
\ga^{\mu_1}\ga^{\a_1}).
$$
Recall from (\ref{formuleA4c}) that
$$
\underset{s=0}\Res\ F_\mu(s) \Tr(\ga^{\mu_4}\ga^{\a_4}\cdots
\ga^{\mu_1}\ga^{\a_1})
=2c \big( \delta^{\a_4\a_3}\delta^{\a_2\a_1}
+\delta^{\a_4\a_1}\delta^{\a_3\a_2}
-2\delta^{\a_4\a_2}\delta^{\a_3\a_1}\big).
$$
As a consequence, we get, with $\wt a_{\a,l}:= a_{\a_1,l_1}\cdots
a_{\a_4,l_4}$,
\begin{align*}
g(\sigma)&=2c\la_\sigma \sum_{l\in (\Z^n)^{4}} e^{\tfrac i2
\th_\sigma} \, \wt
a_{\a,l} \,\delta_{\sum_{i=1}^4 l_i,0} \,\delta_{\sum_{i=1}^3\eps_i
l_i,0}\big(
\delta^{\a_4\a_3}\delta^{\a_2\a_1} +\delta^{\a_4\a_1}\delta^{\a_3\a_2}
-2\delta^{\a_4\a_2}\delta^{\a_3\a_1}\big)\\
&=:2c\la_\sigma(T_1+T_2-2 T_3).
\end{align*}
We proceed to the following change of variable in $T_1$: $l_1\mapsto
l_1$, $l_2\mapsto
l_3$, $l_3\mapsto l_2$, $l_4\mapsto l_4$. Thus, we get
$\th_\sigma\mapsto
\psi_\sigma:=\sigma_2 l_1.\Th l_3+\sigma_3(l_1+l_3).\Th l_2$, and
$\sum_{i=1}^3\eps_i
l_i \mapsto \eps_1 l_1 + \eps_3 l_2 + \eps_2 l_3=:u_\sigma(l)$. With
a similar
permutation on the $\a_i$, we get
$$
T_1 =\sum_{l\in (\Z^n)^{4}} e^{\tfrac i2 \psi_\sigma}\,\wt
a_{\a,l}\,\delta_{\sum_{i=1}^4 l_i,0} \,\delta_{\eps_1 l_1 + \eps_3
l_2 + \eps_2
l_3,0}\, \delta^{\a_4\a_2}\delta^{\a_3\a_1}.
$$
We proceed to the following change of variable in $T_2$: $l_1\mapsto
l_2$, $l_2\mapsto
l_3$, $l_3\mapsto l_1$, $l_4\mapsto l_4$. Thus, we get
$\th_\sigma\mapsto
\phi_\sigma:=\sigma_2 l_2.\Th l_3+\sigma_3(l_2+l_3).\Th l_1$, and
$\sum_{i=1}^3\eps_i
l_i \mapsto \eps_3 l_1 + \eps_1 l_2 + \eps_2 l_3=:v_\sigma(l)$. After
a similar
permutation on the $\a_i$, we get
$$
T_2 ={\sum}_{l\in (\Z^n)^{4}} \, e^{\tfrac i2 \phi_\sigma}\,\wt
a_{\a,l}\,\delta_{\sum_{i=1}^4 l_i,0} \,\delta_{\eps_3 l_1 + \eps_1
l_2 + \eps_2
l_3,0} \, \delta^{\a_4\a_2}\delta^{\a_3\a_1}.
$$
Finally, we proceed to the following change of variable in $T_3$:
$l_1\mapsto l_2$,
$l_2\mapsto l_1$, $l_3\mapsto l_4$, $l_4\mapsto l_3$. Thus, we get
$\th_\sigma\mapsto
-\th_\sigma$, and $\sum_{i=1}^3\eps_i l_i \mapsto (\eps_2-\eps_3) l_1
+
(\eps_1-\eps_3) l_2 - \eps_3 l_3=:w_\sigma(l)$. With a similar
permutation on the
$\a_i$, we get
$$
T_3 ={\sum}_{l\in (\Z^n)^{4}}\, e^{-\tfrac i2 \th_\sigma}\,\wt
a_{\a,l}\,\delta_{\sum_{i=1}^4 l_i,0} \,\delta_{(\eps_2-\eps_3) l_1 +
(\eps_1-\eps_3)
l_2 - \eps_3 l_3,0}\delta^{\a_4\a_2}\delta^{\a_3\a_1}.
$$
As a consequence, we get
$$
g(\sigma)=2c {\sum}_{l\in (\Z^n)^{4}} K_\sigma(l_1,l_2,l_3)\,\wt
a_{\a,l}\,\delta_{\sum_{i=1}^4
l_i,0}\,\delta^{\a_4\a_2}\delta^{\a_3\a_1},
$$
where
$
K_\sigma(l_1,l_2,l_3)=\la_\sigma\big(e^{\tfrac i2
\psi_\sigma}\,\delta_{u_\sigma(l),0}+e^{\tfrac i2
\phi_\sigma}\,\delta_{v_\sigma(l),0}
-e^{\tfrac i2 \th_\sigma}\, \delta_{\sum_{i=1}^3\eps_i
l_i,0}-e^{-\tfrac i2
\th_\sigma}\, \delta_{w_\sigma(l),0}\big).
$

The computation of $K_\sigma(l_1,l_2,l_3)$ for the seven values of
$\sigma$ yields
\begin{align*}
K_{--++}(l_1,l_2,l_3)&=\delta_{l_1+l_3,0}+\delta_{l_2+l_3,0}-\delta_{l_1+l_2,0}-\delta_{l_1+l_2,0},\\
K_{-+-+}(l_1,l_2,l_3)&=\delta_{l_1+l_2,0}+\delta_{l_1+l_2,0}-\delta_{l_1+l_3,0}-\delta_{l_1+l_3,0},\\
K_{--++}(l_1,l_2,l_3)&=\delta_{l_2+l_3,0}+\delta_{l_1+l_3,0}-\delta_{l_2+l_3,0}-\delta_{l_2+l_3,0},\\
K_{---+}(l_1,l_2,l_3)&=-\big( e^{\tfrac i2 l_1.\Th l_2
}\delta_{\sum_{i=1}^3 l_i,0} +
e^{\tfrac i2 l_2.\Th l_1}\delta_{\sum_{i=1}^3 l_i,0}
- e^{\tfrac i2 l_2.\Th l_1} \delta_{\sum_{i=1}^3 l_i,0}- e^{\tfrac i2
l_1.\Th l_2} \delta_{l_3,0} \big),\\
K_{-+++}(l_1,l_2,l_3)&=-\big(e^{\tfrac i2 l_3.\Th l_2 }\delta_{l_1,0}
+ e^{\tfrac i2
l_3.\Th l_1}\delta_{l_2,0}
- e^{\tfrac i2 l_2.\Th l_3}\delta_{l_1,0} - e^{\tfrac i2 l_3.\Th
l_1}\delta_{l_2,0} \big),\\
K_{+-++}(l_1,l_2,l_3)&=-\big(e^{\tfrac i2 l_1.\Th l_2 }\delta_{l_3,0}
+ e^{\tfrac i2
l_2.\Th l_1} \delta_{l_3,0}-
e^{\tfrac i2 l_1.\Th l_3}\delta_{l_2,0} - e^{\tfrac i2 l_3.\Th
l_2}\delta_{l_1,0} \big),\\
K_{++-+}(l_1,l_2,l_3)&=-\big(e^{\tfrac i2 l_1.\Th l_3 }\delta_{l_2,0}
+ e^{\tfrac i2
l_2.\Th l_3}\delta_{l_1,0} - e^{\tfrac i2 l_1.\Th l_2}\delta_{l_3,0}
- e^{\tfrac i2
l_2.\Th l_1}\delta_{\sum_{i=1}^3 l_i,0} \big).
\end{align*}
Thus,
$$
\sum_{\sigma\in C_7} K_\sigma(l_1,l_2,l_3) =
2i(\delta_{l_3,0}-\delta_{\sum_{i=1}^3
l_i,0})\sin \tfrac {l_1.\Th l_2}{2}
$$
and
$$
\sum_{\sigma\in C_7} g(\sigma)=i 4c  \sum_{l\in (\Z^n)^{4}}
(\delta_{l_3,0}-\delta_{\sum_{i=1}^3 l_i,0})\sin \tfrac {l_1.\Th
l_2}{2}\,\wt
a_{\a,l}\,\delta_{\sum_{i=1}^4
l_i,0}\,\delta^{\a_4\a_2}\delta^{\a_3\a_1}.
$$
The following change of variables: $l_1\mapsto l_2$, $l_1\mapsto
l_2$, $l_3\mapsto
l_4$, $l_4\mapsto l_3$ gives
$$
\sum_{l\in (\Z^n)^{4}} \delta_{\sum_{1}^3 l_i,0}\sin \tfrac {l_1.\Th
l_2}{2}\,\wt
a_{\a,l}\,\delta_{\sum_{1}^4
l_i,0}\,\delta^{\a_4\a_2}\delta^{\a_3\a_1} =-\sum_{l\in
(\Z^n)^{4}} \delta_{l_3,0} \sin \tfrac {l_1.\Th l_2}{2}\,\wt
a_{\a,l}\,\delta_{\sum_{1}^4
l_i,0}\,\delta^{\a_4\a_2}\delta^{\a_3\a_1}
$$
so
$$
\sum_{\sigma\in C_7} g(\sigma)=i8c  \sum_{l\in (\Z^n)^{4}}
\delta_{l_3,0}\,\sin \tfrac
{l_1.\Th l_2}{2}\,\wt a_{\a,l}\,\delta_{\sum_{1}^4
l_i,0}\,\delta^{\a_4\a_2}\delta^{\a_3\a_1}.
$$
Finally, the change of variables: $l_2\mapsto l_4, l_4\mapsto l_2$
gives
$$
\sum_{l\in (\Z^n)^{4}} \delta_{l_3,0}\,\sin \tfrac {l_1.\Th
l_2}{2}\,\wt
a_{\a,l}\,\delta_{\sum_{1}^4
l_i,0}\,\delta^{\a_4\a_2}\delta^{\a_3\a_1} = -\sum_{l\in
(\Z^n)^{4}} \delta_{l_3,0}\,\sin \tfrac {l_1.\Th l_2}{2}\,\wt
a_{\a,l}\,\delta_{\sum_{1}^4
l_i,0}\,\delta^{\a_4\a_2}\delta^{\a_3\a_1}
$$
which entails that $\sum_{\sigma\in C_7} g(\sigma)=0$.
\end{proof}

\begin{lemma}
\label{double}
Suppose $n=4$ and $\tfrac{1}{2\pi}\Th$ diophantine. For any
self-adjoint one-form $A$,
$$
\zeta_{D_A}(0)-\zeta_{D}(0)=-c\,
\tau(F_{\a_1,\a_2}F^{\a_1\a_2}).
$$
\end{lemma}
\begin{proof}
By (\ref{termconstanttilde}) and Lemma \ref{ncadmoins1} we get
$$
\zeta_{D_A}(0)-\zeta_{D}(0) =
\sum_{q=1}^n \tfrac {(-1)^q}{q} \sum_{\sigma\in \set{+,-}^q}\ncint
\mathbb{A}^\sigma.
$$
By Lemma \ref{Termaterm} $(iv)$, we see that the crossed terms all
vanish. Thus,
with Lemma \ref{symetrie}, we get
\begin{equation}
\label{n4eq1}
\zeta_{D_A}(0)-\zeta_{D}(0) = 2 \sum_{q=1}^n
\tfrac {(-1)^q}{q} \ncint (\mathbb{A}^+)^q.
\end{equation}

By definition,
\begin{align*}
F_{\alpha_{1} \alpha_{2}}&=i\sum_{k}\big(a_{\alpha_{2},k}\,
k_{\alpha_{1}}-a_{\alpha_{1},k}\, k_{\alpha_{2}}\big)U_{k}
+\sum_{k,\,l}a_{\alpha_{1},k}\,a_{\alpha_{2},l}\,[U_{k},U_{l}]\\
&=i\sum_{k}\big[(a_{\alpha_{2},k} \, k_{\alpha_{1}} -a_{\alpha_{1},k}
\,
k_{\alpha_{2}}) -2\sum_{l}a_{\alpha_{1},k-l}\,a_{\alpha_{2},l}\,
\sin(\tfrac{k.\Th
l}{2})\big] \, U_{k}.
\end{align*}
Thus
\begin{align*}
\tau(F_{\alpha_{1}\alpha_{2}} F^{\alpha_{1}\alpha_{2}})
&=\sum_{\alpha_{1},\,\alpha_{2}=1}^{2^m}\, \sum_{k\in \Z^4}\big[
(a_{\alpha_{2},k} \, k_{\alpha_{1}}-a_{\alpha_{1},k} \,
k_{\alpha_{2}})
-2\sum_{l'\in\Z^4}a_{\alpha_{1},k-l'}\,a_{\alpha_{2},l'}\,
\sin(\tfrac{k.\Th l'}{2})\big]\\
&\hspace{2.2cm} \big[(a_{\alpha_{2},-k} \, k_{\alpha_{1}}
-a_{\alpha_{1},-k}\, k_{\alpha_{2}})
-2\sum_{l"\in
\Z^4}a_{\alpha_{1},-k-l"}\,a_{\alpha_{2},l"}\,\sin(\tfrac{k.\Th
l"}{2})\big].
\end{align*}
One checks that the term in $a^q$ of $\tau(F_{\alpha_{1}\alpha_{2}}
F^{\alpha_{1}\alpha_{2}})$ corresponds to the term $\ncint
(\mathbb{A}^+)^q$ given by
Lemma \ref{Termaterm}. For $q=2$, this is
$$
-2\sum_{l\in\Z^4,\,\alpha_{1},\,\alpha_{2}} \, a_{\alpha_{1},l} \,
a_{\alpha_{2},-l}
\, \big(l_{\alpha_{1}}l_{\alpha_{2}} - \delta_{\alpha_{1}\alpha_{2}}
\vert l \vert
^2\big).
$$
\noindent For $q=3$, we compute the crossed terms:
$$
i\sum_{k,k',l} (a_{{\a_{2}},k} \, k_{\a_{1}} - a_{{\a_{1}},k} \,
k_{\a_{2}})\
a_{k'}^{\a_1}\ a_{l}^{\a_2} \big(U_k[U_{k'},l]+[U_{k'},U_l]U_k \big),
$$
which gives the following $a^3$-term in
$\tau(F_{\alpha_{1}\alpha_{2}} F^{\alpha_{1}\alpha_{2}})$
$$
-8\sum_{l_i} a_{\a_3,-l_1-l_2}\,a^{\a_1}_{l_2}\,a_{\a_1,l_1}\ \sin
\tfrac{l_1.\Th
l_2}{2}\,l_1^{\a_3} .
$$
\noindent For $q=4$, this is
$$-4\sum_{l_i}a_{\alpha_{1},-l_1-l_2-l_3}\,
a_{{\alpha_{2}},l_3} \, a^{\alpha_{1}}_{l_2} \, a^{\alpha_{2}}_{l_1}
\sin \tfrac{l_1
.\Th (l_2+l_3)}{2}\, \sin \tfrac{ l_2 .\Th l_3}{2}
$$
which corresponds to the term $\ncint (\mathbb{A}^+)^4$.
We get finally,
\begin{equation}\label{n4eq2}
\sum_{q=1}^n
\tfrac {(-1)^q}{q} \ncint (\mathbb{A}^+)^q
=- \tfrac c2 \tau(F_{\a_1,\a_2}F^{\a_1\a_2}).
\end{equation}
Equations (\ref{n4eq1}) and (\ref{n4eq2}) yield the result.
\end{proof}

\begin{lemma}
    \label{term-n=2}
Suppose $n=2$. Then, with the same hypothesis as in Lemma
\ref{formegenerale},
\begin{align*}
\hspace{-5.9cm}\text{(i)} \quad \quad &  \ncint (\mathbb
A^+)^2= \ncint (\mathbb A^-)^2=0.
\end{align*}
(ii) Suppose $\tfrac{1}{2\pi}\Th$ diophantine. Then
$$
\ncint \mathbb A^+ \mathbb A^{-}
=  \ncint \mathbb A^- \mathbb A^+=0.
$$
\end{lemma}
\begin{proof}
$(i)$ Lemma \ref{formegenerale} entails that $\ncint \mathbb{A}^{++}=
\underset{s=0}{\Res} \sum_{l\in \Z^2} - f(s,l)$ where
$$
f(s,l):=  {{\sum}'_{k\in\Z^2}}
\tfrac{k_{\mu_1}(k+l)_{\mu_2}}{|k|^{s+2}|k+l|^2}\ \wt
a_{\a,l}\
\Tr(\ga^{\a_2}\ga^{\mu_2}\ga^{\a_1}\ga^{\mu_1})=:f_{\mu,\a}(s,l)
\Tr(\ga^{\a_2}\ga^{\mu_2}\ga^{\a_1}\ga^{\mu_1})
$$
and $\wt a_{\a,l}:=a_{\alpha_1,l}\, a_{\alpha_2,-l}$. This time,
since $n=2$, it is
enough to apply just once (\ref{trick-0}) to obtain an absolutely
convergent series.
Indeed, we get with (\ref{trick-0})
$$
f_{\mu,\a}(s,l)= {\sum_{k\in\Z^2}}'
\tfrac{k_{\mu_1}(k+l)_{\mu_2}}{|k|^{s+4}}\ \wt
a_{\a,l}- {\sum_{k\in\Z^2}}'
\tfrac{k_{\mu_1}(k+l)_{\mu_2}(2k.l+|l|^2)}{|k|^{s+4}|k+l|^2}\ \wt
a_{\a,l}.
$$
and the function $ r(s,l):= {\sum}'_{k\in\Z^2}
\tfrac{k_{\mu_1}(k+l)_{\mu_2}(2k.l+|l|^2)}{|k|^{s+4}|k+l|^2}\ \wt
a_{\a,l} $ is a
linear combination of functions of the type $H(s,l)$ satisfying the
hypothesis of
Corollary \ref{res-somH}. As a consequence, $r(s,l)$ satisfies (H1)
and
$$
f_{\mu,\a}(s,l)\sim {\sum_{k\in\Z^2}}'
\tfrac{k_{\mu_1}(k+l)_{\mu_2}}{|k|^{s+4}}\ \wt
a_{\a,l}\sim  {\sum_{k\in\Z^2}}'
\tfrac{k_{\mu_1}k_{\mu_2}}{|k|^{s+4}}\ \wt a_{\a,l}
$$ Note that the function $(s,l)\mapsto
h_{\mu,\a}(s,l):={\sum}'_{k\in\Z^2} \tfrac{k_{\mu_1}k_{\mu_2}}
{|k|^{s+4}}\ \wt a_{\a,l}$ satisfies (H2). Thus, Lemma \ref{res-som}
yields
$$
\underset{s=0}{\Res}\ f(s,l) =\sum_{l\in \Z^2} \underset{s=0}{\Res}\
h_{\mu,\a}(s,l)\Tr(\ga^{\a_2}\ga^{\mu_2}\ga^{\a_1}\ga^{\mu_1}).
$$
By Proposition \ref{calculres}, we get $\underset{s=0}{\Res}\
h_{\mu,\a}(s,l) =
\delta_{\mu_1\mu_2}\,\pi\, \wt a_{\a,l}$. Therefore,
$$
\ncint \mathbb{A}^{++}=-\pi \sum_{l\in \Z^2}\wt a_{\a,l}
\Tr(\ga^{\a_2}\ga^{\mu}\ga^{\a_1}\ga_{\mu})=0
$$
according to (\ref{Wick1}).

$(ii)$ By Lemma \ref{formegenerale}, we obtain that
$\ncint \mathbb{A}^{-+}= \underset{s=0}{\Res} \sum_{l\in \Z^2}
\lambda_\sigma
f_{\a,\mu}(s,l) \Tr(\ga^{\a_2}\ga^{\mu_2}\ga^{\a_1}\ga^{\mu_1})$
where $\la_\sigma=-(-i)^2=1$ and
$$
f_{\a,\mu}(s,l):=  {\sum_{k\in\Z^2}}'
\tfrac{k_{\mu_1}(k+l)_{\mu_2}}{|k|^{s+2}|k+l|^2}
\, e^{i \eta \,  k.\Th l} \,\wt a_{\a,l}\,
$$
and $\eta:=\half (\sigma_1-\sigma_2) =-1$. As in the proof of $(i)$,
since
the presence of the phase does not change the fact that $r(s,l)$
satisfies (H1), we
get
$$
f_{\a,\mu}(s,l) \sim {\sum_{k\in\Z^2}}'
\tfrac{k_{\mu_1}(k+l)_{\mu_2}}{|k|^{s+4}}\,e^{i \eta \, k.\Th l}\, \wt
a_{\a,l}:=g_{\a,\mu}(s,l)\, .
$$
Since $\tfrac {1}{2\pi}\Th$ is diophantine,  the functions $
s\mapsto \sum_{l\in\Z^2\backslash \{0\}} g_{\a,\mu}(s,l) $ are
holomorphic at $s=0$ by Theorem \ref{analytic} $3$. As a consequence,
$$
\ncint \mathbb{A}^{-+}=\underset{s=0}\Res\ g_{\a,\mu}(s,0)
\Tr(\ga^{\a_2}\ga^{\mu_2}\ga^{\a_1}\ga^{\mu_1})=\underset{s=0}\Res\
{\sum_{k\in\Z^2}}'
\tfrac{k_{\mu_1}k_{\mu_2}}{|k|^{s+4}}\,\wt
a_{\a,0}\,\Tr(\ga^{\a_2}\ga^{\mu_2}\ga^{\a_1}\ga^{\mu_1}).
$$
Recall from Proposition \ref{res-int} that
$\Res_{s=0}\,{\sum}'_{k\in\Z^2}\,\tfrac{k_{i}k_{j}} {\vert
k\vert^{s+4}}=\delta_{ij}\,\pi$. Thus, again with (\ref{Wick1}),
\begin{align*}
\ncint \mathbb{A}^{-+}=\wt a_{\a,0}\,
\pi\,\Tr(\ga^{\a_2}\ga^{\mu}\ga^{\a_1}\ga_{\mu})=0.
\tag*{\qed}
\end{align*}
\hideqed
\end{proof}

\begin{lemma}
    \label{termeconstantn=2}
Suppose $n=2$ and $\tfrac{1}{2\pi}\Th$ diophantine. For any self-adjoint one-form $A$,
\begin{align*}
\zeta_{D_A}(0)-\zeta_{D}(0)= 0.
\end{align*}
\end{lemma}

\begin{proof}
As in Lemma \ref{double}, we use (\ref{termconstanttilde}) and Lemma
\ref{ncadmoins1}
so the result
follows from Lemma \ref{term-n=2}.
\end{proof}

\subsubsection{Odd dimensional case}

\begin{lemma}
\label{impair}
Suppose $n$ odd and $\tfrac{1}{2\pi}\Th$
diophantine.
Then for any self-adjoint 1-form $A$ and $\sigma\in \{-,+\}^q$ with
$2\leq q\leq n$,
$$
\ncint \mathbb{A}^\sigma = 0\, .
$$
\end{lemma}
\begin{proof} Since $\mathbb{A}^\sigma \in \Psi_1(\A)$, Lemma
\ref{ncint-odd-pdo} with $k=n$ gives the result.
\end{proof}

\begin{corollary}
\label{zetaimpair}
With the same hypothesis of Lemma \ref{impair}, for any self-adjoint
one-form $A$,
$\zeta_{D_A}(0)-\zeta_{D}(0)=0.$
\end{corollary}
\begin{proof} As in Lemma \ref{double}, we use
(\ref{termconstanttilde}) and Lemma \ref{ncadmoins1}
so the result
follows from Lemma \ref{impair}.
\end{proof}

\subsection{Proof of the main result}
\begin{proof}[Proof of Theorem \ref{main}.]
$(i)$ By (\ref{formuleaction}) and Proposition \ref{invariance}, we
get
$$
\SS(\DD_{A},\Phi,\Lambda) \, = \, 4\pi \Phi_{2}\,
    \Lambda^{2}  + \Phi(0) \,
    \zeta_{D_{A}}(0) + \mathcal{O}(\Lambda^{-2}),
$$
where $\Phi_{2}= \half\int_{0}^{\infty} \Phi(t) \, dt$. By Lemma
\ref{termeconstantn=2},
$\zeta_{D_A}(0) - \zeta_{D}(0) = 0$ and from Proposition
\ref{zeta(0)}, $\zeta_{D}(0)=0$,
so we get the result.

$(ii)$ Similarly,
$\SS(\DD_{A},\Phi,\Lambda) \, =  8 \pi^2\, \Phi_{4}\,
    \Lambda^{4} + \Phi(0) \,
    \zeta_{D_{A}}(0) + \mathcal{O}(\Lambda^{-2})$
with $\Phi_{4}= \half\int_{0}^{\infty} \Phi(t) \, t \, dt$.
Lemma \ref{double} implies that
$\zeta_{D_A}(0) - \zeta_D(0)=-c\,\tau(F_{\mu\nu}F^{\mu\nu})$
and by Proposition \ref{zeta(0)},
$\zeta_{D_A}(0)=-c\,\tau(F_{\mu\nu}F^{\mu\nu})$ leading to the
result.

$(iii)$ is a direct consequence of (\ref{formuleaction}),
Propositions \ref{zeta(0)}, \ref{invariance}, and Corollary
\ref{zetaimpair}.
\end{proof}

\appendix
\section{Appendix}

\subsection{Proof of Lemma \ref{propOP}}

$(i)$ We have $|D|T|D|^{-1} = T + \delta(T) |D|^{-1}$
and $|D|^{-1}T|D| =  T  - |D|^{-1}\delta(T)$.
A recurrence proves that for any $k\in \N$,
$|D|^{k}T|D|^{-k}= \sum_{q=0}^k \genfrac{(}{)}{0pt}{1}{q}{k} \,
\delta^q(T)
|D|^{-q}$ and we get $|D|^{-k}T|D|^{k}= \sum_{q=0}^k
(-1)^q\genfrac{(}{)}{0pt}{1}{q}{k} \, |D|^{-q} \delta^q(T)$.

As a consequence, since $T$, $|D|^{-q}$
and $\delta^q(T)$ are in $OP^0$
for any $q\in \N$, for any $k\in \Z$,
$|D|^{k}T|D|^{-k}\in OP^0$. Let us fix $p\in
\N_0$ and define $F_p(s):=\delta^p(|D|^{s}T|D|^{-s})$
for $s\in \C$. Since for $k\in \Z$, $F_p(k)$ is bounded,
a complex interpolation proves that $F_p(s)$ is bounded, which
gives $|D|^{s} T |D|^{-s} \in OP^0$.

$(ii)$ Let $T\in OP^\a$ and $T'\in OP^\beta$. Thus,
$T|D|^{-\a}$, $T'|D|^{-\beta}$ are in $OP^0$. By $(i)$ we get
$|D|^{\beta}T|D|^{-\a}|D|^{-\beta} \in OP^0$, so
$T'|D|^{-\beta}|D|^{\beta}T|D|^{-\beta-\a}\in OP^0$. Thus,
$T'T|D|^{-(\a+\beta)}\in OP^{0}$.

$(iii)$ For $T\in OP^\a$, $|D|^{\a-\beta}$ and
$T|D|^{-\a}$ are in $OP^0$, thus
$T|D|^{-\beta}=T|D|^{-\a}|D|^{\a-\beta} \in OP^0$.

$(iv)$ follows from $\delta(OP^0)\subseteq OP^0$.

$(v)$ Since $ \nabla(T)=\delta(T)|D| + |D|\delta(T)-[P_0,T]$,
the result follows from $(ii)$, $(iv)$ and the fact that $P_0$ is in
$OP^{-\infty}$.

\subsection{Proof of Lemma \ref{pdoalg}}

The non-trivial part of the proof is the stability under the product
of operators. Let $T,
T' \in \Psi(\A)$. There exist $d,d' \in \Z$ such that for any
$N\in\N$, $N> |d|+|d'|$,
there exist $P,P'$ in $\DD(\A)$, $p,p'\in \N_0$, $R\in OP^{-N-d'}$,
$R'\in OP^{-N-d}$ such
that $T=PD^{-2p}+R$, $T'=P'D^{-2p'}+R'$, $PD^{-2p}\in OP^d$ and
$P'D^{-2p'}\in
OP^{d'}$.

Thus, $TT'=PD^{-2p}P'D^{-2p'} + RP'D^{-2p'}+PD^{-2p}R' +RR'$.

We also have $RP'D^{-2p'}\in OP^{-N-d'+d'} = OP^{-N}$ and similarly,
$PD^{-2p}R'
\in OP^{-N}$. Since $RR'\in OP^{-2N}$, we get
$$
TT'\sim PD^{-2p}P'D^{-2p'} \mod OP^{-N}.
$$
If $p=0$, then $TT'\sim QD^{-2p'} \mod OP^{-N}$ where $Q=PP'\in
\DD(\A)$ and
$QD^{-2p'}\in OP^{d+d'}$. Suppose $p\neq 0$. A recurrence proves that
for any $q\in
\N_0$,
$$
D^{-2}P' \sim \sum_{k=0}^q (-1)^k \nabla^k(P') D^{-2k-2}
+(-1)^{q+1}D^{-2}
\nabla^{q+1}(P') D^{-2q-2} \mod OP^{-\infty}\, .
$$
By Lemma \ref{propOP} $(v)$, the remainder is in $OP^{d'+2p'-q-3}$,
since $P'\in
OP^{d'+2p'}$. Another recurrence gives for any $q\in \N_0$,
$$
D^{-2p}P'\sim \sum_{k_1,\cdots,k_p=0}^q (-1)^{|k|_1}
\nabla^{|k|_1}(P')
D^{-2|k|_1-2p} \, \mod  OP^{d'+2p' -q -1 - 2p}.
$$
Thus, with $q_N=N+d+d'-1$,
$$
TT'\sim \sum_{k_1,\cdots,k_p=0}^{q_N} (-1)^{|k|_1} P
\nabla^{|k|_1}(P')
D^{-2|k|_1-2(p+p')} \, \mod OP^{-N}.
$$
The last sum can be written $Q_N D^{-2r_N}$ where $r_N:=p\,
q_N+(p+p')$. Since $Q_N\in
\DD(\A)$ and $Q_N D^{-2r_N}\in OP^{d+d'}$, the result follows.

\subsection{Proof of Proposition \ref{tracenc} }
Let $P\in OP^{k_1}, \,Q\in OP^{k_2} \in \Psi(\A)$. With
$[Q,|D|^{-s}]= \big(Q-\sigma_{-s}(Q)\big) \,|D|^{-s}$ and the
equivalence
$Q-\sigma_{-s}(Q) \sim  -\sum_{r=1}^N g(-s,r) \, \eps^r(Q) \mod
OP^{-N-1+k_2}$,
we get
$$
P[Q,|D|^{-s}] \sim -\sum_{r=1}^N  g(-s,r)\,P \eps^r(Q) |D|^{-s} \mod
OP^{-N-1+k_1+k_2-\Re(s)}
$$
which gives, if we choose $N=n+k_1+k_2$,
$$
\underset{s=0}{\Res}\ \Tr\big( P
[Q,|D|^{-s}] \big)= -\sum_{r=1}^{n+k_1+k_2}\underset{s=0}{\Res}\
g(-s,r) \Tr \big( P \eps^r(Q)|D|^{-s} \big).
$$
By hypothesis $s\mapsto \Tr \big( P \eps^r(Q)|D|^{-s} \big)$ has only
simple
poles. Thus, since $s=0$ is a zero of the analytic function $s\mapsto
g(-s,r)$ for any
$r\geq 1$, we have $\underset{s=0}{\Res}\ g(-s,r) \,\Tr \big( P
 \eps^r(Q)|D|^{-s}\big) =0$, which entails that
$\underset{s=0}{\Res}\ \Tr\big(
P[Q,|D|^{-s}]\big)=0$ and thus
\begin{align*}
\ncint PQ= \underset{s=0}{\Res}\ \Tr\big( P|D|^{-s} Q\big)\, .
\end{align*}
When $s\in \C$ with $\Re(s)> 2 \max(k_1+n+1,k_2)$, the operator $P
|D|^{-s/2}$ is trace-class while $|D|^{-s/2} Q$ is bounded, so
$\Tr\big( P|D|^{-s} Q\big) = \Tr \big(|D|^{-s/2} QP|D|^{-s/2}\big) =
\Tr\big(\sigma_{-s/2}(Q P)|D|^{-s}\big)$.
Thus, using (\ref{one-par}) again,
$$
\underset{s=0}{\Res}\ \Tr\big( P |D|^{-s} Q\big)  = \ncint Q P
 +\sum_{r=1}^{n+k_1+k_2} \underset{s=0}{\Res}\ g(-s/2,r)
\Tr\big(\eps^r(Q P)
|D|^{-s}\big).
$$
As before, for any $r\geq 1$, $\underset{s=0}{\Res}\ g(-s/2,r)
\Tr\big(\eps^r(Q P) |D|^{-s}\big)=0$ since $g(0,r) =0$ and the
spectral triple
is simple. Finally,
\begin{align*}
\underset{s=0}{\Res}\ \Tr\big( P |D|^{-s} Q\big) = \ncint Q P.
\end{align*}

\vspace{1cm}

\section*{Acknowledgments}

\hspace{\parindent}
We thank Pierre Duclos, Emilio Elizalde, Victor Gayral,  Thomas
Krajewski, Sylvie Paycha, Joe Varilly, Dmitri Vassilevich and Antony
Wassermann
for helpful discussions and St\'ephane Louboutin for his help
with Proposition \ref{calculres}.

A. Sitarz would like to thank the CPT-Marseilles for its hospitality
and the Universit\'e de Provence for its financial support and
acknowledge the support of Alexander von Humboldt Foundation through
the Humboldt Fellowship.

\end{document}